\documentclass[aps,10pt,twocolumn,superscriptaddress,nofootinbib]{revtex4-2}

\newcommand{\externalizetikz}{0} % 0: remake, 1: export, 2: load
\usepackage[english]{babel}
\usepackage[utf8]{inputenc}
\usepackage{mymath}
\usepackage{amsthm}
\usepackage{mathdots}
\usepackage{nicefrac}
\usepackage{notation}
\usepackage[caption=false]{subfig}
\usepackage{subcaption}
\usepackage{ragged2e}
\DeclareCaptionJustification{justified}{\justifying}
\usepackage{tabularx,multirow,booktabs}
\newcolumntype{L}{>{\raggedright\arraybackslash}X}
\newcolumntype{C}{>{\centering\arraybackslash}X}
\newcolumntype{R}{>{\raggedleft\arraybackslash}X}
\usepackage{xcolor}
\definecolor{myblue}{rgb}{0,0.4470,0.7410}
\definecolor{myred}{rgb}{0.8500,0.3250,0.0980}
\definecolor{myorange}{rgb}{0.9290,0.6940,0.1250}
\definecolor{mypurple}{rgb}{0.4940,0.1840,0.5560}
\definecolor{mygreen}{rgb}{0.4660,0.6740,0.1880}
\definecolor{mylightblue}{rgb}{0.3010,0.7450,0.9330}
\definecolor{mydarkred}{rgb}{0.6350,0.0780,0.1840}
\definecolor{mygrey}{rgb}{0.6, 0.6, 0.6}
\colorlet{Deltacolor}{mygreen!60!white}
\colorlet{Gammacolor}{Deltacolor!40!white}
\colorlet{Rzcolor}{Gammacolor!40!white}
\colorlet{CLkcolor}{myorange!60!white}
\colorlet{ECcolor}{mydarkred!60!white}
\colorlet{DickeOnecolor}{myblue!40!white}
\colorlet{DickeOneDcolor}{DickeOnecolor}
\colorlet{DickeOneUcolor}{DickeOnecolor}
\colorlet{DickeOneUDcolor}{DickeOnecolor}
\colorlet{DickeTwokNcolor}{DickeOnecolor}
\colorlet{DickeTwokNUcolor}{DickeTwokNcolor}
\colorlet{DickeTwokNDcolor}{DickeTwokNcolor}
\colorlet{DickeTwokNUDcolor}{DickeTwokNcolor}
\definecolor{PRcolor}{rgb}{0.6,0.6,0.6}
\definecolor{PLcolor}{rgb}{0.6,0.6,0.6}
\definecolor{PRucolor}{rgb}{0.6,0.6,0.6}
\definecolor{PLucolor}{rgb}{0.6,0.6,0.6}
\definecolor{PRccolor}{rgb}{1,1,1}
\definecolor{PLccolor}{rgb}{1,1,1}
\definecolor{subPRcolor}{rgb}{0.85,0.85,0.85}
\definecolor{mutedgreen}{rgb}{0.6,0.85,0.6}
\ifnum\externalizetikz=2
  \usepackage{tikzexternal}
  \tikzsetexternalprefix{fig/}
\else
  \usepackage{tikz}
  \usepackage{pgfplots}
  \usetikzlibrary{calc}
  \usetikzlibrary{positioning}
  \usetikzlibrary{plotmarks}
  \pgfplotsset{
    compat=newest,
    table/header=false,
    tick label style={font=\footnotesize},
    label style={font=\small},
    legend style={font=\footnotesize},
    legend cell align=left,
    colormap={parula}{
      rgb255=(53,42,135)
      rgb255=(15,92,221)
      rgb255=(18,125,216)
      rgb255=(7,156,207)
      rgb255=(21,177,180)
      rgb255=(89,189,140)
      rgb255=(165,190,107)
      rgb255=(225,185,82)
      rgb255=(252,206,46)
      rgb255=(249,251,14)
    }
  }
  \pgfplotsset{
    myColOne/.style={myblue},
    myColTwo/.style={myred},
    myColThr/.style={myorange},
    myColFou/.style={mypurple},
    myColFiv/.style={mygreen},
    myColSix/.style={mylightblue},
    myColSev/.style={mydarkred}
  }
  \ifnum\externalizetikz=1
    \usepgfplotslibrary{external}
  \fi
\fi

\newcommand*\circled[2][2]{\tikz[baseline=(char.base)]{
\node[shape=circle,draw,inner sep=#1pt,thick] (char) {\textbf{#2}};}}
\usepackage[smallmeter]{myqcircuit}
\newcommand{\myqctmp}[2]{\Qcircuit @C=#2em @R=#1em @!R}
\makeatletter
\NewEnviron{myqcircuit}{\vcenter{\myqctmp{\myqcircuit@rowsep}{\myqcircuit@colsep}{\BODY}}}
\NewEnviron{myqcircuitr}[1]{\vcenter{\myqctmp{#1}{\myqcircuit@colsep}{\BODY}}}
\NewEnviron{myqcircuitc}[1]{\vcenter{\myqctmp{\myqcircuit@rowsep}{#1}{\BODY}}}
\NewEnviron{myqcircuit*}[2]{\vcenter{\myqctmp{#1}{#2}{\BODY}}}
\makeatother
\usepackage[color=black]{qcircuit2D}
\usepackage[colorlinks,urlcolor=blue,linkcolor=black,citecolor=blue]{hyperref}
\usepackage[capitalize,nameinlink]{cleveref}
\usepackage{orcidlink}
\newtheorem{assumption}{Assumption}
\newtheorem{lemma}{Lemma}
\newtheorem{theorem}{Theorem}
\newtheorem{corollary}{Corollary}
\newtheorem{definition}{Definition}
\crefname{assumption}{Assumption}{Assumptions}
\crefname{lemma}{Lemma}{Lemmas}
\crefname{theorem}{Theorem}{Theorems}
\crefname{corollary}{Corollary}{Corollaries}
\crefname{definition}{Definition}{Definitions}
\begin{document}

\title{\bf Practical block encodings of matrix polynomials that can also be trivially controlled}

\author{Martina~Nibbi\orcidlink{0009-0001-6440-0498}}
\email{martina.nibbi@tum.de}
%\thanks{Equal contribution}
\affiliation{These authors contributed equally to this work.}
\affiliation{Technical University of Munich, School of Computation, Information and Technology, Boltzmannstra{\ss}e 3, 85748 Garching, Germany}

\author{Filippo~Della~Chiara\orcidlink{0009-0001-7966-439X}}
\email{filippo.dellachiara@kuleuven.be}
%\thanks{Equal contributions}
\affiliation{These authors contributed equally to this work.}
\affiliation{Department of Computer Science,
KU Leuven, University of Leuven, 3001 Leuven, Belgium}

\author{Yizhi~Shen\orcidlink{0000-0002-4160-5482}} 
\email{yizhis@lbl.gov}
\affiliation{Applied Mathematics and Computational Research Division,
Lawrence Berkeley National Laboratory, Berkeley, CA 94720, USA}

\author{Aaron~Szasz\orcidlink{0000-0002-1127-2111}}
\email{aszasz@google.com}
\affiliation{Google Quantum AI, Santa Barbara, CA 93111, USA}

\author{Roel~Van~Beeumen\orcidlink{0000-0003-2276-1153}}
\email{rvanbeeumen@lbl.gov}
\affiliation{Department of Computer Science,
KU Leuven, University of Leuven, 3001 Leuven, Belgium}
\affiliation{Applied Mathematics and Computational Research Division,
Lawrence Berkeley National Laboratory, Berkeley, CA 94720, USA}
\date{\today}

\begin{abstract}
Quantum circuits naturally implement unitary operations on input quantum states. However, non-unitary operations can also be implemented through ``block encodings'', where additional ancilla qubits are introduced and later measured. 
While block encoding underlies a host of well-established theoretical applications, its implementation in practice has been prohibitively expensive for current quantum hardware.
In this paper, we present practical and explicit block encoding circuits implementing matrix polynomial transformations of a target matrix. 
With standard approaches, block-encoding a degree-$d$ matrix polynomial requires a circuit depth scaling as $d$ times the depth for block-encoding the original matrix alone.
By leveraging the recently introduced Fast One-Qubit Controlled Select Linear Combination of Unitaries (\foxlcu{}) framework, we show that the additional circuit-depth overhead required for encoding matrix polynomials can be reduced to scale linearly in $d$ with no dependence on system size or the cost of block encoding the original matrix. 
Moreover, we demonstrate that the \foxlcu{} circuits and their associated matrix polynomial transformations can be controlled with negligible overhead, enabling efficient applications such as Hadamard tests.
Finally, we provide explicit hardware-efficient circuits for representative non-unitary operators, namely quantum spin models, together with detailed non-asymptotic gate counts and circuit depths.
\end{abstract}
\maketitle

%%%%%%%%%%%%%%%%%%%%%%%%%%%%%
\section{Introduction}
%%%%%%%%%%%%%%%%%%%%%%%%%%%%%

Quantum circuits act unitarily on input states, but important applications of quantum computing to linear algebra and quantum simulation also require non-unitary operations. These operations can be implemented through block encodings \cite{Dong2021,Camps2022,Nguyen2022,Camps2024,Snderhauf2024,kuklinski2024,kuklinski2025,Nibbi_2024,Rullkotter2025,Liu2025,Childs2012,Childs2018,Gilyen2018,Babbush2018, Babbush_2019,sanders2020,Lee2021,Wan2021,Boyd2024,loaiza2024,Chakraborty2024,Kane2025,Georges2025,simon2025,Liu2025nbd,schillo2026,foqcs-lcu-arxiv}. 
Given a matrix $\Ham$, a block encoding embeds $\Ham$ as a sub-block of a larger unitary acting on an extended Hilbert space.
This construction enables the implementation of $\Ham$ using unitary quantum circuits, with post-selection on ancilla qubits restricting to the desired non-unitary matrix block.

Although well understood theoretically, block encoding has been viewed as impractical on present-day quantum hardware, as the standard routines require large multi-controlled gates and circuits that are non-local on two-dimensional qubit chip architectures. The recently introduced Fast One-Qubit Controlled Select Linear Combination of Unitaries (\foxlcu{}) framework \cite{foqcs-lcu-arxiv} overcomes the gate-complexity bottleneck for a wide class of matrices $\Ham$, replacing multi-controlled gates with just two parallel layers of two-qubit gates. 
In exchange, we must use a number of ancilla qubits that scales linearly with the system size.
%For many important Hamiltonians, such as one-dimensional Ising and Heisenberg models, the \foxlcu{} block-encoding yields a short circuit depth that is linear in the number of spins.
For many important matrices, such as the one-dimensional quantum Ising and Heisenberg Hamiltonians, the \foxlcu{} block-encoding yields a short circuit depth that is linear in the number of spins.

Even with the barriers to block encoding removed for large classes of non-unitary matrices, there are further challenges to the practical application of block encodings in useful algorithms. The first key challenge is the controlled application of a block-encoded non-unitary operation. 
Such controlled block encodings are required, for example, in Hadamard-test circuits. 
The problem is that, even if a block encoding for $\Ham$ can be efficiently decomposed into native one- and two-qubit gates, once the encoding is controlled, we would naively end up with a large number of Toffoli gates. Such circuits would again be impractical on present-day devices. In this paper, we show that \foxlcu{} block-encodings for a large class of non-unitaries $\Ham$ can be controlled with negligible overhead, requiring just two additional \cnot{} gates compared with the cost of the non-controlled block encoding of $\Ham$.

The second key challenge is the need to block-encode not just $\Ham$, but polynomials of $\Ham$. This arises, for example, in using block encoding to implement time evolution under a quantum spin Hamiltonian, where we approximate the time-evolution operator $\text{exp}(-i\Ham t)$ with a suitable polynomial expansion in $\Ham$~\cite{Berry_2015,Childs2018,Meister_2022,Sze_2025,Kirby2023exactefficient}. For common techniques such as the Quantum Singular Value Transformation (QSVT) \cite{Low_2019,Gilyen2018,Martyn2021,Kikuchi2023,Motlagh2024,Meister_2022,Sze_2025,Sze:2025vjh}, the circuit depth to block-encode a degree-$d$ polynomial in $\Ham$ would scale as $d$ times the depth to encode $\Ham$. 
In this paper, we generalize the \foxlcu{} framework and show that the extra circuit depth when encoding a polynomial in $\Ham$ relative to encoding just $\Ham$ scales only with $d$; that is, the extra depth for the polynomial is independent of the complexity of the original block encoding of $\Ham$ and therefore also independent of system size. Even more striking, the entire block-encoding for the polynomial can additionally be controlled by adding just 4 \cnot{} gates, independent of system size and of the polynomial degree.

Furthermore, we show that our block encodings can be efficiently implemented on quantum chips with realistic two-dimensional layouts using only local (nearest-neighbor and single-qubit) gates.  
We first provide explicit circuit constructions in terms of single- and two-qubit gates assuming all-to-all connectivity, from which we derive exact, non-asymptotic \cnot{} gate counts and circuit depths.  We then provide a detailed mapping of those circuits to a two-dimensional square grid architecture, with the final circuit written only in terms of local gates in $2$D: single-qubit gates and nearest-neighbor two-qubit gates. 
The mapping from all-to-all to 2D incurs only a small overhead---the increase in circuit depth is just a few layers, independent of the system size.

To summarize, for important classes of non-unitary operators $\Ham$, including 1D quantum spin models, this paper resolves many of the bottlenecks that have made block encoding impractical until now:
\begin{enumerate}
\item[(1)] The block encoding of $\Ham$ can be controlled with negligible overhead. 
\item[(2)] Block encodings of matrix polynomials in $\Ham$ are nearly as efficient as block-encodings of $\Ham$ alone.
\item[(3)] Block encodings of matrix polynomials in $\Ham$ can be controlled with negligible overhead.
\item[(4)] All of these block encodings can be implemented in low-depth circuits using only nearest-neighbor two-qubit gates on a two-dimensional square lattice.
\end{enumerate}
By introducing explicit and practical circuits, our framework turns block encodings and controlled block encodings from theoretical abstractions into concrete modules that can be compiled, benchmarked, and stress-tested.

%%%%%%%%%%%%%%%%%%%%%%%%%%%%%%%%%%%
\paragraph*{\bf Road map.}
%%%%%%%%%%%%%%%%%%%%%%%%%%%%%%%%%%%
The paper is organized as follows. \cref{sec:genral_control_theory} discusses the simplification of a generic controlled multi-qubit gate when it acts on special reference states. 
This is a key ingredient for our main results on controlled block encodings.
In \cref{sec:foqcs_lcu}, we examine block encodings via the \foxlcu{} formalism, originally conceptualized in \cite{foqcs-lcu-arxiv}, and further develop the construction and implementation of its controlled versions. In \cref{sec:transformations_general}, we proceed to show how to realize a matrix polynomial in a generic operator $\Ham$ given its \foxlcu{} block encoding and discuss how to implement a controlled version of the polynomial matrix transformations with negligible overhead. 
In \cref{sec:applications}, we outline explicit circuits and potential applications of \foxlcu{}-enabled block encodings of matrix polynomials, in both uncontrolled and controlled forms. 
%More specifically, in \cref{sec:spin_hamiltonians} we use the one-dimensional Heisenberg XYZ Hamiltonian as a representative example and establish exact \cnot{} counts and \cnot{} depths on a two-dimensional architecture with nearest-neighbor connectivity. \Cref{sec:hadamard_tests,sec:time_evolution} highlight possible applications of the simplifications introduced in our work. 
We conclude and summarize
in \cref{sec:conclusions}.

%%%%%%%%%%%%%%%%%%%%%%%%%%%%%%%%%%%%%%%%%%%%%%%%%%%%%
\section{Controlled unitaries and  common eigenstates}
\label{sec:genral_control_theory}
%%%%%%%%%%%%%%%%%%%%%%%%%%%%%%%%%%%%%%%%%%%%%%%%%%%
Controlled unitaries are essential and ubiquitous building blocks in many quantum algorithms. In this section, we consider unitaries with a decomposition of the form
\begin{equation}
 \genGate = \Bgate_s \cdot \Agate_s \cdots \Bgate_1 \cdot \Agate_1, 
 \label{eq:Udecomp}
\end{equation}
where we assume that the unitaries $\{\Agate_i\}_{i=1}^s$ can be implemented using very few elementary gates, and the unitaries $\{\Bgate_i\}_{i=1}^s$ share a common eigenstate $\ket{\xi}$ corresponding to eigenvalue $1$:
\begin{align}
 \Bgate_i \ket{\xi} &= \ket{\xi}, & \forall \ i &= 1,\dots,s.
 \label{eq:common-eigvec}
\end{align}
Next we prove that a controlled version of such unitaries $\genGate$, whenever \emph{applied to} or \emph{measured in} the state $\ket{\xi}$, can be efficiently implemented with only negligible overhead compared to the implementation of $\genGate$ itself.

\subsection{Applying controlled unitaries to a specific state}

Controlling a unitary of the form \cref{eq:Udecomp} leads to the following lemma, which formalizes how such controlled implementations can be efficiently realized when applied to the specific common eigenstate $\ket{\xi}$.

%%% LEMMA %%%
\begin{lemma}\label{lemma:unitary_decomposition_initial_state}
Let $\genGate$ be defined as in \cref{eq:Udecomp,eq:common-eigvec}, with $\ket{\xi}$ a common eigenstate of $\{\Bgate_i\}_{i=1}^s$. Then, the controlled application of $\genGate$ to $\ket{\xi}$ can be simplified as follows:
\begin{equation}
 \quad%
 \begin{myqcircuit}
 & \qw & \ctrl{1} & \qw \\
 \lstick{\ket{\xi}\!\!} & {/\strut^{n}}\qw & \gate{\genGate} & \qw
 \end{myqcircuit}
 \ = \quad\ \,
 \begin{myqcircuit}
 & \qw & \ctrl{1} & \qw & \qw &&\cdots &&&\ctrl{1} &\qw & \qw \\
 \lstick{\ket{\xi}\!\!} & {/\strut^{n}}\qw & \gate{\Agate_1} & \gate{\Bgate_1} & \qw &&\cdots &&&\gate{\Agate_s} &\gate{\Bgate_s} & \qw
 \end{myqcircuit}
 \label{eq:gengate_controlled_decomposition}
\end{equation}
\end{lemma}
\begin{proof}
Let the control qubit be prepared in an arbitrary superposition $\alpha\ket{0} + \beta\ket{1}$. Then the left-hand side of \cref{eq:gengate_controlled_decomposition} yields the following joint state:
\begin{equation}
 \label{eq:unitary_decomposition_initial_state}
 \left(\alpha\ket{0}+\beta\ket{1}\right)\ket{\xi} \longrightarrow
 \alpha\ket{0}\ket{\xi} + \beta\ket{1}\genGate\ket{\xi},
\end{equation}
while the right-hand side yields the state:
\begin{equation}\label{eq:proof_lemma1_0}
 \begin{split}
 \left(\alpha\ket{0}+\beta\ket{1}\right)\ket{\xi} &\longrightarrow \\
 \alpha\ket{0}\prod_{i=1}^s &\Bgate_i\ket{\xi} + \beta\ket{1}\prod_{i=1}^s \Bgate_i\Agate_i\ket{\xi}.
 \end{split}
\end{equation}
Next, using the eigenstate properties in \cref{eq:common-eigvec} and the decomposition in \cref{eq:Udecomp}, we obtain:
\begin{align}
 \alpha\ket{0}\prod_{i=1}^s \Bgate_i\ket{\xi} &= \alpha\ket{0}\ket{\xi}, \label{eq:proof_lemma1_1}\\
 \beta\ket{1}\prod_{i=1}^s \Bgate_i\Agate_i\ket{\xi} &= \beta\ket{1}\genGate\ket{\xi}. \label{eq:proof_lemma1_2}
\end{align}
Inserting \cref{eq:proof_lemma1_1,eq:proof_lemma1_2} into \cref{eq:proof_lemma1_0} yields \cref{eq:unitary_decomposition_initial_state}, completing the proof.
\end{proof}

As a consequence, it suffices to only control the unitaries $\Agate_i$. Under the assumption that the implementation costs of $\Agate_i$ are negligible relative to those of $\Bgate_i$, the controlled unitary acting on $\ket{\xi}$ has a cost comparable to its non-controlled implementation.

\subsection{Measuring controlled unitaries in a specific state}

From \cref{lemma:unitary_decomposition_initial_state}, we can derive similar simplifications for the measurement of a controlled unitary.
These are summarized in the following lemma and corollary, where $ \text{P}_\xi$ is defined to be a unitary that prepares the common eigenstate of $\{\Bgate_i\}_{i=1}^s$, $\ket{\xi}$:
\begin{equation}
 \ket{\xi} = \text{P}_\xi \ket{0}_n.
 \label{eq:Pxi}
\end{equation}

\begin{lemma}\label{lemma:unitary_decomposition_measurement_2}
Let $\genGate$ be defined as in \cref{eq:Udecomp,eq:common-eigvec,eq:Pxi}.
Then, the controlled application of $\genGate$ to an arbitrary $n$-qubit state $\ket{\varphi}$, followed by a measurement post-selected on the outcome $\ket{\xi}$, can be simplified as follows:
\begin{equation}
 \begin{split}
	&{\begin{myqcircuitr}{-0.2}
 & \qw & \ctrl{1} & \qw & \qw \\
 \lstick{\ket{\varphi}\!\!}& {/\strut^{n}}\qw & \gate{\genGate} & \gate{\text{P}_\xi^\dagger} & \meter{} & \cw & 0
	\end{myqcircuitr}} \ \ \raisebox{1ex}{\upshape=} \\[1ex]
	&\qquad{\begin{myqcircuitr}{-0.1}
 & \qw & \ctrl{1} & \qw & \qw &&\cdots &&& \ctrl{1} & \qw & \qw & \qw \\
 \lstick{\ket{\varphi}\!\!}& {/\strut^{n}}\qw & \gate{\Agate_1} & \gate{\Bgate_1} & \qw &&\cdots &&&\gate{\Agate_s} &\gate{\Bgate_s} & \gate{\text{P}_\xi^\dagger} & \meter{} & \cw & 0
	\end{myqcircuitr}\hspace{-1.5em}}
 \end{split}
 \label{eq:gengatedag_controlled_decomposition_2}
\end{equation}
\end{lemma}
\begin{proof}
Starting from a generic control-qubit state $\alpha\ket{0}+\beta\ket{1}$, the left-hand side of \cref{eq:gengatedag_controlled_decomposition_2} produces the following joint state before the measurement:
\begin{equation}
\left(\alpha\ket{0}+\beta\ket{1}\right)\ket{\varphi} \longrightarrow
 \alpha\ket{0}\text{P}^\dagger_{\xi}\ket{\varphi} + \beta\ket{1}\text{P}^\dagger_{\xi}\genGate\ket{\varphi}.
\end{equation}
Consequently, when the bottom $n$ qubits are measured in the state $\ket{0}$, we obtain:
\begin{equation}
 \label{eq:unitary_decomposition_measurement_2}
 \alpha\ket{0} \braket{\xi\vert \varphi}+\beta\ket{1} \bra{\xi}\genGate\ket{\varphi}.
\end{equation}

On the other hand, the right-hand side of \cref{eq:gengatedag_controlled_decomposition_2} produces, before measurement, the joint state:
\begin{equation}
 \begin{split}
 \left(\alpha\ket{0}+\beta\ket{1}\right)\ket{\varphi} &\longrightarrow \\
 \alpha\ket{0}\text{P}^\dagger_{\xi}\prod_{i=1}^s &\Bgate_i\ket{\varphi} + \beta\ket{1}\text{P}^\dagger_{\xi}\prod_{i=1}^s \Bgate_i\Agate_i\ket{\varphi},\hspace{-1ex}
 \end{split}
\end{equation}
and after measuring the bottom $n$ qubits in the state $\ket{0}$, we obtain:
\begin{equation} \label{eq:proof_lemma2_0}
 \alpha\ket{0}\bra{\xi}\prod_{i=1}^s \Bgate_i\ket{\varphi} + \beta\ket{1}\bra{\xi}\prod_{i=1}^s \Bgate_i\Agate_i\ket{\varphi}.
\end{equation}
Next, we note that since $\Bgate_i$ are unitary, their left eigenvectors are the conjugate transposes of their right eigenvectors. Thus, \cref{eq:common-eigvec} is equivalent to:
\begin{align}
 \bra{\xi}\Bgate_i &= \bra{\xi}, & \forall \ i &= 1,\dots,s.
 \label{eq:common-eigvec-left}
\end{align}
Finally, using these properties and the decomposition in \cref{eq:Udecomp}, we obtain:
\begin{align}
 \alpha\ket{0}\bra{\xi}\prod_{i=1}^s \Bgate_i\ket{\varphi} &= \alpha\ket{0} \braket{\xi\vert \varphi}, \label{eq:proof_lemma2_1}\\
 \beta\ket{1}\bra{\xi}\prod_{i=1}^s \Bgate_i\Agate_i\ket{\varphi} &= \beta\ket{1} \bra{\xi}\genGate\ket{\varphi}. \label{eq:proof_lemma2_2}
\end{align}
Inserting \cref{eq:proof_lemma2_1,eq:proof_lemma2_2} into \cref{eq:proof_lemma2_0} yields \cref{eq:unitary_decomposition_measurement_2}, completing the proof.
\end{proof}

\begin{corollary}\label{lemma:unitary_decomposition_measurement}
Let $\genGate$ be defined as in \cref{eq:Udecomp,eq:common-eigvec,eq:Pxi}.
Then, the controlled application of $\genGate^\dagger$ to an arbitrary $n$-qubit state $\ket{\varphi}$, followed by a measurement post-selected on the outcome $\ket{\xi}$, can be simplified as follows:
\begin{equation}
 \begin{split}
	\quad&{\begin{myqcircuitr}{-0.2}
 & \qw & \ctrl{1} & \qw & \qw \\
 \lstick{\ket{\varphi}\!\!}& {/\strut^{n}}\qw & \gate{\genGate^\dagger} & \gate{\text{P}_\xi^\dagger} & \meter{} & \cw & 0
	\end{myqcircuitr}} \ \ \raisebox{1ex}{\upshape=} \\[1ex]
	&\qquad{\begin{myqcircuitr}{-0.1}
 & \qw & \qw & \ctrl{1} & \qw &&\cdots &&& \qw & \ctrl{1} & \qw & \qw \\
 \lstick{\ket{\varphi}\!\!}& {/\strut^{n}}\qw & \gate{\Bgate_s^\dagger} & \gate{\Agate_s^\dagger} & \qw &&\cdots &&&\gate{\Bgate_1^\dagger} &\gate{\Agate_1^\dagger} & \gate{\text{P}_\xi^\dagger} & \meter{} & \cw & 0
	\end{myqcircuitr}\hspace{-1em}}
 \end{split}
 \label{eq:gengatedag_controlled_decomposition}
\end{equation}
\end{corollary}
\begin{proof}
 We use \cref{lemma:unitary_decomposition_measurement_2}, with $\genGate^\dagger$ as the target unitary, which in particular admits the following decomposition:
 \begin{equation}
 \genGate^\dagger = \Agate_1^\dagger \cdot \Bgate_1^\dagger \cdots \Agate_s^\dagger \cdot \Bgate_s^\dagger
 \end{equation}
 instead of the decomposition of $\genGate$ given in \cref{eq:Udecomp}.
 Note that since $\Bgate_i$ are unitary, they share the same eigenstates with their adjoints $\Bgate_i^\dagger$, including the state $\ket{\xi}$.
\end{proof}

In the following sections, we show how these decompositions lead to substantial simplifications in several of the \foxlcu{}'s applications.

%%%%%%%%%%%%%%%%%%%%%%%%%%%%%%%%%%%
\section{The \foxlcu{} block encoding}\label{sec:foqcs_lcu}
%%%%%%%%%%%%%%%%%%%%%%%%%%%%%%%%%%%
We start by reviewing the Fast One-Qubit Controlled Select Linear Combination of Unitaries (\foxlcu{}) block encoding proposed in \cite{foqcs-lcu-arxiv}.
In contrast to standard LCU, 
% formalism, 
which expresses the target matrix as a linear combination of black-box unitaries $\Umat_\m$ \cite{Childs2012, Childs2018,Gilyen2018,sanders2020,Chakraborty2024,Babbush_2019,Kane2025,Babbush2018,Lee2021,loaiza2024,Wan2021,Boyd2024,Chakraborty2024,Georges2025,simon2025,Liu2025nbd,foqcs-lcu-arxiv,Loaiza_2023, Sze_2025}:
\begin{align}\label{eq:def_lcu}
 \Ham &= \sum_{\m=0}^{\M-1}\pscoeff_\m \Umat_\m, &
 \pscoeff_\m &\in \mathbb{C},
\end{align}
\foxlcu{} capitalizes on a structured operator basis. Specifically, without loss of generality and in line with quantum-simulation applications, the unitaries $\Umat_\m$ are chosen to be Pauli strings. Accordingly, we can rewrite \cref{eq:def_lcu} by formally separating the $\Xp$ and $\Zp$ components of each Pauli string \cite{Georges2025,foqcs-lcu-arxiv,Kirby2023exactefficient}:
\begin{equation}\label{eq:lcu_general_check_matrix}
 \Ham = \sum_{\ix=0}^{2^{\n}-1} \sum_{\jz=0}^{2^{\n}-1} \newpscoeff_{\ix\jz} \bigotimes_{\elle=0}^{{\n-1}}\Zp^{\jzbit{\elle}}\Xp^{\ixbit{\elle}},
\end{equation}
where the new coefficients $\newpscoeff_{\ix\jz}$ preserve the sparsity of the original $\alpha_\m$ and absorb an additional factor of the imaginary unit whenever $\ixbit{\elle}=\jzbit{\elle}=1$.

%%%%%%%%%%%%%%%%%%%%%%%%%%%%%%
\subsection{High-level \foxlcu{} circuit}
\label{sec:foqcs_lcu_circuit}
%%%%%%%%%%%%%%%%%%%%%%%%%%%%%%
A fully general \foxlcu{} block encoding circuit is given by \cite{foqcs-lcu-arxiv}:
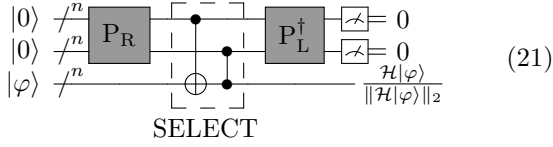
\begin{equation}\label{eq:circuit_foqcs_lcu}
\begin{myqcircuit}
\lstick{\ket{0}} & {/\strut^{n}}\qw &\multicgate{1}{\PR}{PRcolor} & \qw &\ctrl{2} & \qw & \qw & \multicgate{1}{\PLdag}{PLcolor} & \meter{} & \cw & 0 \\
\lstick{\ket{0}} & {/\strut^{n}}\qw &\ghost{\PR} & \qw & \qw & \ctrl{1} & \qw & \ghost{\PL} & \meter{} & \cw & 0 \\
\lstick{\ket{\varphi}} & {/\strut^{n}}\qw & \qw & \qw & \targ{} & \control\qw & \qw & \qw & \qw & & \frac{\Ham\ket{\varphi}}{\normtwo{\Ham\ket{\varphi}}} \gategroup{1}{5}{3}{6}{1em}{--} \\
& & && \hspace*{1.75ex}\raisebox{-3ex}{\select{}} & & 
\end{myqcircuit}
\vspace{1ex}
\end{equation}
where the central routine, known as the \select{} oracle, consists of two parallel layers of $n$ \cnot{} and $n$ \cz{} gates, respectively:
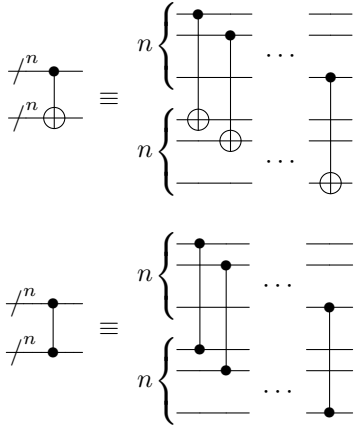
\begin{align*}
 {\begin{myqcircuitr}{0.95}
 &{/\strut^{n}}\qw & \ctrl{1} & \qw \\
 &{/\strut^{n}}\qw & \targ{} & \qw
 \end{myqcircuitr}}
 \ &\equiv \quad\,\
 {\begin{myqcircuit*}{-0.1}{0.4}
 & \ctrl{5} & \qw & \ghost{\cdots} & \qw & \qw \\
 & \qw & \ctrl{5} & \ghost{\cdots} & \qw & \qw \\
 &&& \cdots & & \\
 & \qw & \qw & \ghost{\cdots} & \ctrl{5} & \qw \\
 &&&&& \\
 & \targ{} & \qw & \ghost{\cdots} & \qw & \qw \\
 & \qw & \targ{} & \ghost{\cdots} & \qw & \qw \\
 &&& \cdots & & \\
 & \qw & \qw & \ghost{\cdots} & \targ{} & \qw
 \inputgroupv{1}{4}{0.8em}{1.2em}{n\hspace*{-1em}}
 \inputgroupv{6}{9}{0.8em}{1.2em}{n\hspace*{-1em}}
 \end{myqcircuit*}} \\[1em]
 {\begin{myqcircuit*}{1.6}{0.33}
 & \qw &{/\strut^{n}}\qw & \qw & \qw & \ctrl{1} & \qw & \qw & \qw \\
 & \qw &{/\strut^{n}}\qw & \qw & \qw & \ctrl{0} & \qw & \qw & \qw
 \end{myqcircuit*}}
 \ &\equiv \quad\,\
 {\begin{myqcircuit*}{-0.1}{0.75}
 & \ctrl{5} & \qw & \ghost{\cdots} & \qw & \qw \\
 & \qw & \ctrl{5} & \ghost{\cdots} & \qw & \qw \\
 &&& \cdots & & \\
 & \qw & \qw & \ghost{\cdots} & \ctrl{5} & \qw \\
 &&&&& \\
 & \ctrl{0} & \qw & \ghost{\cdots} & \qw & \qw \\
 & \qw & \ctrl{0} & \ghost{\cdots} & \qw & \qw \\
 &&& \cdots & & \\
 & \qw & \qw & \ghost{\cdots} & \ctrl{0} & \qw
 \inputgroupv{1}{4}{0.8em}{1.2em}{n\hspace*{-1em}}
 \inputgroupv{6}{9}{0.8em}{1.2em}{n\hspace*{-1em}}
 \end{myqcircuit*}}
\end{align*}
In this paper, we employ a slightly different convention for the state-preparation oracles, \PR{} and \PL{}, acting on the $2n$ ancilla qubits, compared to that in~\cite{foqcs-lcu-arxiv}. More precisely,
\begin{align} 
 \PR\ket{0}_{2n} &= \frac{1}{\sqrt{\normfact}} \sum_{\ix=0}^{2^{\n}-1} \sum_{\jz=0}^{2^{\n}-1} e^{i\arg(\newpscoeff_{\ix\jz})}\sqrt{\abs{\newpscoeff_{\ix\jz}}} \ket{\ix}\ket{\jz}, \label{eq:pr_def}\\
 \PL\ket{0}_{2n} &= \frac{1}{\sqrt{\normfact}} \sum_{\ix=0}^{2^{\n}-1} \sum_{\jz=0}^{2^{\n}-1} {\sqrt{\abs{\newpscoeff_{\ix\jz}}}} \ket{\ix}\ket{\jz}, \label{eq:pl_def}
\end{align}
where phase information of the coefficients $\newpscoeff_{ij}$ is encoded solely in \PR{} and with an overall normalization factor: 
\begin{equation}\label{eq:def_foxlcu_norm}
 \normfact = \sum_{i=0}^{2^n-1}\sum_{j=0}^{2^n-1} \abs{\newpscoeff_{ij}} .
\end{equation}
Since the success probability of a block encoding, i.e., the probability that each ancilla qubit is measured in $\ket0$, is inversely proportional to the square of the normalization, one might be concerned about whether the \foxlcu{} normalization $\normfact$ is larger than for the standard LCU. In fact, for the standard LCU (\cref{eq:def_lcu}) with $U_m$ taken to be Pauli strings, the corresponding normalization factor is precisely $\normfact$. 
%We remark that \foxlcu{} shares the same normalization factor $\normfact$ as the standard LCU. 
Indeed, the coefficients $\newpscoeff_{\ix\jz}$ possess the same sparsity pattern as $\pscoeff_m$ and differ only by phase factors.

A naive implementation of \PR{} and \PL{} amounts to preparing an arbitrary $2n$-qubit state, which in general entails an exponential overhead.
However, by making use of the preparation of Dicke states~\cite{Dicke_1954,Bartschi_2019, Bartschi_2022, Piroli2024, yu2024, Farrell2025}, we can obtain efficient implementations for a variety of Hamiltonians, including the one-dimensional Heisenberg and spin glass models~\cite{foqcs-lcu-arxiv}.

%%%%%%%%%%%%%%%%%%%%%%%%%%%%%%%%%%%
\subsection{Controlling \foxlcu{}}
\label{sec:control_foqcs_lcu}
%%%%%%%%%%%%%%%%%%%%%%%%%%%%%%%%%%%

Now we apply the theory in \cref{sec:genral_control_theory} to \foxlcu{}. In particular, we demonstrate how to control the entire block-encoding circuit by exploiting its factorization into \PR{}, \select{}, and \PLdag{}.
For the circuit in \cref{eq:circuit_foqcs_lcu}, it suffices to control only the \PR{} and \PLdag{} oracles, as formalized in the following theorem.

\begin{theorem}\label{theo:select_no_control}
 Let {\upshape$\Ufox$} be the unitary circuit implementing the \foxlcu{} block encoding defined in \cref{eq:circuit_foqcs_lcu}.
 If the $2n$ ancilla qubits are prepared 
 %and post-selected
 in the state $\ket{0}_{2n}$, then the controlled application of {\upshape$\Ufox$} can be simplified for any initial system state $\ket{\varphi}$ by controlling only the {\upshape\PR{}} and {\upshape\PLdag{}} oracles:
 {\upshape\begin{equation}\label{eq:circuit_control_foqcs_lcu}
 \quad
 \begin{myqcircuit}
 &\qw & \ctrl{1} & \qw \\
 \lstick{\ket{0}} &{/\strut^{n}}\qw &\multigate{2}{\Ufox} & \qw \\
 \lstick{\ket{0}} & {/\strut^{n}}\qw &\ghost{\Ufox} & \qw \\
 \lstick{\ket{\varphi}} & {/\strut^{n}}\qw & \ghost{\Ufox} & \qw
 \end{myqcircuit}
 \; = \qquad
 \begin{myqcircuit}
 & \qw & \ctrl{1} & \qw & \qw & \ctrl{1} & \qw & \\
 \lstick{\ket{0}} & {/\strut^{n}}\qw &\multicgate{1}{\PR}{PRcolor} & \ctrl{2} & \qw & \multicgate{1}{\PLdag}{PLcolor} & \qw & \\
 \lstick{\ket{0}} & {/\strut^{n}}\qw &\ghost{\PR} & \qw & \ctrl{1} & \ghost{\PL} & \qw & \\
 \lstick{\ket{\varphi}} & {/\strut^{n}}\qw & \qw & \targ{} & \control\qw & \qw & \qw & \\
 \end{myqcircuit}\hspace*{-1.25em}
 \end{equation}}
\end{theorem}
\begin{proof}
 We first note that, by construction, the \select{} oracle in \foxlcu{} acts trivially on states of the form $\ket{0}_{2n}\ket{\varphi}$, where $\ket{\varphi}$ is an arbitrary $n$-qubit state:
\begin{align}\label{eq:control_select_trivial}
 \select \ket{0}_{2n}\ket{\varphi} &= \ket{0}_{2n}\ket{\varphi}, & \forall &\ket{\varphi}.
\end{align}
That is, $\ket{0}_{2n}\ket{\varphi}$ is an eigenstate of the \select{} unitary with eigenvalue 1.

Applying \cref{lemma:unitary_decomposition_initial_state}, with $s = 2$, to the decomposition defined by: 
% of the \foxlcu{} circuit:
\begin{align}
 \Agate_1 &\equiv \PR\otimes\Ip_n, & \Bgate_1 &\equiv \select, \nonumber \\
 \Agate_2 &\equiv \PLdag \otimes \Ip_n, & \Bgate_2 &\equiv \Ip_{3n},
\end{align}
proves \cref{eq:circuit_control_foqcs_lcu}.
\end{proof}

We next turn to the controlled \PR{} and \PLdag{} oracles. By additionally applying \cref{lemma:unitary_decomposition_initial_state} to these state-preparation oracles, we derive efficient controlled implementations under the assumption stated below, which is satisfied by all \PR{} and \PL{} oracles in both \cite{foqcs-lcu-arxiv} and \cref{sec:applications}.

%%% ASSUMPTION %%%
\begin{assumption}
\label{ass:PR-PL}
The oracles {\upshape\PR{}} and {\upshape\PL{}} admit unitary decompositions of the form:
{\upshape\begin{align}
 \PR &\equiv \PRmod\cdot\PRc, & % \label{eq:pr_ab}\\
 \PL &\equiv \PLmod\cdot\PLc, \label{eq:pl_ab}
\end{align}}%
where {\rm \PRmod} 
% {\upshape\PRu{}} 
and {\rm \PLmod}
% {\upshape\PLu{}}
have the all-zero state $\ket{0}_{2n}$ as an eigenstate corresponding to the eigenvalue $1$:
{\upshape\begin{align}
 \PRmod\ket{0}_{2n} &= \ket{0}_{2n}, & % \label{eq:b_eigenstate}\\
 \PLmod\ket{0}_{2n} &= \ket{0}_{2n}, \label{eq:bp_eigenstate}
\end{align}}%
and {\upshape\PRc{}} and {\upshape\PLc{}} consist of only $\bigO{1}$ single-qubit gates.
\end{assumption}

Note that the unitary decompositions in \cref{ass:PR-PL} are a special case, specifically with $s=1$ and $\ket{\xi} = \ket{0}_{2n}$, of \cref{lemma:unitary_decomposition_initial_state}:
\begin{equation}
 {\begin{myqcircuitc}{0.33}
 & \qw & \qw & \qw & \qw & \qw & \ctrl{1} & \qw & \qw \\
 \lstick{\ket{0}} & \qw & \qw & {/\strut^{2n}}\qw & \qw & \qw & \cgate{\PR}{PRcolor} & \qw & \qw
 \end{myqcircuitc}}
 \ \raisebox{0.5ex}= \qquad
 {\begin{myqcircuit*}{0.1}{0.33}
 & \qw & \qw & \qw & \qw & \qw & \ctrl{1} & \qw & \qw & \qw & \qw \\
 \lstick{\ket{0}} & \qw & \qw & {/\strut^{2n}}\qw & \qw & \qw & \gate{\PRc} & \qw & \cgate{\PRmod}{PRucolor} & \qw & \qw
 \end{myqcircuit*}} \label{eq:pr_control_uncontrol}
 \end{equation}
 and \cref{lemma:unitary_decomposition_measurement}:
\begin{equation}
 {\begin{myqcircuit*}{0}{0.33}
 & \qw & \qw & \qw & \qw & \qw & \ctrl{1} & \qw & \qw \\
 & \qw & \qw & {/\strut^{2n}}\qw & \qw & \qw & \cgate{\PLdag}{PLcolor} & \qw & \meter{} & \cw & \cw & & 0
 \end{myqcircuit*}}
 \ \ \raisebox{1ex}= \,\
 {\begin{myqcircuit*}{-0.1}{0.33}
 & \qw & \qw & \qw & \qw & \qw & \qw & \qw & \ctrl{1} & \qw & \qw \\
 & \qw & \qw & {/\strut^{2n}}\qw & \qw & \qw & \cgate{\PLdagmod}{PLucolor} & \qw & \gate{\PLcdag} & \qw & \meter{} & \cw & \cw & & 0
 \end{myqcircuit*}} \label{eq:pldag_control_uncontrol}
\end{equation}

Assuming \cref{ass:PR-PL} holds, we obtain the following corollary as a direct consequence of \cref{theo:select_no_control}.

\begin{corollary}\label{corol:control_foqcs}
 Suppose {\upshape\PR{}} and {\upshape\PL{}} satisfy \cref{ass:PR-PL}. Then the circuit in \cref{eq:circuit_control_foqcs_lcu} can be simplified as follows:
 {\upshape\begin{equation}\label{eq:control_foxlcu_decomposed}
 \qquad\begin{myqcircuit}
 &\qw &\ctrl{1} & \qw & \qw & \qw & \qw &\ctrl{1} & \qw\\
 \lstick{\ket{0}} & {/\strut^{n}}\qw & \multigate{1}{\PRc} & \multicgate{1}{\PRmod}{PRucolor} & \ctrl{2} & \qw & \multicgate{1}{\PLdagmod{}}{PLucolor} & \multigate{1}{\PLcdag} & \qw & \\
 \lstick{\ket{0}} & {/\strut^{n}}\qw & \ghost{\PRc} & \ghost{\PRmod} & \qw &\ctrl{1} & \ghost{\PLdagmod} & \ghost{\PLcdag} & \qw & \\
 \lstick{\ket{\varphi}} & {/\strut^{n}}\qw & \qw & \qw & \targ{} & \ctrl{0} & \qw & \qw & \qw 
 \end{myqcircuit} 
 \end{equation}}
\end{corollary}
\begin{proof}
 The proof mirrors that of \cref{theo:select_no_control}.  Under \cref{ass:PR-PL}, we split $\PR$ into two parts: $\PRc$, which is composed of $\bigO{1}$ single-qubit gates, and $\PRmod$, which has $\ket{0}_{2n}$ as an eigenstate.
 Likewise, $\PL$ can be decomposed into two gates, namely $\PLc$ and $\PLmod$, with the same properties as $\PRc$ and $\PRmod$, respectively.
 Then, as a consequence of \cref{eq:control_select_trivial}, the composite unitary $\PRmod,\select,\PLdagmod$ leaves the state $\ket{0}_{2n}\ket{\phi}$ invariant for any $n$-qubit state $\ket{\phi}$:
 \begin{equation}
 \PLdagmod \cdot \select \cdot \PRmod \ket{0}_{2n}\ket{\phi} = \ket{0}_{2n}\ket{\phi}.
 \end{equation}
 Applying \cref{lemma:unitary_decomposition_initial_state}, with $s = 2$, to the decomposition defined by:
 \begin{align}
 \Agate_1 &\equiv \PRc\otimes\Ip_n, & \Bgate_1 &\equiv \PLdagmod \cdot \select \cdot \PRmod, \nonumber\\
 \Agate_2 &\equiv \PLcdag \otimes \Ip_n, & \Bgate_2 &\equiv \Ip_{3n},
 \end{align}
 proves \cref{eq:control_foxlcu_decomposed}.
\end{proof}

\renewcommand{\PRu}[0]{\ensuremath{\PRmod}}
\renewcommand{\PRudag}[0]{\ensuremath{\PRmod^\dagger}}
\newcommand{\PLu}[0]{\ensuremath{\PLmod}}
\newcommand{\PLudag}[0]{\ensuremath{\PLmod^\dagger}}

%%%%%%%%%%%%%%%%%%%%%%%%%%%%%%%%%%%
\section{Matrix polynomials using \foxlcu{}} \label{sec:transformations_general}
%%%%%%%%%%%%%%%%%%%%%%%%%%%%%%%%%%%

In this section, we describe how to explicitly implement matrix polynomials of a matrix $\Ham$ within the \foxlcu{} block-encoding framework. We call the resulting construction the \polyfoqcs{} block encoding.
We begin in \cref{sec:product of fox matrices} by constructing products of \foxlcu{} block encodings, which enables efficient implementation of
powers of $\Ham$. In \cref{sec:poly_general}, we generalize this approach to realize an arbitrary matrix polynomial in $\Ham$.
Finally, in \cref{sec:control_simp_poly}, we show how the simplifications introduced in \cref{sec:genral_control_theory} can be leveraged to reduce the overhead required to control the entire circuit implementing a matrix polynomial in $\Ham$.

%%%%%%%%%%%%%%%%%%%%%%%%%%%%%%%%%%%%%%%%%%%%%%%%%%%%%%%%%%%
\subsection{Block encodings for products of matrices}\label{sec:product of fox matrices}
%%%%%%%%%%%%%%%%%%%%%%%%%%%%%%%%%%%%%%%%%%%%%%%%%%%%%%%%%%%
Here, we lift the \foxlcu{} formalism by introducing a circuit for the product of block-encoded matrices.
We start by considering the product of two \foxlcu{} block encodings, which can be implemented as described in the following theorem.
\begin{theorem} \label{lemma:product}
Let $\Mat{1}$ and $\Mat{2}$ be two $n$-qubit matrices written as a linear combination of Pauli strings:
{\upshape\begin{align}
 \Mat{1} &= \sum_{i,j=0}^{2^n-1}\alpha_{ij}\bigotimes_{\ell= 0}^{n-1}Z^{\texttt{j}_\ell}X^{\texttt{i}_\ell},\\
 \Mat{2} &= \sum_{p,q=0}^{2^n-1}\beta_{pq}\bigotimes_{\ell= 0}^{n-1}Z^{\texttt{q}_\ell}X^{\texttt{p}_\ell}.
\end{align}}%
Then the product matrix $C = \Mat{2} \cdot \Mat{1}$ can be written as:
{\upshape\begin{equation}\label{product of 2 be}
 C = \sum_{i,j,p,q= 0}^{2^n-1}\alpha_{ij}\beta_{pq}\bigotimes_{\ell= 0}^{n-1}Z^{\texttt{q}_\ell}X^{\texttt{p}_\ell}Z^{\texttt{j}_\ell}X^{\texttt{i}_\ell}, 
\end{equation}}%
and a block encoding of $C$ can be realized by the circuit:
{\upshape\begin{equation}\label{circuit: product of 2}
\hspace*{2ex}\begin{myqcircuit}
\lstick{\ket{0}} & {/\strut^{n}}\qw &\multicgate{1}{\PRindex{\Mat{1}}}{PRcolor} & \ctrl{4} & \qw & \qw & \qw & \multicgate{1}{\PLindexdag{\Mat{1}}}{PLcolor} & \meter{} & \cw & 0 \\
\lstick{\ket{0}} & {/\strut^{n}}\qw &\ghost{\PRindex{\Mat{1}}} & \qw & \ctrl{3} & \qw & \qw & \ghost{\PLindexdag{\Mat{1}}} & \meter{} & \cw & 0 \\
\lstick{\ket{0}} & {/\strut^{n}}\qw &\multicgate{1}{\PRindex{\Mat{2}}}{PRcolor} & \qw & \qw & \ctrl{2} & \qw & \multicgate{1}{\PLindexdag{\Mat{2}}}{PLcolor} & \meter{} & \cw & 0 \\
\lstick{\ket{0}} & {/\strut^{n}}\qw &\ghost{\PRindex{\Mat{2}}} & \qw & \qw & \qw & \ctrl{1} & \ghost{\PLindexdag{\Mat{2}}} & \meter{} & \cw & 0 \\
\lstick{\ket{\varphi}} & {/\strut^{n}}\qw & \qw & \targ{} & \ctrl{0} & \targ{} & \ctrl{0} & \qw & \qw & & \frac{C\ket{\varphi}}{\norm{C\ket{\varphi}}_2} \\
\end{myqcircuit}
\end{equation}}%
where {\upshape$\PRindex{\Mat{1}}$}, {\upshape$\PLindex{\Mat{1}}^{\dagger}$}, {\upshape$\PRindex{\Mat{2}}$}, and {\upshape$\PLindex{\Mat{2}}^{\dagger}$} denote the state-preparation oracles defined in \cref{eq:pr_def,eq:pl_def}, associated with the \foxlcu{} block encodings of $\Mat{1}$ and $\Mat{2}$, respectively.
\end{theorem}
\begin{proof}
 See \cref{app: proof lemma}.
\end{proof}

This theorem readily generalizes to products of more than two \foxlcu{} block-encoded matrices, as detailed in the following corollary.
\begin{corollary} \label{corollary: products}
 Let $\{M_s\}_{s= 1}^{k}$ be a set of matrices written as linear combination of Pauli strings: 
 {\upshape\begin{equation}
 M_s = \sum_{i_s,j_s=0}^{2^n-1}\alpha_{i_s j_s}\bigotimes_{\ell= 0}^{n-1}Z^{\texttt{j}_{\ell,s}}X^{\texttt{i}_{\ell,s}}.
 \end{equation}}%
Then the product matrix $C = M_k \cdots \Mat{2} \cdot \Mat{1}$ can be written as follows:
{\upshape\begin{equation}
 \begin{split}
 C = &\sum_{\{i_s=0,j_s=0\}_{s=1}^k}^{2^n-1} \alpha_{i_1 j_1}\alpha_{i_2j_2}\cdots \alpha_{i_k j_k} \\
 &\qquad\qquad\bigotimes_{\ell=0}^{n-1}Z^{\texttt{j}_{\ell,k}}X^{\texttt{i}_{\ell,k}}\cdots Z^{\texttt{j}_{\ell,2}} X^{\texttt{i}_{\ell,2}} Z^{\texttt{j}_{\ell,1}}X^{\texttt{i}_{\ell,1}},\hspace{-1em}
 \end{split}
\end{equation}}%
and a block encoding of $C$ can be realized by the circuit:
{\upshape
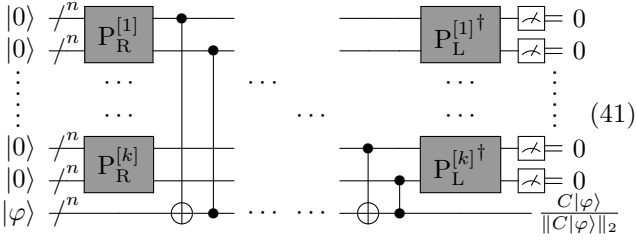
\begin{equation}\label{circuit: product of k}
\begin{myqcircuit}
\lstick{\ket{0}} & {/\strut^{n}}\qw &\multicgate{1}{\PR^{[1]}}{PRcolor} & \ctrl{6} & \qw & \ghost{\cdots\hspace*{0.5em}\cdots} & \qw & \qw &\multicgate{1}{{\PL^{[1]}}^\dagger}{PLcolor} & \meter{} & \cw & 0 \\
\lstick{\ket{0}} & {/\strut^{n}}\qw &\ghost{\PR^{[1]}} & \qw & \ctrl{5} & \ghost{\cdots\hspace*{0.5em}\cdots} & \qw&\qw & \ghost{{\PL^{[1]}}^\dagger} & \meter{} & \cw & 0 \\
%&& &&&&&&&& && \\
 \lstick{\raisebox{1.5ex}{$\vdots$}\hspace*{1.25ex}} && \cdots &&&\cdots\hspace*{0.5em}\phantom{\cdots} &&& \cdots & & \raisebox{1.5ex}{$\vdots$}\hspace*{1.25ex} & \\
 \lstick{\raisebox{1.5ex}{$\vdots$}\hspace*{1.25ex}} && \cdots &&& \phantom{\cdots}\hspace*{0.5em}\cdots &&& \cdots & & \raisebox{1.5ex}{$\vdots$}\hspace*{1.25ex} & \\ 
%&& &&&&&&&& && \\vdot\
\lstick{\ket{0}} & {/\strut^{n}}\qw &\multicgate{1}{\PR^{[k]}}{PRcolor} & \qw & \qw & \ghost{\cdots\hspace*{0.5em}\cdots} & \ctrl{2} & \qw & \multicgate{1}{{\PL^{[k]}}^\dagger}{PLcolor} & \meter{} & \cw & 0 \\
\lstick{\ket{0}} & {/\strut^{n}}\qw &\ghost{\PR^{[k]}} & \qw & \qw& \ghost{\cdots\hspace*{0.5em}\cdots} &\qw& \ctrl{1} & \ghost{{\PL^{[k]}}^\dagger} & \meter{} & \cw & 0 \\
\lstick{\ket{\varphi}} & {/\strut^{n}}\qw & \qw & \targ{} & \ctrl{0} & \push{\hspace*{0.5em}\cdots\hspace*{0.5em}\cdots\hspace*{0.5em}}\qw & \targ{} & \ctrl{0} & \qw & \qw & & \frac{C\ket{\varphi}}{\norm{C\ket{\varphi}}_2} \\
\end{myqcircuit}\hspace{-2em}
\end{equation}}%
where {\upshape$\PR^{[s]}$} and {\upshape${\PL^{[s]}}^\dagger$} denote the state-preparation oracles defined in \cref{eq:pr_def,eq:pl_def}, associated with the \foxlcu{} block encodings of $\{M_s\}_{s=1}^{k}$.
\end{corollary}
\begin{proof}
The proof is a straightforward generalization of \cref{lemma:product} and differs only in the inclusion of additional index sets.
\end{proof}

Finally, we show how \cref{corollary: products} applies in the specific case of block encodings of powers of $\Ham$.

\begin{corollary}\label{corol:powers}
 Let $\Ham$ be a matrix expressed as a linear combination of Pauli strings, as in \cref{eq:lcu_general_check_matrix}. Then, a block encoding of $\Ham^k$ can be realized by the circuit:
 {\upshape
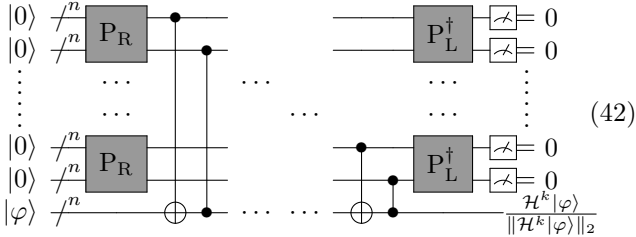
\begin{equation}\label{circuit: product of k}
 \begin{myqcircuit}
 \lstick{\ket{0}} & {/\strut^{n}}\qw &\multicgate{1}{\PR}{PRcolor} & \ctrl{6} & \qw & \ghost{\cdots\hspace*{0.5em}\cdots} & \qw & \qw &\multicgate{1}{\PLdag}{PLcolor} & \meter{} & \cw & 0 \\
 \lstick{\ket{0}} & {/\strut^{n}}\qw &\ghost{\PR} & \qw & \ctrl{5} & \ghost{\cdots\hspace*{0.5em}\cdots} & \qw&\qw & \ghost{\PLdag} & \meter{} & \cw & 0 \\
 %&& &&&&&&&& && \\
 \lstick{\raisebox{1.5ex}{$\vdots$}\hspace*{1.25ex}} && \cdots &&&\cdots\hspace*{0.5em}\phantom{\cdots} &&& \cdots & & \raisebox{1.5ex}{$\vdots$}\hspace*{1.25ex} & \\
 \lstick{\raisebox{1.5ex}{$\vdots$}\hspace*{1.25ex}} && \cdots &&& \phantom{\cdots}\hspace*{0.5em}\cdots &&& \cdots & & \raisebox{1.5ex}{$\vdots$}\hspace*{1.25ex} & \\ 
 %&& &&&&&&&& && \\vdot\
 \lstick{\ket{0}} & {/\strut^{n}}\qw &\multicgate{1}{\PR}{PRcolor} & \qw & \qw & \ghost{\cdots\hspace*{0.5em}\cdots} & \ctrl{2} & \qw & \multicgate{1}{\PLdag}{PLcolor} & \meter{} & \cw & 0 \\
 \lstick{\ket{0}} & {/\strut^{n}}\qw &\ghost{\PR} & \qw & \qw& \ghost{\cdots\hspace*{0.5em}\cdots} &\qw& \ctrl{1} & \ghost{\PLdag} & \meter{} & \cw & 0 \\
 \lstick{\ket{\varphi}} & {/\strut^{n}}\qw & \qw & \targ{} & \ctrl{0} & \push{\hspace*{0.5em}\cdots\hspace*{0.5em}\cdots\hspace*{0.5em}}\qw & \targ{} & \ctrl{0} & \qw & \qw & & \frac{\Ham^k\ket{\varphi}}{\norm{\Ham^k\ket{\varphi}}_2} \\
 \end{myqcircuit}\hspace{-1em}
 \end{equation}}%
 where {\upshape\PR{}} and {\upshape\PL{}} are defined in \cref{eq:pr_def,eq:pl_def} and are performed, in parallel, $k$ times.
\end{corollary}
\begin{proof}
 This corollary can be viewed as a special case of \cref{corollary: products}, where $M_s \equiv \Ham$, for $s = 1,\ldots,k$.
\end{proof}

%%%%%%%%%%%%%%%%%%%%%%%%%%%%%%%%%%%
\subsection{Block encodings for matrix polynomials}\label{sec:poly_general}
%%%%%%%%%%%%%%%%%%%%%%%%%%%%%%%%%%%

Given block encodings of the matrix powers $\Ham^k$, we now build a \foxlcu{}-type block encoding for their linear combination. 
Concretely, we consider a degree-$d$ matrix polynomial in $\Ham$:
\begin{equation}\label{eq:poly_transform_H}
 p_d(\Ham ) = a_0 I + a_1 \Ham + a_2 \Ham ^2 + \cdots + a_d \Ham ^d,
\end{equation}
with coefficients~$a_k \inC$ for $k = 0,1,\ldots,d$.

To realize an efficient \polyfoqcs{} block encoding for $p_d(\Ham)$, we embed the block encodings of $\Ham^k$ within a second LCU layer.
This \emph{outer} LCU procedure requires state-preparation oracles \polyR{} and \polyL{} to encode the coefficients $a_k$, as well as a \select{} oracle to activate the corresponding powers of $\Ham$. Crucially, instead of using the standard approach to implement the outer LCU, which would require $\lceil\log_2(d)\rceil$ ancilla qubits, we load the coefficients $a_k$ through a \emph{unary encoding} \cite{Babbush2018,Sze:2025vjh} on $d$ ancilla qubits.
As a consequence, we again avoid multi-controlled unitaries in the \select{} operation for $\Ham^k$.

We first introduce the weighted coefficients associated with the matrix polynomial coefficients $a_k$:
\begin{equation}
 \label{eq:POLY-wk}
 w_k := \sqrt{\abs{a_k}\,\normfact^{k}},
\end{equation}
where $\normfact$ denotes the \foxlcu{} normalization factor defined in \cref{eq:def_foxlcu_norm}.
Next, for $k = 0,\ldots,d-1$, we define the relative phase differences
\begin{equation}
 \label{eq:POLY-phik}
 \phi_k := \arg(a_{k+1}) - \arg(a_k),
\end{equation}
and introduce the angles
\begin{equation}
 \label{eq:POLY-thetak}
 \theta_k := 2\,\arccos\!\left(
 \frac{\abs{w_k}}{\sqrt{1 - \sum_{p=0}^{k-1} \abs{w_p}^2}}
 \right).
\end{equation}
We then define the outer state-preparation oracles, \polyR{} and \polyL{}, respectively, as follows.

%%% DEFINITION %%%
\begin{definition}
\label{def:POLY}
Let $a_k \inC$, for $k = 0,1,\ldots,d$, be the coefficients of the matrix polynomial in \cref{eq:poly_transform_H}, and $\phi_k$ and $w_k$ the associated set of real parameters defined in \cref{eq:POLY-phik,eq:POLY-wk}.
Then, we can define the $d$-qubit outer state-preparation oracles\footnote{When $k=0$, we take $\sum_{j=0}^{-1}\phi_j= 0$.}:
{\upshape\begin{align}
\polyR \ket{0}_d &= \frac{1}{\sqrt{\polynormfact}} \sum_{k=0}^{d} e^{i\sum_{j= 0}^{k-1}\phi_j}
w_k \ket{k_u}, \label{eq:def_polyr}
\\ 
\polyL \ket{0}_d &= \frac{1}{\sqrt{\polynormfact}} \sum_{k=0}^{d} w_k \ket{k_u}, \label{eq:def_polyl}
\end{align}}%
where $\polynormfact$ is a normalization factor:
\begin{equation}\label{eq:def_W}
 \polynormfact = \sum_{k=0}^{d} w_k^2 = \sum_{k=0}^{d} \vert a_k \vert \normfact^k, 
\end{equation}
and $\ket{k_u}$ the unary encoding of the index $k$, i.e.:
\begin{equation}\label{eq:unary-ku}
\begin{aligned}
 \ket{0_u} &= \ket{000\cdots0}, \\
 \ket{1_u} &= \ket{100\cdots0}, \\
 \ket{2_u} &= \ket{110\cdots0}, \\[-0.75ex]
 &~\vdots \\
 \ket{d_u} &= \ket{111\cdots1}.
\end{aligned}
\end{equation}
\end{definition}

We also define the following $d$-qubit circuits that will serve as essential subroutines for realizing the two state-preparation oracles \polyR{} and \polyL{}:
\begin{align}
\label{eq:circuit_CRY_ladder}
\hspace*{-1ex}{\begin{myqcircuit}
 & {/\strut^{d}}\qw & \gate{\cry(\bftheta_1^{d-1})} & \qw
\end{myqcircuit}}
&\equiv
\scalebox{0.9}{$\begin{myqcircuitr}{0}
 & \ctrl{1} & \qw & \qw & \qw & \qw \\
 & \gate{R_y(\theta_1)} & \ctrl{1} & \qw & \qw & \qw \\
 & \qw & \gate{R_y(\theta_2)} & \qw & \qw & \qw \\
 & & & \raisebox{1ex}{$\ddots$} & \\
 & \qw & \qw & \qw & \ctrl{1} & \qw \\
 & \qw & \qw & \qw & \gate{R_y(\theta_{d-1})} & \qw \\
\end{myqcircuitr}$}\hspace*{-1.5ex} \\[1em]
\label{eq:circuit_parallel_P}
{\begin{myqcircuit}
 & {/\strut^{d}}\qw & \gate{\text{P}(\bfphi_0^{d-1})} & \qw
\end{myqcircuit}}
&\equiv
\scalebox{0.9}{$\begin{myqcircuit}
 & \gate{P(\phi_0)} & \qw \\
 & \gate{P(\phi_1)} & \qw \\
 & \raisebox{0.5em}{$\vdots$} & \\
 & \gate{P(\phi_{d-1})} & \qw \\
\end{myqcircuit}$}
\end{align}
where rotation and phase gates follow the conventions:
\begin{align}\label{eq:rydef}
    &R_y(\theta) = \begin{pmatrix}
        \cos(\frac{\theta}{2})& -\sin(\frac{\theta}{2}) \\[3pt]
        \sin(\frac{\theta}{2})& \cos(\frac{\theta}{2})
    \end{pmatrix}, 
    &P(\phi) = \begin{pmatrix}
        1 & 0\\
        0 & e^{i\phi}
    \end{pmatrix}.
\end{align}
Using \cref{eq:circuit_CRY_ladder,eq:circuit_parallel_P}, we can efficiently implement the \polyR{} oracle to generate a weighted superposition over unary-encoded indices, as originally proposed in \cite{Berry_2015} and summarized in the following lemma.

%%% LEMMA %%%
\begin{lemma}
\label{lem:POLYR}
The state-preparation oracle {\upshape\polyR{}}, defined in \cref{eq:def_polyr}, can be realized by the $d$-qubit circuit:
{\upshape\begin{equation}\label{eq:circuit_poly_r}
\begin{myqcircuit}
 & {/\strut^{d}}\qw & \gate{\polyR} & \qw
\end{myqcircuit}
=
\begin{myqcircuitr}{-0.2}
 & \gate{R_y(\theta_0)} & \multigate{3}{\cry(\bftheta_1^{d-1})} & \multigate{3}{\text{P}(\bfphi_0^{d-1})} & \qw \\
 & \qw & \ghost{\cry(\bftheta_1^{d-1})} & \ghost{\text{P}(\bfphi_1^{d-1})} & \qw \\
 & \vdots & & & \\
 & \qw & \ghost{\cry(\bftheta_1^{d-1})} & \ghost{\text{P}(\bfphi_0^{d-1})} & \qw \\
\end{myqcircuitr}
\end{equation}}%
where the rotation angles $\theta_k$
and the phase angles $\phi_k$ are given in \cref{eq:POLY-thetak} and \cref{eq:POLY-phik}, respectively.
\end{lemma}

Similarly, we can efficiently implement the \polyL{} oracle, as stated in the following lemma.
\begin{lemma}
\label{lem:POLYL}
The state-preparation oracle {\upshape\polyL{}}, defined in \cref{eq:def_polyl}, can be realized by the $d$-qubit circuit:
{\upshape\begin{equation}\label{eq:circuit_poly_l}
\begin{myqcircuit}
 & {/\strut^{d}}\qw & \gate{\polyL} & \qw
\end{myqcircuit}
\ = \
\begin{myqcircuitr}{-0.2}
 & \gate{R_y(\theta_0)} & \multigate{3}{\cry(\bftheta_1^{d-1})} & \qw \\
 & \qw & \ghost{\cry(\bftheta_1^{d-1})} & \qw \\
 & \vdots & & \\
 & \qw & \ghost{\cry(\bftheta_1^{d-1})} & \qw \\
\end{myqcircuitr}
\end{equation}}%
where the rotation angles $\theta_k$ are given in \cref{eq:POLY-thetak}.
\end{lemma}

We remark that the \polyR{} and \polyL{} circuits given in \cref{lem:POLYR,lem:POLYL}, respectively, each require only $2(d-1)$ CNOT gates following \cref{eq:circuit_CRY_ladder}.

Finally, by combining the outer state-preparation oracles \polyR{} and \polyL{}, defined in \cref{def:POLY} and implemented as described in \cref{lem:POLYR,lem:POLYL}, together with the constructions from \cref{sec:product of fox matrices}, we obtain an efficient block encoding of the matrix polynomial $p_d(\Ham)$, as formalized in the following theorem.

%%% THEOREM %%%
\begin{theorem}
\label{th:matrix-poly}
Let $p_d(\Ham)$ be a degree-$d$ matrix polynomial in an $n$-qubit matrix $\Ham$, as defined within \cref{eq:poly_transform_H}.
Then, the $(d + (2d+1)n)$-qubit circuit shown in \cref{fig:poly_general_pr} implements a block encoding of $p_d(\Ham)$.
\end{theorem}
\begin{proof}
See \cref{app:proof_matrix_polynomial}.
\end{proof}

%%% FIGURE %%%
\begin{figure*}
\subfloat[General implementation\label{fig:poly_general_pr}]{%
\scalebox{0.93}{$
\begin{myqcircuit}
%
% POLY
\lstick{\ket{0}} &\multigate{3}{\polyR} & \ctrl{4} & \qw & \ghost{\cdots} & \qw & \qw & \qw & \qw & \qw & \ghost{\cdots\hspace*{0.5em}\cdots} & \qw & \qw & \qw & \ghost{\cdots} & \qw & \ctrl{4} &\multigate{3}{\polyL^{\dagger}} & \meter{} & \cw & 0\\
\lstick{\ket{0}} & \ghost{\polyR} & \qw & \ctrl{5} & \ghost{\cdots} & \qw & \qw & \qw & \qw & \qw & \ghost{\cdots\hspace*{0.5em}\cdots} & \qw & \qw & \qw & \ghost{\cdots} & \ctrl{5} & \qw &\ghost{\polyL^{\dagger}} & \meter{} & \cw & 0\\
\lstick{\raisebox{1.5ex}{$\vdots$}\hspace*{1.25ex}} & & & & \cdots & & & & & & & & & & \cdots & & & & & \raisebox{1.5ex}{$\vdots$}\hspace*{1.25ex} & \\
\lstick{\ket{0}} & \ghost{\polyR} & \qw & \qw & \ghost{\cdots} & \ctrl{7} & \qw & \qw & \qw & \qw & \ghost{\cdots\hspace*{0.5em}\cdots} & \qw & \qw & \ctrl{7} & \ghost{\cdots} & \qw & \qw &\ghost{\polyL^{\dagger}} & \meter{} & \cw & 0\\
%
% PR1
\lstick{\ket{0}} & {/\strut^{n}}\qw &\multicgate{1}{\PR{}}{PRcolor} & \qw & \ghost{\cdots} & \qw & \ctrl{8}& \qw & \qw & \qw & \ghost{\cdots\hspace*{0.5em}\cdots} & \qw & \qw & \qw & \ghost{\cdots} & \qw & \multicgate{1}{\PLdag{}}{PLcolor} & \qw & \meter{} & \cw & 0\\
\lstick{\ket{0}} & {/\strut^{n}}\qw &\ghost{\PR{}} & \qw & \ghost{\cdots} & \qw & \qw &\ctrl{7}& \qw & \qw & \ghost{\cdots\hspace*{0.5em}\cdots} & \qw & \qw & \qw & \ghost{\cdots} & \qw & \ghost{\PLdag{}} & \qw & \meter{} & \cw & 0\\
% PR2
\lstick{\ket{0}} & {/\strut^{n}}\qw & \qw &\multicgate{1}{\PR{}}{PRcolor} & \ghost{\cdots} & \qw & \qw & \qw &\ctrl{6}& \qw & \ghost{\cdots\hspace*{0.5em}\cdots} & \qw & \qw & \qw & \ghost{\cdots} & \multicgate{1}{\PLdag{}}{PLcolor} & \qw & \qw & \meter{} & \cw & 0\\
\lstick{\ket{0}} & {/\strut^{n}}\qw & \qw & \ghost{\PR{}} & \ghost{\cdots} & \qw & \qw & \qw & \qw &\ctrl{5}& \ghost{\cdots\hspace*{0.5em}\cdots} & \qw & \qw & \qw & \ghost{\cdots} & \ghost{\PLdag{}} & \qw & \qw & \meter{} & \cw & 0\\
\lstick{\raisebox{1.5ex}{$\vdots$}\hspace*{1.25ex}} & & & & \cdots & & & & & & \cdots\hspace*{0.5em}\phantom{\cdots} & & & & \cdots & & & & & \raisebox{1.5ex}{$\vdots$}\hspace*{1.25ex} & \\
\lstick{\raisebox{1.5ex}{$\vdots$}\hspace*{1.25ex}} & & & & \cdots & & & & & & \phantom{\cdots}\hspace*{0.5em}\cdots & & & & \cdots & & & & & \raisebox{1.5ex}{$\vdots$}\hspace*{1.25ex} & \\
% PRd
\lstick{\ket{0}} & {/\strut^{n}}\qw & \qw & \qw & \ghost{\cdots} & \multicgate{1}{\PR{}}{PRcolor}& \qw & \qw & \qw & \qw & \ghost{\cdots\hspace*{0.5em}\cdots} &\ctrl{2}& \qw &\multicgate{1}{\PLdag{}}{PLcolor} & \ghost{\cdots} & \qw & \qw & \qw & \meter{} & \cw & 0\\
\lstick{\ket{0}} & {/\strut^{n}}\qw & \qw & \qw & \ghost{\cdots} & \ghost{\PR{}}& \qw & \qw & \qw & \qw & \ghost{\cdots\hspace*{0.5em}\cdots} & \qw &\ctrl{1} & \ghost{\PLdag{}} & \ghost{\cdots} & \qw & \qw & \qw & \meter{} & \cw & 0\\
%
% data qubits
\lstick{\ket{\varphi}} & {/\strut^{n}}\qw & \qw & \qw & \ghost{\cdots} & \qw &\targ{}&\ctrl{0}&\targ{}&\ctrl{0}& \push{\hspace*{0.5em}\cdots\hspace*{0.5em}\cdots\hspace*{0.5em}}\qw &\targ{}&\ctrl{0} & \qw & \ghost{\cdots} & \qw & \qw & \qw & \qw & & & \frac{p_d(\Ham)\ket{\varphi}}{\norm{p_d(\Ham)\ket{\varphi}}}
\end{myqcircuit}
$}
} \\
\subfloat[Efficient controlled implementation of \PR{} and \PLdag{} oracles under \cref{ass:PR-PL}\label{fig:poly_pr_ab}]{%
\scalebox{0.93}{$
\begin{myqcircuit}
%
% POLY
\lstick{\ket{0}} &\multigate{3}{\polyR} & \ctrl{4} & \qw & \ghost{\cdots} & \qw &\qw& \qw & \qw &\qw & \qw & \ghost{\cdots\hspace*{0.5em}\cdots} & \qw& \qw & \qw & \qw & \ghost{\cdots} & \qw & \ctrl{4} &\multigate{3}{\polyL^{\dagger}} & \meter{} & \cw & 0\\
\lstick{\ket{0}} & \ghost{\polyR} & \qw & \ctrl{5} & \ghost{\cdots} & \qw &\qw& \qw & \qw & \qw &\qw & \ghost{\cdots\hspace*{0.5em}\cdots} & \qw& \qw & \qw & \qw & \ghost{\cdots} & \ctrl{5} & \qw &\ghost{\polyL^{\dagger}} & \meter{} & \cw & 0\\
\lstick{\raisebox{1.5ex}{$\vdots$}\hspace*{1.25ex}} & & & & \cdots & & & & & & & & & & & & \cdots & & & & & \raisebox{1.5ex}{$\vdots$}\hspace*{1.25ex} & \\
\lstick{\ket{0}} & \ghost{\polyR} & \qw & \qw & \ghost{\cdots} & \ctrl{7} &\qw& \qw & \qw & \qw &\qw & \ghost{\cdots\hspace*{0.5em}\cdots} & \qw & \qw& \qw & \ctrl{7}& \ghost{\cdots} & \qw & \qw &\ghost{\polyL^{\dagger}} & \meter{} & \cw & 0\\
%
% PR1
\lstick{\ket{0}} & {/\strut^{n}}\qw &\multicgate{1}{\PRc{}}{PRccolor} & \qw & \ghost{\cdots} & \qw &\multicgate{1}{\PRu{}}{PRucolor} &\ctrl{8}& \qw & \qw & \qw & \ghost{\cdots\hspace*{0.5em}\cdots} & \qw & \qw &\multicgate{1}{\PLudag}{PLucolor}& \qw & \ghost{\cdots} & \qw & \multicgate{1}{\PLcdag{}}{PLccolor} & \qw & \meter{} & \cw & 0\\
\lstick{\ket{0}} & {/\strut^{n}}\qw &\ghost{\PRc{}} & \qw & \ghost{\cdots} & \qw &\ghost{\PRu{}}& \qw &\ctrl{7}& \qw & \qw & \ghost{\cdots\hspace*{0.5em}\cdots} & \qw & \qw &\ghost{\PLudag} & \qw & \ghost{\cdots} & \qw & \ghost{\PLcdag{}} & \qw & \meter{} & \cw & 0\\
% PR2
\lstick{\ket{0}} & {/\strut^{n}}\qw & \qw &\multicgate{1}{\PRc{}}{PRccolor} & \ghost{\cdots} & \qw &\multicgate{1}{\PRu{}}{PRucolor}& \qw & \qw &\ctrl{6}& \qw & \ghost{\cdots\hspace*{0.5em}\cdots} & \qw & \qw &\multicgate{1}{\PLudag}{PLucolor}& \qw & \ghost{\cdots} & \multicgate{1}{\PLcdag{}}{PLccolor} & \qw & \qw & \meter{} & \cw & 0\\
\lstick{\ket{0}} & {/\strut^{n}}\qw & \qw & \ghost{\PRc{}} & \ghost{\cdots} & \qw &\ghost{\PRu{}}& \qw & \qw & \qw &\ctrl{5}& \ghost{\cdots\hspace*{0.5em}\cdots} & \qw & \qw &\ghost{\PLudag{}}& \qw & \ghost{\cdots} & \ghost{\PLcdag{}} & \qw & \qw & \meter{} & \cw & 0\\
\lstick{\raisebox{1.5ex}{$\vdots$}\hspace*{1.25ex}} & & & & \cdots & & & & & & & \cdots\hspace*{0.5em}\phantom{\cdots} & & & & & \cdots & & & & & \raisebox{1.5ex}{$\vdots$}\hspace*{1.25ex} & \\
\lstick{\raisebox{1.5ex}{$\vdots$}\hspace*{1.25ex}} & & & & \cdots & & & & & & & \phantom{\cdots}\hspace*{0.5em}\cdots & & & & & \cdots & & & & & \raisebox{1.5ex}{$\vdots$}\hspace*{1.25ex} & \\
% PRd
\lstick{\ket{0}} & {/\strut^{n}}\qw & \qw & \qw & \ghost{\cdots} & \multicgate{1}{\PRc{}}{PRccolor}&\multicgate{1}{\PRu{}}{PRucolor}& \qw & \qw & \qw & \qw & \ghost{\cdots\hspace*{0.5em}\cdots} &\ctrl{2}& \qw&\multicgate{1}{\PLudag{}}{PLucolor} &\multicgate{1}{\PLcdag{}}{PLccolor} & \ghost{\cdots} & \qw & \qw & \qw & \meter{} & \cw & 0\\
\lstick{\ket{0}} & {/\strut^{n}}\qw & \qw & \qw & \ghost{\cdots} & \ghost{\PRc{}}& \ghost{\PRu{}}&\qw & \qw & \qw & \qw & \ghost{\cdots\hspace*{0.5em}\cdots} & \qw &\ctrl{1}&\ghost{\PLudag{}} & \ghost{\PLcdag{}} & \ghost{\cdots} & \qw & \qw & \qw & \meter{} & \cw & 0\\
%
% data qubits
\lstick{\ket{\varphi}} & {/\strut^{n}}\qw & \qw & \qw & \ghost{\cdots} & \qw &\qw&\targ{}&\ctrl{0}&\targ{}&\ctrl{0}& \push{\hspace*{0.5em}\cdots\hspace*{0.5em}\cdots\hspace*{0.5em}}\qw &\targ{}&\ctrl{0} &\qw & \qw & \ghost{\cdots} & \qw & \qw & \qw & \qw & & & \frac{p_d(\Ham)\ket{\varphi}}{\norm{p_d(\Ham)\ket{\varphi}}}
\end{myqcircuit}
$}
}\\
\caption{\polyfoqcs{} block encoding circuits for a matrix polynomial $p_d(\Ham) = a_0 I + a_1 \Ham + a_2 \Ham^2 + \cdots + a_d \Ham^d$.}
\label{fig:poly}
\end{figure*}

The circuit in \cref{fig:poly_general_pr} is fully general and applies to any matrix $\Ham$ implemented through the \foxlcu{} block encoding.
In principle, controlling the associated state-preparation oracles can incur a significant overhead.
However, the circuit identities in \cref{eq:pr_control_uncontrol,eq:pldag_control_uncontrol} show that the controlled \PR{} and \PLdag{} oracles can be simplified if \cref{ass:PR-PL} is valid, leading to the simplified construction of the following corollary.

\begin{corollary}
\label{corol:matrix-poly-simp}
If \PR{} and \PL{} satisfy \cref{ass:PR-PL}, the block-encoding circuit in \cref{fig:poly_general_pr} admits the simplification shown in \cref{fig:poly_pr_ab}.
\end{corollary}
\begin{proof}
 The circuit in \cref{fig:poly_pr_ab} follows directly from \cref{fig:poly_general_pr}, assuming \PR{} and \PLdag{} satisfy \cref{ass:PR-PL}.
\end{proof}

%%%%%%%%%%%%%%%%%%%%%%%%%%%%%%%%%%%
\subsection{Controlling the matrix polynomial of $\Ham$ }\label{sec:control_simp_poly}
%%%%%%%%%%%%%%%%%%%%%%%%%%%%%%%%%%% 
To control the entire circuits implementing $p_d(\Ham)$ shown in \cref{fig:poly}, we start with the following two lemmas, which detail how to efficiently control the \polyR{} and \polyL{} circuits.
\begin{lemma}\label{lemma:control_polyr}
 Let the {\upshape\polyR{}} circuit be defined as in \cref{lem:POLYR}. Then, controlling the entire {\upshape\polyR{}} requires controlling only the first $R_y$ gate when we apply the gate to the $\ket{0}_d$ state:
 {\upshape\begin{equation}
 \begin{split}
 &{\begin{myqcircuit}
 & \qw & \ctrl{1} & \qw \\
 \lstick{\ket{0}} & \qw{/\strut^{d}} & \gate{\polyR} & \qw
 \end{myqcircuit}}
 \ \ \raisebox{0.5ex}{=} \\[1ex]
 &\qquad{\begin{myqcircuitr}{-0.2}
 & \ctrl{1} & \qw & \qw & \qw \\
 \lstick{\ket{0}} & \gate{R_y(\theta_0)} & \multigate{3}{\cry(\bftheta_1^{d-1})} & \multigate{3}{\text{P}(\bfphi_0^{d-1})} & \qw \\
 \lstick{\ket{0}} & \qw & \ghost{\cry(\bftheta_1^{d-1})} & \ghost{\text{P}(\bfphi_1^{d-1})} & \qw \\
 & \vdots & & & \\
 \lstick{\ket{0}} & \qw & \ghost{\cry(\bftheta_1^{d-1})} & \ghost{\text{P}(\bfphi_0^{d-1})} & \qw \\
 \end{myqcircuitr}}
 \end{split}
 \end{equation}}%
 \end{lemma}
 \begin{proof}
 The \polyR{} unitary can be decomposed using \cref{lemma:unitary_decomposition_initial_state} with $s=1$. In particular, from the circuit implementation in \cref{eq:circuit_poly_r}, we can isolate the first $R_y$ rotation, since the remaining part of the circuit has $\ket{0}_d$ as an eigenstate with eigenvalue $1$.
 \end{proof}
 
 \begin{lemma}\label{lemma:control_polyl}
 Let the {\upshape\polyL{}} circuit be defined as in \cref{lem:POLYL}.
 Then, controlling the entire {\upshape$\polyL^\dag$} requires controlling only the last $R_y$ gate when we measure all the bottom $d$ qubits in the $\ket{0}_d$ state:
 {\upshape\begin{equation}
 \begin{split}
 &{\begin{myqcircuitr}{0}
 & \qw & \ctrl{1} & \qw \\
 & \qw{/\strut^{d}} & \gate{\polyL^\dagger} & \meter{} & \cw & 0
 \end{myqcircuitr}}
 \ \ \raisebox{1ex}{=} \\[1ex]
 &\qquad{\begin{myqcircuitr}{-0.2}
 & \qw & \ctrl{1}& \qw &  \qw \\
 & \multigate{3}{{\cry(\bftheta_1^{d-1})}^\dagger} & \gate{R_y(-\theta_0)} & \meter{} & \cw & 0 \\
 & \ghost{{\cry(\bftheta_1^{d-1})}^\dagger} & \qw & \meter{} & \cw & 0 \\
 & & & \cdots & \\
 & \ghost{{\cry(\bftheta_1^{d-1})}^\dagger} & \qw & \meter{} & \cw & 0 \\
 \end{myqcircuitr}}
 \end{split}
 \end{equation}}
\end{lemma}
\begin{proof}
 The $\polyL{}^\dag$ unitary can be decomposed using \cref{lemma:unitary_decomposition_measurement} with $s=1$. By considering the adjoint circuit of \cref{eq:circuit_poly_l}, we can isolate the last $R_y$ rotation in $\polyL^\dagger$, since the remaining part of the circuit has $\ket{0}_d$ as an eigenstate with eigenvalue $1$.
\end{proof}

Consequently, controlling the entire matrix polynomial circuit requires controlling only these boundary $R_y$ rotations. This is summarized in the following theorem.

\begin{theorem}
 Let $p_d(\Ham)$ be the matrix polynomial defined in \cref{eq:poly_transform_H}, and consider both the general and simplified block encoding circuits shown in \cref{fig:poly}. For either implementation, controlling the entire circuit is equivalent to controlling only the first and last $R_y$ rotations of {\upshape\polyR{}} and {\upshape$\polyL^\dagger$}, respectively.
\end{theorem}
\begin{proof}
 First, we note that in order to control the entire circuit, it suffices to control only the \polyR{} and \polyL{} oracles defined in \cref{def:POLY}. 
 Indeed, if these oracles are not activated, the controlled \PR{} and \PLdag{} gates are not applied, since the control register remains in $\ket{0}_d$, and the entire circuit thus acts trivially on the state $\ket{0}_d\ket{0}_{2dn}\ket{\varphi}$.
 Then we apply \cref{lemma:control_polyr,lemma:control_polyl}, such that both circuits in \cref{fig:poly} can be controlled by simply controlling the first and last $R_y$ rotation gates within \polyR{} and $\polyL^\dagger$, respectively.
 This completes the proof.
\end{proof}

We conclude by remarking that a controlled-$R_y$ gate can be implemented with up to $2$ \cnot{} gates.
As a result, all matrix polynomial block encodings defined in \cref{sec:poly_general} can be controlled with a trivial overhead of 4 \cnot{}s.

%%% TABLE %%%
\begin{table*}
 \begin{tabular}{|c|cc|cc|cc|}
 \hline
 \multirow{2}{*}{\bf Circuit} & \multicolumn{2}{c|}{\bf \cnot{} count} & \multicolumn{2}{c|}{\bf \cnot{} depth} & \multicolumn{2}{c|}{\bf Number of qubits} \\
 \cline{2-7}
 & all-to-all 
 & square grid 
 & all-to-all 
 & square grid 
 & all-to-all 
 & square grid   \\
 \hline
 \PR{}/\PLdag{} & $11 \n -11$ & $11\n-7$& $4\n+4$ & $4\n + 8$ & $2n\phantom{{}+1}$ & $2n\phantom{{}+1}$ \\
 Controlled-\PR{}/\PLdag{} & $11 \n -10$ & $11\n-6$ & $4\n+5$ &$4\n + 9$ & $2n+1$ & $2n+1$ \\
 \hline
 \foxlcu{} & $24 \n -22 $ & $26 \n -14 $& $8\n+10$ &$8\n+20$ & $3n\phantom{{}+1}$ & $3n\phantom{{}+1}$ \\
 Controlled-\foxlcu{} & $24 \n -20$ & $26 \n -12 $& $8\n+12$ & $8\n+22$ & $3n+1$&$3n+1$ \\
 \hline
 $p_d(\Ham)$ & \,$24 d\n-18d-4$\, & \,$26 d\n-2d-6$\, & $8\n + 6d +6\phantom{1}$ & $8\n+10d+12$& $2dn+d+n\phantom{{}+1}$ & $2dn+2d+n\phantom{{}+1}$ \\
 Controlled-$p_d(\Ham)$ & \,$24 d\n-18d\phantom{{}-4}$\, & $26 d\n-2d-2$ & \,$8\n + 6d + 10$\, & \,$8\n+10d+16$\, & \,$2dn+d+n+1$\, & \,$2dn+2d+n+1$\, \\
 \hline
 \end{tabular}
 \caption{Resource analysis for the \PR{} and \PLdag{} oracles, the \foxlcu{} block encoding, a generic matrix polynomial of $\Ham$ and their controlled circuits for the one-dimensional XYZ Heisenberg Hamiltonian defined in \cref{eq:def_heisenberg xyz}.}
 \label{tab:resource_estimation_xyz}
\end{table*}

%%%%%%%%%%%%%%%%%%%%%%%%%%%%%%%%%%%
\section{Applications}\label{sec:applications}
%%%%%%%%%%%%%%%%%%%%%%%%%%%%%%%%%%%
As discussed in the previous sections, the circuit depth of controlled \foxlcu{} block encodings can be significantly reduced when \cref{ass:PR-PL} is satisfied. In particular, decomposing the \PR{} and \PLdag{} oracles into two components---one consisting solely of $\bigO{1}$ single-qubit gates and another for which the all-zero state $\ket{0}_{2n}$ is an eigenstate---enables substantial simplifications of controlled \foxlcu{} circuits. While these conditions are not met by arbitrary state-preparation routines, they do hold for all block encodings introduced in \cite{foqcs-lcu-arxiv}, making these simplifications practically achievable.

Motivated by this observation, we focus in this section on representative one-dimensional spin Hamiltonians.
In \cref{sec:spin_hamiltonians} we present explicit decompositions of \PR{} and \PLdag{}. These decompositions require negligible overhead for controlling the \foxlcu{} circuits and allow for substantial simplifications in the implementation of matrix polynomials. 
We then present the Hadamard test as a representative application of controlled \foxlcu{} circuits in \cref{sec:hadamard_tests}. We also demonstrate applications of matrix polynomial circuits to time evolution in \cref{sec:time_evolution}.
Finally, in \cref{sec:qsvt_comp} we compare the \polyfoqcs{} block encoding with the Quantum Singular Value Transformation method applied to both standard LCU and \foxlcu{}.

%%%%%%%%%%%%%%%%%%%%%%%%%%%%%%%%%%%
\subsection{Explicit circuits for spin Hamiltonians}\label{sec:spin_hamiltonians}
%%%%%%%%%%%%%%%%%%%%%%%%%%%%%%%%%%%

Let us consider the one-dimensional XYZ Heisenberg Hamiltonian with open boundary conditions:
\begin{equation}\label{eq:def_heisenberg xyz}
\Ham = g \sum_{i=0}^{n-1} Z_i + \sum_{i=0}^{n-2} \Jx X_i X_{i+1} +\Jy Y_i Y_{i+1} +\Jz Z_i Z_{i+1}.
\end{equation}
%Note that more general field terms $\sum_{i} \gx X_i + \gy Y_i + \gz Z_i$, as considered in \cite{foqcs-lcu-arxiv}, can be rotated to the $Z$-axis through a single global change of spin coordinates. Without loss of generality, we may therefore choose $\gx = \gy =0$ and $\gz = g$.
Note that this Hamiltonian is simplified relative to the Heisenberg model considered in \cite{foqcs-lcu-arxiv}, with no $X_i$ or $Y_i$ terms, allowing us to further compress the \PR{} and \PL{} oracles.
%We exploit this simplification in the Hamiltonian definition to further compress the \PR{} and \PL{} oracles from \cite{foqcs-lcu-arxiv}. 
In particular, \cref{fig:simp_pr_1} illustrates the optimized \PR{} oracle for $n=4$ spins, while \cref{app:xyz_model} provides a detailed construction for general $n$. 

%%% FIGURE %%%
\begin{figure*}
\begin{equation*}
\qquad\begin{myqcircuitr}{0}
\lstick{\ket{0}} & \qw & \qw & \qw & \qw & \qw & \qw & \qw & \qw & \qw & \qw & \qw & \qw & \qw & \qw & \qw & \qw & \qw & \qw & \qw & \qw & \qw& \qw & \qw & \qw & \targ & \qw & \qw & \qw & \qw & \qw & \\
\lstick{\ket{0}} & \gate{X} & \qw & \ctrl{3} & \targ & \ctrl{1} & \targ & \qw & \qw & \qw & \qw & \qw & \qw & \qw & \qw & \qw & \qw & \qw & \qw & \ctrl{4} & \qw & \qw & \qw & \qw & \qw & \ctrl{-1} & \targ & \qw & \qw & \qw & \qw & \\
\lstick{\ket{0}} & \qw & \qw & \qw & \qw & \gate{R_y} & \ctrl{-1} & \ctrl{1} & \targ & \qw & \qw & \qw & \qw & \qw & \qw & \qw & \qw & \qw & \qw & \qw & \qw & \ctrl{4} & \qw & \qw& \qw & \qw & \ctrl{-1} & \targ & \qw & \qw & \qw & \\
\lstick{\ket{0}} & \qw & \qw & \qw & \qw & \qw & \qw & \gate{R_y} & \ctrl{-1} & \qw & \qw & \qw & \qw & \qw & \qw & \qw & \qw & \qw & \qw & \qw & \qw & \qw & \qw & \ctrl{4} & \qw & \qw & \qw & \ctrl{-1} & \qw & \qw & \qw & \\
\lstick{\ket{0}} & \qw & \qw & \gate{R_y} & \ctrl{-3} & \ctrl{1} & \targ & \qw & \qw & \qw & \qw & \gate{R_y} & \qw & \qw & \qw & \qw & \qw & \qw & \targ & \qw & \qw & \qw & \qw & \qw & \qw & \targ & \qw & \qw & \qw & \qw & \qw & \\
\lstick{\ket{0}} & \qw & \qw & \qw & \qw & \gate{R_y} & \ctrl{-1} & \ctrl{1} & \targ & \qw & \qw & \ctrl{-1} & \gate{R_y} & \qw & \qw & \qw & \qw & \targ & \ctrl{-1} & \gate{R_y} & \qw & \qw & \qw & \qw & \qw & \ctrl{-1} & \targ & \qw & \qw & \qw & \qw & \\
\lstick{\ket{0}} & \qw & \qw & \qw & \qw & \qw & \qw & \gate{R_y} & \ctrl{-1} & \ctrl{1} & \targ & \qw & \ctrl{-1} & \gate{R_y} & \qw & \qw & \targ & \ctrl{-1} & \qw & \qw & \qw & \gate{R_y} & \qw & \qw & \qw & \qw & \ctrl{-1} & \targ & \qw & \qw & \qw & \\
\lstick{\ket{0}} & \qw & \qw & \qw & \qw & \qw & \qw & \qw & \qw & \gate{R_y} & \ctrl{-1} & \qw & \qw & \ctrl{-1} & \qw & \qw & \ctrl{-1} & \qw & \qw & \qw & \qw & \qw & \qw & \gate{R_y} & \qw & \qw & \qw & \ctrl{-1} & \qw & \qw & \qw \gategroup{1}{4}{8}{28}{2.5em}{--} \\
\end{myqcircuitr}
\end{equation*}
\caption{Simplified \PR{} circuit for the one-dimensional XYZ Heisenberg Hamiltonian with open boundary conditions, defined in \cref{eq:def_heisenberg xyz}, for the case $n=4$. The dashed region corresponds to the gate $\PRmod$ in \cref{ass:PR-PL}.}
\label{fig:simp_pr_1}
\end{figure*}

Crucially, the \PR{} oracle from \cref{fig:simp_pr_1} satisfies \cref{ass:PR-PL} and can be decomposed into two parts by isolating the first $\Xp$ gate from the remainder of the circuit:
\begin{equation}\label{eq:Prmod_def}
 \begin{myqcircuit}
 & \multicgate{3}{\PR{}}{PRcolor} & \qw \\
 & \ghost{\PR{}} & \qw \\
 \lstick{\raisebox{1ex}{$\vdots$}\hspace{-0.75em}} & & \rstick{\hspace{-0.75em}\raisebox{1ex}{$\vdots$}} \\
 & \ghost{\PR{}} & \qw
 \end{myqcircuit}
 \ \ \equiv \ \
 \begin{myqcircuitr}{0}
 & \qw & \multicgate{3}{\PRmod{}}{PRcolor} & \qw \\
 & \gate{\Xp} & \ghost{\PRmod{}} & \qw \\
 & \cdots & && \\
 &\qw & \ghost{\PRmod{}} & \qw
 \end{myqcircuitr}
\end{equation}
Because $\PRmod$ contains only controlled gates, it has eigenstate $\ket{0}_{2n}$ with eigenvalue $1$. 
As a consequence, \cref{eq:pr_control_uncontrol} yields: 
\begin{equation}\label{eq:trivial_control_pr}
 \begin{myqcircuit}
 & \ctrl{1} & \qw \\
 \lstick{\ket{0}} & \multicgate{3}{\PR{}}{PRcolor} & \qw \\
 \lstick{\ket{0}} & \ghost{\PR{}} & \qw \\
 \lstick{\raisebox{1.5ex}{$\vdots$}\hspace*{1.25ex}} \\
 \lstick{\ket{0}} & \ghost{\PR{}} & \qw
 \end{myqcircuit}
 \ \ = \qquad \
 \begin{myqcircuitr}{0}
 & \ctrl{2} & \qw & \qw \\
 \lstick{\ket{0}} & \qw & \multicgate{3}{\PRmod{}}{PRcolor} & \qw \\
 \lstick{\ket{0}} & \gate{\Xp} & \ghost{\PRmod{}} & \qw \\
 \lstick{\raisebox{1.5ex}{$\vdots$}\hspace*{1.25ex}} \\
 \lstick{\ket{0}} &\qw & \ghost{\PRmod{}} & \qw
 \end{myqcircuitr}
\end{equation}
Controlling the XYZ Heisenberg \PR{} oracle thus incurs a trivial overhead, namely \emph{a single additional} \cnot{} gate.
Analogously, \cref{ass:PR-PL,eq:pldag_control_uncontrol} also apply to the simplification of the controlled \PLdag{} gate:
\begin{equation}\label{eq:trivial_control_pldag}
 \begin{myqcircuit}
 & \ctrl{1} & \qw \\
 & \multicgate{3}{\PLdag{}}{PRcolor} & \meter{} & \cw & 0\\
 & \ghost{\PLdag{}} & \meter{} & \cw & 0 \\
 & & & \raisebox{1.5ex}{$\vdots$}\hspace*{1.25ex} & \\
 &\ghost{\PLdag{}} & \meter{} & \cw & 0
 \end{myqcircuit}
 \ \ \, = \ \
 \begin{myqcircuitr}{0}
 & \qw & \ctrl{2} & \qw \\
 & \multicgate{3} {\PLdagmod{}}{PRcolor} & \qw & \meter{} & \cw & 0 \\
 & \ghost{\PLdagmod{}} & \gate{\Xp} & \meter{} & \cw & 0 \\
 & & & & \raisebox{1.5ex}{$\vdots$}\hspace*{1.25ex} & \\
 & \ghost{\PLdagmod{}} & \qw & \meter{} & \cw & 0
 \end{myqcircuitr}
\end{equation}

In \cref{tab:resource_estimation_xyz}, we evaluate the non-asymptotic resource cost for this model, for which \cref{eq:trivial_control_pr} and \cref{eq:trivial_control_pldag} hold. 
To derive the circuit depth and \cnot{} count, we consider both all-to-all and square-grid qubit connectivity. We observe that the circuit implemented on a two-dimensional grid has nearly the same complexity as the all-to-all connectivity case. This is made possible by an appropriate mapping of qubits onto the two-dimensional grid, which allows us to effectively exploit horizontal and vertical nearest-neighbor qubit interactions. \cref{app:2D circuit} details the implementation of two-dimensional \foxlcu{} and \polyfoqcs{} circuits on a square grid, while \cref{sec: non_clifford_cost} gives their non-Clifford costs.

In addition to the XYZ Heisenberg model, we explicitly provide the \PR{} circuits for the XXZ and Ising models in \cref{app:XXZ_model,app:ising_model}. In both cases, the decomposition in \cref{eq:Prmod_def} applies, as do the corresponding simplifications for the controlled circuits in \cref{eq:trivial_control_pr,eq:trivial_control_pldag}.

The consequences of this simplification are significant: for a broad class of spin models considered here and in \cite{foqcs-lcu-arxiv}, controlling the full \foxlcu{} block encoding incurs the cost of \emph{only two additional} \cnot{} gates, as illustrated in \cref{fig:foqcs_lcu_controlled}.
Moreover, the circuit implementing a matrix polynomial of $\Ham$ admits a substantial simplification as well:
control of each $\PRc$ or $\PLcdag$ gate in \cref{fig:poly_pr_ab} introduces only a single \cnot{} gate.
% each controlled gate $\PRc$ and $\PLcdag$ in \cref{fig:poly_pr_ab} reduces to a single \cnot{} gate.
For concreteness, \cref{fig:poly_xyz_depth_2D} shows the total \cnot{} depth required to implement matrix polynomials of increasing degree for the XYZ model.

%%% FIGURE %%%
\begin{figure*}
\[
\begin{myqcircuitr}{0.9}
 & \qw & \ctrl{1} & \qw \\
 \lstick{\ket{0}} & \qw{/\strut^{n}} & \multigate{2}{\Ufox} & \qw \\
 \lstick{\ket{0}} & \qw{/\strut^{n}} & \ghost{\Ufox} & \qw \\
 \lstick{\ket{\varphi}}& \qw {/\strut^{n}} & \ghost{\Ufox} & \qw & \\
\end{myqcircuitr}
\quad = \quad\qquad
\begin{myqcircuitr}{0}
% Ancillae
& \ctrl{2} & \qw &\qw &\qw &\qw &\qw &\qw &\qw &\qw & \qw & \qw & \ctrl{2} & \qw &\\
% X
\lstick{\ket{0}} &\qw &\multicgate{7}{\PRmod{}}{PRcolor} & \ctrl{8} & \qw &\ghost{\cdots}& \qw & \qw & \qw &\ghost{\cdots}& \qw & \multicgate{7}{\PLdagmod}{PLcolor} & \qw & \qw \\
\lstick{\ket{0}} & \gate{\Xp} &\ghost{\PRmod{}} & \qw & \ctrl{8} &\ghost{\cdots}& \qw & \qw & \qw &\ghost{\cdots}& \qw & \ghost{\PLdagmod} & \gate{\Xp} & \qw \\
\lstick{\raisebox{1.5ex}{$\vdots$}\hspace*{1.25ex}} &&&&& \cdots &&&&  && \\
\lstick{\ket{0}} & \qw &\ghost{\PRmod{}} & \qw & \qw &\ghost{\cdots}& \ctrl{8} & \qw & \qw &\ghost{\cdots}& \qw & \ghost{\PLdagmod} & \qw & \qw \\
% Z
\lstick{\ket{0}} & \qw &\ghost{\PRmod{}} & \qw & \qw &\ghost{\cdots}& \qw & \ctrl{4} & \qw & \ghost{\cdots} & \qw & \ghost{\PLdagmod} & \qw & \qw \\
\lstick{\ket{0}} & \qw &\ghost{\PRmod{}} & \qw & \qw &\ghost{\cdots}& \qw & \qw & \ctrl{4} & \ghost{\cdots} & \qw & \ghost{\PLdagmod} & \qw & \qw \\
\lstick{\raisebox{1.5ex}{$\vdots$}\hspace*{1.25ex}} &&&&&  &&&&\cdots&&&\\
\lstick{\ket{0}} & \qw &\ghost{\PRmod{}} & \qw & \qw &\ghost{\cdots}& \qw & \qw & \qw & \ghost{\cdots} & \ctrl{4} & \ghost{\PLdagmod} & \qw & \qw \\
% state
 & \qw& \qw & \targ{} & \qw &\ghost{\cdots}& \qw & \ctrl{0} & \qw & \ghost{\cdots} & \qw & \qw & \qw & \qw &\\
& \qw & \qw & \qw & \targ{} &\ghost{\cdots}& \qw & \qw & \ctrl{0} & \ghost{\cdots} & \qw & \qw & \qw & \qw & \\
&&&&& \cdots &&&&\cdots&& \\
& \qw & \qw & \qw & \qw &\ghost{\cdots}& \targ{} & \qw & \qw & \ghost{\cdots} & \ctrl{0} & \qw & \qw & \qw
\inputgroupv{10}{13}{1em}{2em}{\ket{\varphi}}
\end{myqcircuitr}
\]
\caption{Circuit implementing the controlled \foxlcu{} block encoding of $\Ham$.}
\label{fig:foqcs_lcu_controlled}
\end{figure*}

%%% FIGURE %%%
\begin{figure}[h]
\centering
\begin{tikzpicture}
\begin{axis}[
    width=0.43\textwidth,
    view={0}{-90},
    xtick={2,4,...,20},
    ytick={1,2,...,10},
    xtick style={draw=none},
    ytick style={draw=none},
    yticklabel pos=right,
    xlabel={Number of sites $n$},
    ylabel={Polynomial degree $d$},
    ylabel style={rotate=90},
    colorbar,
    colorbar style={
        ytick={50, 100,...,300}
    },
    title={\cnot{} depth for $p_d(\Ham)$ in the XYZ model}
]
\addplot3[
    matrix plot,
    mesh/cols=10,
    nodes near coords,
    nodes near coords style={font=\tiny},
    coordinate style/.condition={meta < 150}{white},
    nodes near coords style={anchor=center}
] table[header=true] {Graphs/poly_xyz_depth_2D.dat};
\end{axis}
\end{tikzpicture}
\caption{Total \cnot{} depth for implementing the matrix polynomial $p_d(\Ham)$ of the one-dimensional XYZ Heisenberg Hamiltonian as a function of system size $n$ and polynomial degree $d$, assuming a square-grid connectivity (See \cref{tab:resource_estimation_xyz}).}
\label{fig:poly_xyz_depth_2D}
\end{figure}

Finally, for completeness, a full and explicit circuit implementing a matrix polynomial of the Ising Hamiltonian is included in \cref{sec:full_circuit_ising} for $n = 3$ and $d = 3$.

%%%%%%%%%%%%%%%%%%%%%%%%%%%%%%%%%%%
\subsection{Hadamard tests}\label{sec:hadamard_tests}
%%%%%%%%%%%%%%%%%%%%%%%%%%%%%%%%%%%
The Hadamard test is a canonical method to estimate the expectation value $\braket{\varphi|\genGate|\varphi}$ of a unitary $\genGate$ in a state $\ket{\varphi}$. This is achieved by controlling $\genGate$ and then measuring an ancilla qubit: $X$-basis measurements yield $\real \braket{\varphi|\genGate|\varphi}$ and $Y$-basis measurements $\imag \braket{\varphi|\genGate|\varphi}$. 
Specifically, the Hadamard-test circuit for extracting the real part is:
\begin{equation}\label{eq:def_hadamard_test}
 \begin{myqcircuit}
 \lstick{\ket{0}} & \gate{H} & \ctrl{1} & \gate{H} & \meter{}  & \qquad\qquad\qquad\leadsto\real{\bra{\varphi}\genGate\ket{\varphi}}\\
 \lstick{\ket{\psi}} & \qw {/\strut^{n}} & \gate{\genGate} & \qw & \qw
 \end{myqcircuit}
\end{equation}
where $H$ is the Hadamard gate. 
In the case of a non-unitary operator $\Ham$, we can still perform the Hadamard test by controlling the entire block encoding circuit while initializing the \foxlcu{} ancilla qubits in the state $\ket{0}_{2n}$:
\begin{equation}\label{eq:hadamard_block_encoding}
 \begin{myqcircuit}
 \lstick{\ket{0}} & \gate{H} & \ctrl{1} & \gate{H} & \meter{} & &\qquad\qquad\rightsquigarrow\real{\bra{\varphi}\Ham\ket{\varphi}}\\
 \lstick{\ket{0}} & \qw{/\strut^{n}} & \multigate{2}{\Ufox} & \qw & \qw &  \\
 \lstick{\ket{0}} & \qw{/\strut^{n}} & \ghost{\Ufox} & \qw & \qw & \\ 
 \lstick{\ket{\varphi}} & \qw{/\strut^{n}} & \ghost{\Ufox} & \qw & \qw
 \end{myqcircuit}
\end{equation}
Notice that, in this framework, post-selecting the block encoding ancilla qubits is not strictly necessary. A detailed explanation is given in \cref{app: Hadamard test}.

In general, controlling a block encoding of $\Ham$ incurs significant overhead.
In \cref{sec:control_foqcs_lcu}, however, we show that this overhead can be substantially reduced for the \foxlcu{} block encoding when the state-preparation oracles satisfy \cref{ass:PR-PL}.
Specifically, in \cref{sec:spin_hamiltonians} we demonstrate that this assumption holds for the one-dimensional XYZ Heisenberg model. It likewise holds for the XXZ and Ising models, as verified in \cref{app:XXZ_model,app:ising_model}.
In all these cases, controlling the block encoding of $\Ham$ costs only two additional \cnot{} gates.

As a result, \foxlcu{} is particularly well suited for Hadamard-test-based applications.
\Cref{fig:hadamard_test} summarizes the resources required to implement a Hadamard test for each considered spin Hamiltonian, for both all-to-all and square-grid connectivity architectures.
We observe that, in all cases, both the \cnot{} count and \cnot{} depth for the block encoding circuits scale linearly with the system size $n$.
Moreover, moving from all-to-all to square-grid connectivity incurs only a modest overhead; in fact, the circuit-depth overhead is constant.

%%% FIGURE %%%
\begin{figure*}
 \centering
 \subfloat[]{%
 \begin{tikzpicture}
\begin{axis}[
  width=0.5\textwidth,
  xmin=0,
  ymin=0, ymax=600,
  xlabel={Number of sites $\n$},
  ylabel={\cnot{} count},
  legend pos=north west,
  legend style={draw=none,fill=none},
  title={\cnot{} count for Hadamard test}
]
% XYZ all-to-all
\addplot[myblue, mark=square*] table[x index=0, y index=1] {Graphs/xyz_count.dat};
% XYZ square grid
\addplot[myblue, mark=square*, dashed, mark options={solid}] table[x index=0, y index=2] {Graphs/xyz_count.dat};
% XXZ all-to-all
\addplot[myred, mark=*] table[x index=0, y index=1] {Graphs/xxz_count.dat};
% XXZ square grid
\addplot[myred, mark=*, dashed, mark options={solid}] table[x index=0, y index=2] {Graphs/xxz_count.dat};
% Ising all-to-all
\addplot[mygreen, mark=triangle*] table[x index=0, y index=1] {Graphs/ising_count.dat};
% Ising square grid
\addplot[mygreen, mark=triangle*, dashed, mark options={solid}] table[x index=0, y index=2] {Graphs/ising_count.dat};
\legend{XYZ all-to-all, XYZ square grid, XXZ all-to-all, XXZ square grid, Ising all-to-all, Ising square grid} 
\end{axis}
\end{tikzpicture}
 \label{fig:count_hadamard_test}}%
 \hfill
 \subfloat[]{%
 \begin{tikzpicture}
\begin{axis}[
  width=0.5\textwidth,
  xmin=0,
  ymin=0,ymax=200,
  xlabel={Number of sites $\n$},
  ylabel={\cnot{} depth},
  legend pos=north west,
  legend style={draw=none,fill=none},
  title={\cnot{} depth for Hadamard test}
]
% XYZ all-to-all
\addplot[myblue, mark=square*] table[x index=0, y index=1] {Graphs/xyz_depth.dat};
% XYZ square grid
\addplot[myblue, mark=square*, dashed, mark options={solid}] table[x index=0, y index=2] {Graphs/xyz_depth.dat};
% XXZ all-to-all
\addplot[myred, mark=*] table[x index=0, y index=1] {Graphs/xxz_depth.dat};
% XXZ square grid
\addplot[myred, mark=*, dashed, mark options={solid}] table[x index=0, y index=2] {Graphs/xxz_depth.dat};
% Ising all-to-all
\addplot[mygreen, mark=triangle*] table[x index=0, y index=1] {Graphs/ising_depth.dat};
% Ising square grid
\addplot[mygreen, mark=triangle*, dashed, mark options={solid}] table[x index=0, y index=2] {Graphs/ising_depth.dat};
\legend{XYZ all-to-all, XYZ square grid, XXZ all-to-all, XXZ square grid, Ising all-to-all, Ising square grid}
\end{axis}
\end{tikzpicture}
 \label{fig:depth_hadamard_test}}%
 \caption{\protect\cnot{} count \subref{fig:count_hadamard_test} and depth \subref{fig:depth_hadamard_test} for a circuit implementing the Hadamard test, defined in \cref{eq:hadamard_block_encoding}. The computational cost for preparing the state $\ket{\varphi}$ is excluded. We consider the one-dimensional XYZ and XXZ Heisenberg models as well as the Ising model, with computational costs detailed in \cref{tab:resource_estimation_xyz,tab:resource_estimation_xxz,tab:resource_estimation_ising}.}
 \label{fig:hadamard_test}
\end{figure*}

%%%%%%%%%%%%%%%%%%%%%%%%%%%%%%%%%%%
\subsection{Time evolution and Loschmidt echo}\label{sec:time_evolution}
%%%%%%%%%%%%%%%%%%%%%%%%%%%%%%%%%%%
A direct application of matrix polynomial transformations is the approximation of the time-evolution operator by a degree-$d$ polynomial,
\begin{equation}\label{eq:approx_time_evolution}
e^{-i\Ham t} \approx p_d(\Ham) = \sum_{k=0}^d a_k \Ham^k ,
\end{equation}
which can be obtained, for instance, through a Taylor or Chebyshev expansion.

Additionally, the Hadamard tests can be applied to the time evolution operator to measure the Loschmidt echo $g(t) = \bra{\varphi} e^{-i\Ham t} \ket{\varphi}$:
\begin{equation}\label{eq:def_loschmidt_circuit}
\begin{myqcircuit}
\lstick{\ket{0}} & \gate{H} & \ctrl{1} & \gate{H} & \meter{} && \qquad\quad\rightsquigarrow \real(g(t)) \\
\lstick{\ket{\varphi}} & \qw{/\strut^{n}} & \gate{e^{-i\Ham t}} & \qw & \qw
\end{myqcircuit}
\end{equation}
The circuit above constitutes a core component of many hardware-efficient quantum phase estimation (QPE) methods~\cite{Somma_2019,Lu2021,Lin2022,blunt2023,Dutkiewicz2022,schiffer2025}.

With the approximation of \cref{eq:approx_time_evolution}, we can estimate $g(t) \approx \bra{\varphi} p_d(\Ham) \ket{\varphi}$ by controlling the block encoding of the matrix polynomial $p_d(\Ham)$ while initializing all ancilla qubits in the state $\ket{0}_{2dn+d}$, similarly to \cref{eq:hadamard_block_encoding}:
\begin{equation}
\hspace*{-2cm}
 \begin{myqcircuit}
 \lstick{\ket{0}} & \gate{H} & \qw & \ctrl{1} & \gate{H} & \meter{} & & \qquad\qquad\qquad\rightsquigarrow\real{\bra{\varphi} p_d(\Ham) \ket{\varphi}}\\
 \lstick{\ket{0}} & \qw{/\strut^{2dn+d}} & \qw & \multigate{1}{p_d(\Ham)} & \qw & \qw & \\
 \lstick{\ket{\varphi}} & \qw {/\strut^{n}} & \qw & \ghost{p_d(\Ham)} & \qw & \qw
 \end{myqcircuit}
\end{equation}
Also in this case, measuring and post-selecting the block encoding ancillae is not strictly necessary, as illustrated in \cref{app: Hadamard test}.

In \cref{sec:control_simp_poly}, we establish that controlling the circuit implementing $p_d(\Ham)$ incurs only a trivial overhead of four additional \cnot{} gates. Consequently, the \foxlcu{} block encoding can substantially reduce the control cost associated with QPE routines, therefore improving their feasibility on near-term quantum hardware.

%%%%%%%%%%%%%%%%%%%%%%%%%%%%%%%%%%%%%%%%
\subsection{Relation to the Quantum Singular Value Transformation}
\label{sec:qsvt_comp}
%%%%%%%%%%%%%%%%%%%%%%%%%%%%%%%%%%%%%%%%
%In this section, we consider an alternative method to implement a polynomial transformation of a matrix $\Ham$ on a quantum circuit, namely the Quantum Singular Value Transformation (QSVT) algorithm \cite{Gilyen2018,Martyn2021}.
The \polyfoqcs{} construction generalizes \foxlcu{} \cite{foqcs-lcu-arxiv} to directly encode a matrix polynomial $p_d(\Ham)$ of a Hamiltonian $\Ham$ in a quantum circuit. Now we compare this approach with two alternative techniques:
\begin{enumerate}
    %\item \polyfoqcs{}: the \foxlcu{} matrix polynomial block encoding introduced in this manuscript, which encodes $p_d(\Ham)$ directly;
    \item \qsvtlcu{}: a textbook LCU block encoding of $\Ham$, combined with the Quantum Singular Value Transformation (QSVT) \cite{Gilyen2018,Martyn2021};
    \item \qsvtfoqcs{}: \foxlcu{} block encoding of $\Ham$ \cite{foqcs-lcu-arxiv}, combined with the QSVT.
\end{enumerate}

The comparison is summarized in \cref{tab:qsvt_comparison}, with more detailed analysis provided within \cref{sec:qsvt}.
The three approaches have different advantages and disadvantages.  The two QSVT-based constructions have higher success probability, while the direct matrix-polynomial \foxlcu{} construction admits much lower circuit depth, as we briefly explain below.  Of the two QSVT approaches, the \foxlcu{} version achieves lower asymptotic scaling of circuit depth and therefore the smallest overall spacetime volume.
%Overall, \foxlcu{} reduces the circuit depth relative to textbook LCU in both the direct matrix-polynomial and QSVT-based constructions. The two \foxlcu{} approaches, however, exhibit a different trade-off between circuit depth and qubit count on one hand, and success probability on the other.

\begin{table*}
    \centering
    \begin{tabular}{|l|c|c|c|c|c|c|}
    \hline
        Method & CNOT depth & Number of qubits & Space-time volume & Near-term & Success probability & Parity constraint \\
        \hline
        \polyfoqcs{} & \bigO{n+d} & \bigO{n\cdot d} & \bigO{n\cdot d\cdot(n+d)} & Yes & Low ($1/\vecnorm{f_d}^2$) & None \\
        \qsvtlcu{} & \bigO{n\log(n)\cdot d} & \bigO{n} & \bigO{n^2\log(n)\cdot d} & No & High ($1/\norm{f_d}_\infty^2$) & Even/odd \\
        \qsvtfoqcs{} & \bigO{n\cdot d} & \bigO{n} & \bigO{n^2\cdot d} & No & High ($1/\norm{f_d}_\infty^2$) & Even/odd \\
        \hline
    \end{tabular}
    \caption{Comparison among (1) matrix polynomial \foxlcu{} introduced in this work, (2) QSVT combined with textbook LCU, and (3) QSVT combined with \foxlcu{}, for the spin Hamiltonians studied in the main text. Here $n$ denotes the number of system qubits and $d$ the polynomial degree. The norm $\vecnorm{\cdot}$ is defined in \cref{eq:def_vecnorm} as the 1-norm of the coefficients' vector, and the success probabilities are reported up to the common factor $\norm{f_d(\Ham/\normfact)\ket{\varphi}}^2$ (see \cref{eq:p_poly,eq:p_qsvt}).}
    \label{tab:qsvt_comparison}
\end{table*}

\paragraph*{Advantages of QSVT:} Both QSVT approaches are sequential: they interleave $d$ queries to the block encoding of $\Ham$ with $d+1$ signal processing operators.
The CNOT depth of \qsvtfoqcs{}, in particular, increases linearly with the product of the system size $\bigO{n}$ and polynomial degree $d$, while the qubit count remains $\bigO{n}$.
As a result, \qsvtfoqcs{} achieves the smallest asymptotic space-time volume of the three methods. The QSVT circuits also benefit from a smaller normalization factor and run with a higher success probability, as shown in \cref{sec:qsvt}.

\paragraph*{Advantages of direct polynomial approach:} \polyfoqcs{} instead prepares the $d+1$ powers of $\Ham$ in parallel, which reduces the \cnot{} depth to $\bigO{n+d}$ at the cost of $\bigO{n \cdot d}$ total qubits.
\polyfoqcs{} is therefore the only approach considered here whose depth is compatible with late-NISQ capabilities, whereas the QSVT-based circuits demand depths more naturally associated with fault-tolerant computation. 
Moreover, when the polynomial degree grows at most linearly with the system size, i.e., $d = \bigO{n}$, \polyfoqcs{} and \qsvtfoqcs{} share the same asymptotic space-time volume, $\bigO{n^2 d}$.
Finally, \polyfoqcs{} does not impose any parity constraint on $p_d$, whereas QSVT requires $f_d$ to be strictly even or odd.

%%%%%%%%%%%%%%%%%%%%%%%%%%%%%%%%%%%
\section{Conclusions}
\label{sec:conclusions}
%%%%%%%%%%%%%%%%%%%%%%%%%%%%%%%%%%%
By leveraging the \foxlcu{} formalism, we have shown that matrix polynomial transformations can be implemented with a \cnot{}-depth overhead that is only linear in the polynomial degree $d$, and independent of both the system size $n$ and the cost of block encoding the original matrix.
For the spin models considered here, this results in an overall \cnot{} depth scaling as the sum $c_1 n + c_2 d$, with small integer coefficients $c_1$ and $c_2$.
By contrast, a QSVT-based implementation for the same systems would entail a circuit depth proportional to the product $d\cdot n$.

This substantial reduction in circuit depth comes at an expense of additional ancilla qubits.
While this trade-off must be taken into account, it aligns well with the current trends in quantum hardware development: qubit counts continue to increase steadily, whereas achievable circuit depth remains a primary bottleneck due to noise and coherence constraints.
From this perspective, prioritizing shallow and highly parallelizable circuits is often more practical than minimizing the total qubit count. 

Moreover, we have shown that both the \foxlcu{} block encoding and the associated matrix polynomial circuits can be controlled with negligible overhead.
In particular, a Hadamard test applied to a \foxlcu{} block encoding requires only $2$ extra \cnot{} gates, while a Hadamard test applied to a matrix polynomial implemented using \foxlcu{} incurs the overhead of just $2$ controlled single-qubit rotations, corresponding to a total of $4$ extra \cnot{} gates.

For all applications considered, we provide explicit circuit constructions, including qubit mappings for two-dimensional square-grid architectures, and we show that the additional circuit depth incurred by this connectivity constraint remains constant with respect to the system size.

The circuits and techniques introduced in this paper are broadly applicable to problems across quantum simulation and quantum linear algebra, and are not limited to the specific physical models considered here.
Furthermore, the resulting circuit depths and qubit requirements bring several of these constructions within reach of current and near-term quantum hardware. In doing so, our framework opens the door to practical implementations of block-encoding-based algorithms that have, until now, remained largely theoretical due to prohibitive resource demands.

%%%%%%%%%%%%%%%%%%%%%%%%%%%%%%%%%%%
\section*{Acknowledgments}
%%%%%%%%%%%%%%%%%%%%%%%%%%%%%%%%%%%
MN acknowledges funding by the Munich Quantum Valley, section K5 Q-DESSI. The research is part of the Munich Quantum Valley, which is supported by the Bavarian state government with funds from the Hightech Agenda Bayern Plus.
This research is supported by the U.S. Department of Energy, Office of Science, Accelerated Research in Quantum Computing, Fundamental Algorithmic Research toward Quantum Utility (FAR-Qu).

\bibliography{bibliography}

@article{Kirby2023exactefficient,
doi = {10.22331/q-2023-05-23-1018},
title = {Exact and efficient {L}anczos method on a quantum computer},
author = {Kirby, William and Motta, Mario and Mezzacapo, Antonio},
journal = {{Quantum}},
publisher = {{Verein zur F{\"{o}}rderung des Open Access Publizierens in den Quantenwissenschaften}},
volume = {7},
pages = {1018},
year = {2023}
}

@article{foqcs-lcu-arxiv,
title = {Efficient {LCU} block encodings through {Dicke} states preparation},
author = {Della Chiara, Filippo and Nibbi, Martina and Shen, Yizhi and Van Beeumen, Roel},
year = {2025},
journal = {arXiv},
eprint = {2507.20887}
}

@article{Babbush2018,
title = {Encoding Electronic Spectra in Quantum Circuits with Linear {T} Complexity},
author = {Babbush, Ryan and Gidney, Craig and Berry, Dominic W. and Wiebe, Nathan and McClean, Jarrod and Paler, Alexandru and Fowler, Austin and Neven, Hartmut},
journal = {Phys. Rev. X},
volume = {8},
issue = {4},
pages = {041015},
numpages = {36},
year = {2018},
publisher = {American Physical Society},
doi = {10.1103/PhysRevX.8.041015},
}

@inbook{Bartschi_2019,
title={Deterministic Preparation of {Dicke} States},
ISSN={1611-3349},
DOI={10.1007/978-3-030-25027-0_9},
booktitle={Fundamentals of Computation Theory},
publisher={Springer International Publishing},
author={Bartschi, Andreas and Eidenbenz, Stephan},
year={2019},
pages={126–139}
}

@inproceedings{Bartschi_2022,
title={Short-Depth Circuits for {Dicke} State Preparation},
DOI={10.1109/qce53715.2022.00027},
booktitle={2022 IEEE International Conference on Quantum Computing and Engineering (QCE)},
author={Bartschi, Andreas and Eidenbenz, Stephan},
pages={87-96},
year={2022},
}

@article{Berry_2015,
title={Simulating {Hamiltonian} Dynamics with a Truncated {Taylor} Series},
volume={114},
doi={10.1103/physrevlett.114.090502},
issue={9},
pages = {090502},
journal={Phys. Rev. Lett.},
publisher={American Physical Society (APS)},
author={Berry, Dominic W. and Childs, Andrew M. and Cleve, Richard and Kothari, Robin and Somma, Rolando D.},
year={2015},
}

@misc{Boyd2024,
title={Low-Overhead Parallelisation of {LCU} via Commuting Operators},
author={Gregory Boyd},
year={2024},
eprint={2312.00696},
archivePrefix={arXiv},
}

@inproceedings{Camps2022,
title={{FABLE}: Fast Approximate Quantum Circuits for Block-Encodings},
DOI={10.1109/qce53715.2022.00029},
booktitle={2022 IEEE International Conference on Quantum Computing and Engineering (QCE)},
author={Camps, Daan and Van Beeumen, Roel},
pages={104-113},
year={2022},
}

@article{Camps2024,
doi = {10.1137/22M1484298},
author = {Camps, Daan and Lin, Lin and Van Beeumen, Roel and Yang, Chao},
issue = {1},
journal = {SIAM J. Matrix Anal. Appl.},
pages = {801-827},
publisher = {Society for Industrial and Applied Mathematics Publications},
title = {EXPLICIT QUANTUM CIRCUITS FOR BLOCK ENCODINGS OF CERTAIN SPARSE MATRICES},
volume = {45},
year = {2024}
}

@article{Chakraborty2024,
doi = {10.22331/q-2024-10-10-1496},
title = {Implementing any linear combination of unitaries on intermediate-term quantum computers},
author = {Chakraborty, Shantanav},
journal = {Quantum},
publisher = {{Verein zur F{\"{o}}rderung des Open Access Publizierens in den Quantenwissenschaften}},
volume = {8},
pages = {1496},
year = {2024}
}

@article{Childs2012,
author = {Childs, Andrew M. and Wiebe, Nathan},
title = {{Hamiltonian} Simulation Using Linear Combinations of Unitary Operations},
doi={10.26421/qic12.11-12},
volume={12},
number={11 \& 12},
pages={0901-0924},
journal={Quantum Inf. Comput.},
publisher={Rinton Press},
year={2012},
}

@article{Childs2018,
author = {Andrew M. Childs  and Dmitri Maslov  and Yunseong Nam  and Neil J. Ross  and Yuan Su },
title = {Toward the first quantum simulation with quantum speedup},
journal = {Proc. Natl. Acad. Sci. U.S.A.},
volume = {115},
number = {38},
pages = {9456-9461},
year = {2018},
doi = {10.1073/pnas.1801723115},
}

@article{Dong2021,
title = {Random circuit block-encoded matrix and a proposal of quantum LINPACK benchmark},
author = {Dong, Yulong and Lin, Lin},
journal = {Phys. Rev. A},
volume = {103},
issue = {6},
pages = {062412},
numpages = {21},
year = {2021},
publisher = {American Physical Society},
doi = {10.1103/PhysRevA.103.062412},
}

@inproceedings{Gilyen2018,
author = {Gily{\'{e}}n, Andr{\'{a}}s and Su, Yuan and Low, Guang Hao and Wiebe, Nathan},
booktitle = {Proceedings of the 51st Annual ACM SIGACT Symposium on Theory of Computing},
doi = {10.1145/3313276.3316366},
pages = {193--204},
title = {{Quantum singular value transformation and beyond: exponential improvements for quantum matrix arithmetics}},
year = {2019}
}

@article{Georges2025,
author = {Georges, Timothy N and Bothe, Marius and S{\"{u}}nderhauf, Christoph and Berntson, Bjorn K and Izs{\'{a}}k, R{\'{o}}bert and Ivanov, Aleksei V},
doi = {10.1038/s41534-025-00987-1},
journal = {npj Quantum Inf.},
number = {1},
pages = {55},
title = {{Quantum simulations of chemistry in first quantization with any basis set}},
volume = {11},
year = {2025}
}

@article{Lee2021,
title = {Even More Efficient Quantum Computations of Chemistry Through Tensor Hypercontraction},
author = {Lee, Joonho and Berry, Dominic W. and Gidney, Craig and Huggins, William J. and McClean, Jarrod R. and Wiebe, Nathan and Babbush, Ryan},
journal = {PRX Quantum},
volume = {2},
issue = {3},
pages = {030305},
numpages = {62},
year = {2021},
publisher = {American Physical Society},
doi = {10.1103/PRXQuantum.2.030305},
}

@misc{kuklinski2025,
title={Efficient block-encodings require structure},
author={Parker Kuklinski and Benjamin Rempfer and Justin Elenewski and Kevin Obenland},
year={2025},
eprint={2509.19667},
archivePrefix={arXiv},
}

@article{loaiza2024,
doi = {10.1088/2058-9565/add9c1},
year = {2025},
volume = {10},
number = {3},
pages = {035035},
author = {Loaiza, Ignacio and Sankar Brahmachari, Aritra and Izmaylov, Artur F},
title = {Majorana tensor decomposition: a unifying framework for decompositions of fermionic {Hamiltonians} to linear combination of unitaries},
journal = {Quantum Sci. Technol.}
}

@article{Low_2019,
title={{Hamiltonian} Simulation by Qubitization},
volume={3},
pages={163},
DOI={10.22331/q-2019-07-12-163},
journal={Quantum},
publisher={Verein zur Forderung des Open Access Publizierens in den Quantenwissenschaften},
author={Low, Guang Hao and Chuang, Isaac L.},
year={2019},
}

@article{Meister_2022,
title={Tailoring Term Truncations for Electronic Structure Calculations Using a Linear Combination of Unitaries},
volume={6},
pages={637},
doi={10.22331/q-2022-02-02-637},
journal={Quantum},
publisher={Verein zur Forderung des Open Access Publizierens in den Quantenwissenschaften},
author={Meister, Richard and Benjamin, Simon C. and Campbell, Earl T.},
year={2022},
}

@article{Motlagh2024,
title = {Generalized Quantum Signal Processing},
author = {Motlagh, Danial and Wiebe, Nathan},
journal = {PRX Quantum},
volume = {5},
issue = {2},
pages = {020368},
numpages = {16},
year = {2024},
publisher = {American Physical Society},
doi = {10.1103/PRXQuantum.5.020368},
}

@article{Nibbi_2024,
title={Block encoding of matrix product operators},
volume={110},
doi={10.1103/physreva.110.042427},
issue={4},
pages={042427},
journal={Phys. Rev. A},
publisher={American Physical Society (APS)},
author={Nibbi, Martina and Mendl, Christian B.},
year={2024},
}

@book{Nielsen_2011,
Author = {Michael A. Nielsen and Isaac L. Chuang},
Title = {Quantum Computation and Quantum Information: 10th Anniversary Edition},
Publisher = {Cambridge University Press},
Year = {2011},
Doi = {10.1017/CBO9780511976667}
}

@article{Nguyen2022,
doi = {10.22331/q-2022-12-13-876},
title = {Block-encoding dense and full-rank kernels using hierarchical matrices: applications in quantum numerical linear algebra},
author = {Nguyen, Quynh T. and Kiani, Bobak T. and Lloyd, Seth},
journal = {Quantum},
publisher = {{Verein zur F{\"{o}}rderung des Open Access Publizierens in den Quantenwissenschaften}},
volume = {6},
pages = {876},
year = {2022}
}

@article{Piroli2024,
title = {Approximating Many-Body Quantum States with Quantum Circuits and Measurements},
author = {Piroli, Lorenzo and Styliaris, Georgios and Cirac, J. Ignacio},
journal = {Phys. Rev. Lett.},
volume = {133},
issue = {23},
pages = {230401},
numpages = {7},
year = {2024},
publisher = {American Physical Society},
doi = {10.1103/PhysRevLett.133.230401},
}

@misc{simon2025,
title={Ladder Operator Block-Encoding},
author={William A. Simon and Carter M. Gustin and Kamil Serafin and Alexis Ralli and Gary R. Goldstein and Peter J. Love},
year={2025},
eprint={2503.11641},
archivePrefix={arXiv},
}

@article{Wan2021,
doi = {10.22331/q-2021-01-12-380},
title = {Exponentially faster implementations of {S}elect({H}) for fermionic {H}amiltonians},
author = {Wan, Kianna},
journal = {Quantum},
publisher = {{Verein zur F{\"{o}}rderung des Open Access Publizierens in den Quantenwissenschaften}},
volume = {5},
pages = {380},
year = {2021}
}

@misc{yu2024,
title={Efficient preparation of {Dicke} states},
author={Jeffery Yu and Sean R. Muleady and Yu-Xin Wang and Nathan Schine and Alexey V. Gorshkov and Andrew M. Childs},
year={2024},
eprint={2411.03428},
archivePrefix={arXiv},
}

@article{zindorf2025,
  title = {Efficient implementation of multicontrolled quantum gates},
  author = {Zindorf, Ben and Bose, Sougato},
  journal = {Phys. Rev. Appl.},
  volume = {24},
  issue = {4},
  pages = {044030},
  numpages = {30},
  year = {2025},
  month = {Oct},
  publisher = {American Physical Society},
  doi = {10.1103/8blx-nfcr},
  url = {https://link.aps.org/doi/10.1103/8blx-nfcr}
}

@misc{zindorf2025_lnn,
      title={Multi-Controlled Quantum Gates in Linear Nearest Neighbor}, 
      author={Ben Zindorf and Sougato Bose},
      year={2025},
      eprint={2506.00695},
      archivePrefix={arXiv},
      primaryClass={quant-ph},
      url={https://arxiv.org/abs/2506.00695}, 
}

@article{Liu2025,
title = {An efficient quantum circuit for block encoding a pairing {Hamiltonian}},
journal = {J. Comput. Sci.},
volume = {85},
pages = {102480},
year = {2025},
doi = {https://doi.org/10.1016/j.jocs.2024.102480},
author = {Diyi Liu and Weijie Du and Lin Lin and James P. Vary and Chao Yang},
}

@misc{Liu2025nbd,
title={Block encoding with low gate count for second-quantized {Hamiltonians}},
author={Diyi Liu and Shuchen Zhu and Guang Hao Low and Lin Lin and Chao Yang},
year={2025},
eprint={2510.08644},
archivePrefix={arXiv},
}

@article{Dicke_1954,
title = {Coherence in Spontaneous Radiation Processes},
author = {Dicke, R. H.},
journal = {Phys. Rev.},
volume = {93},
issue = {1},
pages = {99--110},
numpages = {0},
year = {1954},
publisher = {American Physical Society},
doi = {10.1103/PhysRev.93.99},
}

@misc{Farrell2025,
title={Digital quantum simulations of scattering in quantum field theories using {W} states},
author={Roland C. Farrell and Nikita A. Zemlevskiy and Marc Illa and John Preskill},
year={2025},
eprint={2505.03111},
archivePrefix={arXiv},
}

@Article{Babbush_2019,
author={Babbush, Ryan 
and Berry, Dominic W. 
and McClean, Jarrod R. 
and Neven, Hartmut},
title={Quantum simulation of chemistry with sublinear scaling in basis size},
journal={npj Quantum Inf.},
year={2019},
day={01},
volume={5},
number={1},
pages={92},
issn={2056-6387},
doi={10.1038/s41534-019-0199-y},
}

@article{Loaiza_2023,
doi = {10.1088/2058-9565/acd577},
year = {2023},
publisher = {IOP Publishing},
volume = {8},
number = {3},
pages = {035019},
author = {Loaiza, Ignacio and Khah, Alireza Marefat and Wiebe, Nathan and Izmaylov, Artur F},
title = {Reducing molecular electronic {Hamiltonian} simulation cost for linear combination of unitaries approaches},
journal = {Quantum Sci. Technol.}
}

@misc{Sze_2025,
title={{Hamiltonian} dynamics simulation using linear combination of unitaries on an ion trap quantum computer},
author={Michelle Wynne Sze and Yao Tang and Silas Dilkes and David Muñoz Ramo and Ross Duncan and Nathan Fitzpatrick},
year={2025},
eprint={2501.18515},
archivePrefix={arXiv},
}

@misc{Sze:2025vjh,
author = {Sze, Michelle Wynne and Manrique, David Zsolt and Ramo, David Mu{\~n}oz and Fitzpatrick, Nathan},
title = {Shorter width truncated {Taylor} series for {Hamiltonian} dynamics simulations},
eprint = {2511.09461},
archivePrefix = {arXiv},
year = {2025}
}

@article{Snderhauf2024,
author = {Christoph S\"underhauf and Earl Campbell and Joan Camps},
doi = {10.22331/q-2024-01-11-1226},
journal = {Quantum},
pages = {1226},
publisher = {Verein zur Forderung des Open Access Publizierens in den Quantenwissenschaften},
title = {Block-encoding structured matrices for data input in quantum computing},
volume = {8},
year = {2024}
}

@article{sanders2020,
title = {Compilation of Fault-Tolerant Quantum Heuristics for Combinatorial Optimization},
author = {Sanders, Yuval R. and Berry, Dominic W. and Costa, Pedro C.S. and Tessler, Louis W. and Wiebe, Nathan and Gidney, Craig and Neven, Hartmut and Babbush, Ryan},
journal = {PRX Quantum},
volume = {1},
issue = {2},
pages = {020312},
numpages = {70},
year = {2020},
publisher = {American Physical Society},
doi = {10.1103/PRXQuantum.1.020312},
}

@article{Kane2025,
doi = {10.22331/q-2025-05-15-1747},
title = {Block encoding bosons by signal processing},
author = {Kane, Christopher F. and Hariprakash, Siddharth and Modi, Neel S. and Kreshchuk, Michael and Bauer, Christian W},
journal = {Quantum},
publisher = {{Verein zur F{\"{o}}rderung des Open Access Publizierens in den Quantenwissenschaften}},
volume = {9},
pages = {1747},
year = {2025}
}

@misc{kuklinski2024,
title={{S-FABLE} and {LS-FABLE}: Fast approximate block-encoding algorithms for unstructured sparse matrices},
author={Parker Kuklinski and Benjamin Rempfer},
year={2024},
eprint={2401.04234},
archivePrefix={arXiv},
}

@article{Kikuchi2023,
author = {Kikuchi, Yuta and {Mc Keever}, Conor and Coopmans, Luuk and Lubasch, Michael and Benedetti, Marcello},
doi = {10.1038/s41534-023-00762-0},
journal = {npj Quantum Inf.},
number = {1},
pages = {93},
title = {{Realization of quantum signal processing on a noisy quantum computer}},
volume = {9},
year = {2023}
}

@misc{Rullkotter2025,
title={Resource-efficient variational block-encoding},
author={Leon Rullkötter and Sebastian Weber and Vamshi Mohan Katukuri and Christian Tutschku and Bharadwaj Chowdary Mummaneni},
year={2025},
eprint={2507.17658},
archivePrefix={arXiv},
}

@misc{schillo2026,
title={Block Encoding Linear Combinations of {Pauli} Strings Using the Stabilizer Formalism},
author={Niclas Schillo and Andreas Sturm and Rüdiger Quay},
year={2026},
eprint={2601.05740},
archivePrefix={arXiv},
}

@misc{schiffer2025,
title={Hardware-efficient quantum phase estimation via local control},
author={Benjamin F. Schiffer and Dominik S. Wild and Nishad Maskara and Mikhail D. Lukin and J. Ignacio Cirac},
year={2025},
eprint={2506.18765},
archivePrefix={arXiv},
}

@article{Somma_2019,
doi = {10.1088/1367-2630/ab5c60},
year = {2019},
publisher = {IOP Publishing},
volume = {21},
number = {12},
pages = {123025},
author = {Somma, Rolando D},
title = {Quantum eigenvalue estimation via time series analysis},
journal = {New J. Phys.}
}

@article{Lu2021,
title = {Algorithms for Quantum Simulation at Finite Energies},
author = {Lu, Sirui and Ba\~nuls, Mari Carmen and Cirac, J. Ignacio},
journal = {PRX Quantum},
volume = {2},
issue = {2},
pages = {020321},
numpages = {22},
year = {2021},
publisher = {American Physical Society},
doi = {10.1103/PRXQuantum.2.020321},
}

@article{Lin2022,
title = {Heisenberg-Limited Ground-State Energy Estimation for Early Fault-Tolerant Quantum Computers},
author = {Lin, Lin and Tong, Yu},
journal = {PRX Quantum},
volume = {3},
issue = {1},
pages = {010318},
numpages = {21},
year = {2022},
publisher = {American Physical Society},
doi = {10.1103/PRXQuantum.3.010318},
}

@article{Dutkiewicz2022,
doi = {10.22331/q-2022-10-06-830},
title = {Heisenberg-limited quantum phase estimation of multiple eigenvalues with few control qubits},
author = {Dutkiewicz, Alicja and Terhal, Barbara M. and O'Brien, Thomas E.},
journal = {{Quantum}},
publisher = {{Verein zur F{\"{o}}rderung des Open Access Publizierens in den Quantenwissenschaften}},
volume = {6},
pages = {830},
year = {2022}
}

@article{blunt2023,
title = {Statistical Phase Estimation and Error Mitigation on a Superconducting Quantum Processor},
author = {Blunt, Nick S. and Caune, Laura and Izs\'ak, R\'obert and Campbell, Earl T. and Holzmann, Nicole},
journal = {PRX Quantum},
volume = {4},
issue = {4},
pages = {040341},
numpages = {23},
year = {2023},
publisher = {American Physical Society},
doi = {10.1103/PRXQuantum.4.040341},
}

@article{Martyn2021,
  title = {Grand Unification of Quantum Algorithms},
  author = {Martyn, John M. and Rossi, Zane M. and Tan, Andrew K. and Chuang, Isaac L.},
  journal = {PRX Quantum},
  volume = {2},
  issue = {4},
  pages = {040203},
  numpages = {40},
  year = {2021},
  month = {Dec},
  publisher = {American Physical Society},
  doi = {10.1103/PRXQuantum.2.040203},
  url = {https://link.aps.org/doi/10.1103/PRXQuantum.2.040203}
}

\newpage
\onecolumngrid
\appendix
\labelformat{subsection}{\thesection#1}

%%%%%%%%%%%%%%%%%%%%%%%%%%%%%%%%%%%%%%%%%%%%%
\section{Proof of \cref{lemma:product}} \label{app: proof lemma}
%%%%%%%%%%%%%%%%%%%%%%%%%%%%%%%%%%%%%%%%%%%%%
Applying the $\PRindex{\Mat{1}}$ and $\PRindex{\Mat{2}}$ oracles to the initial state $\ket{0}_{2n} \ket{0}_{2n} \ket{\varphi}$ yields: 
\begin{align}
 \PRindex{\Mat{1}}\ket{0}_{\n}\ket{0}_{\n}&\otimes \PRindex{\Mat{2}}\ket{0}_{\n}\ket{0}_{\n}\ket{\varphi} \nonumber \\
&=\left(\sum_{\ix,\jz=0}^{2^{\n}-1}e^{i(\arg(\alpha_{\ix\jz}))}\sqrt{\abs{\alpha_{\ix\jz}}}\ket{\ix}\ket{\jz}\right )\otimes\left(\sum_{p,q=0}^{2^{\n}-1}e^{i(\arg(\beta_{pq}))}\sqrt{\abs{\beta_{pq}}}\ket{p}\ket{q}\right ) \ket{\varphi} ,
\end{align}
where the tensor product can be expanded as:
\begin{equation}
\sum_{p,q,i,j= 0}^{2^{\n}-1}e^{i(\arg(\alpha_{\ix\jz})+\arg(\beta_{pq}))}\sqrt{\abs{\alpha_{\ix\jz}\beta_{pq}}}\ket{\ix}\ket{\jz}\ket{p}\ket{q}\ket{\varphi}.
\end{equation}
After the first $n$ \cnot{}s and $n$ \cz{}s of the \select{} section, the state becomes:
\begin{equation}
\sum_{p,q,i,j= 0}^{2^{\n}-1}e^{i(\arg(\alpha_{\ix\jz})+\arg(\beta_{pq}))}\sqrt{\abs{\alpha_{\ix\jz}\beta_{pq}}}\ket{\ix}\ket{\jz}\ket{p}\ket{q} \otimes\left ( \bigotimes_{\ell=0}^{n-1} Z^{\texttt{j}_\ell}X^{\texttt{i}_\ell} \right)\ket{\varphi},
\end{equation}
and subsequently, after the other $n$ \cnot{}s and $n$ \cz{}s, we get:
\begin{equation}
\sum_{p,q,i,j= 0}^{2^{\n}-1}e^{i(\arg(\alpha_{\ix\jz})+\arg(\beta_{pq}))}\sqrt{\abs{\alpha_{\ix\jz}\beta_{pq}}}\ket{\ix}\ket{\jz}\ket{p}\ket{q} \otimes \left ( \bigotimes_{\ell=0}^{n-1} Z^{\texttt{q}_\ell}X^{\texttt{p}_\ell}Z^{\texttt{\jz}_\ell}X^{\texttt{\ix}_\ell} \right)\ket{\varphi}.
\end{equation}
Finally, we apply the left state preparation oracles $\PLindexdag{\Mat{1}}$ and $\PLindexdag{\Mat{2}}$ and we project onto the ancilla qubits:
\begin{equation}
 _{4n}\bra{0}\PLindexdag{\Mat{2}}\otimes \PLindexdag{\Mat{1}} = \sum_{p,q,i,j= 0}^{2^{\n}-1}\sqrt{\abs{\alpha_{\ix\jz}\beta_{pq}}}\bra{\ix}\bra{\jz}\bra{p}\bra{q}\, ,
\end{equation}
so that we get the final (non-normalized) state:
\begin{equation}
 \sum_{i,j,p,q = 0}^{2^n-1}\alpha_{ij}\beta_{pq}\left(\bigotimes_{\ell= 0}^{n-1}Z^{\texttt{q}_\ell}X^{\texttt{p}_\ell}Z^{\texttt{j}_\ell}X^{\texttt{i}_\ell}\right )\ket{\varphi} = C\ket{\varphi} , %{\norm{C\ket{\varphi}}_2}.
\end{equation}
for $C$ defined in \cref{product of 2 be}, which completes the proof.

%%%%%%%%%%%%%%%%%%%%%%%%%%%%%%%%%%%
\section{Proof of \cref{th:matrix-poly}}
\label{app:proof_matrix_polynomial}
%%%%%%%%%%%%%%%%%%%%%%%%%%%%%%%%%%%
We prove that the circuit in \cref{fig:poly_general_pr} correctly implements the matrix polynomial defined in \cref{eq:poly_transform_H}.
We begin by analyzing the action of \polyR{}, defined in \cref{eq:circuit_poly_r}, on the initial state $\ket{0}_d \ket{0}_{2dn}\ket{\varphi}$.
Applying \polyR{} yields
\begin{equation}
 \ket{0}_{d}\ket{0}_{2dn}\ket{\varphi}
 \xrightarrow{\polyR{}}
 \frac{1}{\sqrt{\polynormfact}}
 \sum_{k=0}^{d}
 e^{i\sum_{j=0}^{k-1}\phi_j}
 \sqrt{\abs{a_k}\,\normfact^{k}}\,
 \ket{k_u}\ket{0}_{2dn}\ket{\varphi}.
\end{equation}
By definition of the unary-encoded states $\ket{k_u}$ in \cref{eq:unary-ku}, each such state contains exactly $k$ leading ones.
Consequently, $\ket{k_u}$ activates precisely the first $k$ \PR{} gates in the circuit.
Applying \cref{corol:powers} to the normalized matrix $\frac{\Ham}{\normfact}$, we conclude that, after applying all \PR{}, \select{}, and \PLdag{} oracles, followed by post-selection on the all-zero outcome of the $2dn$ ancilla qubits on which \PR{} and \PLdag{} act, the resulting state is: 
\begin{equation}
 \frac{1}{\sqrt{\polynormfact}}\sum_{k=0}^{d}e^{i\sum_{j=0}^{k-1} \phi_j}
 \sqrt{\abs{a_k}\,\normfact^{k}}\;
 \ket{k_u}
 \frac{\Ham^{k}}{\normfact^{k}}
 \ket{\varphi}.
\end{equation}
Next, by applying $\polyL^\dagger$, defined by \cref{eq:circuit_poly_l}, as well as measuring and post-selecting $\ket{0}$ on the remaining $d$ ancilla qubits, we obtain the (non-normalized) state:
\begin{equation}\label{eq:proof_theorem3_1}
 \sum_{k=0}^{d} e^{i\sum_{j=0}^{k-1}\phi_j}\abs{a_k}\,\Ham^k\ket{\varphi}\, .
\end{equation}
%where $v = \sum_{k=0}^{d} e^{i\sum_{j=0}^{k-1}\phi_j}\abs{a_k}\,\Ham^k\ket{\varphi}$. 
Now, we just have to observe that from \cref{eq:POLY-phik} we have:
\begin{equation}
 \sum_{j=0}^{k-1} \phi_j = \arg(a_k) - \arg(a_0),
\end{equation}
so that the final state in \cref{eq:proof_theorem3_1} becomes:
\begin{equation}
e^{-i\arg(a_0)}
%{\norm{v}_2}
\sum_{k=0}^{d} e^{i\arg(a_k)}\abs{a_k}\,\Ham^k\ket{\varphi}\, =e^{-i\arg(a_0)}
%{\norm{v}_2} 
p_{d}(\Ham)\ket{\varphi}.
\end{equation}
The resulting transformation is exactly the desired one, up to an irrelevant global phase $e^{-i\arg(a_0)}$, which completes the proof.

%%%%%%%%%%%%%%%%%%%%%%%%%%%%%%%%%%%%%%%%
\section{\polyfoqcs{} vs \qsvtlcu{} vs \qsvtfoqcs{}}
\label{sec:qsvt}
%%%%%%%%%%%%%%%%%%%%%%%%%%%%%%%%%%%%%%%%
In this section we provide a detailed analysis of the resource requirements of\polyfoqcs{},
\qsvtlcu{} and \qsvtfoqcs{} summarized in \cref{tab:qsvt_comparison}.

%%%%%%%%%%%%%%%%%%%%%%%%%%%%%%%%%%%%%%%%%%%%%
\subsection{Matrix function}
\label{sec:qsvt_funm}
%%%%%%%%%%%%%%%%%%%%%%%%%%%%%%%%%%%%%%%%%%%%%
The first difference lies in the matrix functions themselves.
QSVT applies a degree-$d$ polynomial transformation $f_d$ to the singular values of
the block-encoded, and therefore already normalized, Hamiltonian $\Ham$
\cite{Gilyen2018,Martyn2021}:
\begin{equation}\label{eq:qsvt_def}
    \begin{pmatrix}
        \Ham/\normfact & * \\
        * & *
    \end{pmatrix}
    \xrightarrow{\text{QSVT}}
    \begin{pmatrix}
        f_d(\Ham/\normfact)/\qsvtnormfact & * \\
        * & *
    \end{pmatrix} ,
\end{equation}
where $f_d$ must have definite parity, i.e.\ be either even or odd, and the
normalization factor $\qsvtnormfact$ is equal to the infinity-norm of
$f_d$:
\begin{equation}\label{def:qsvt_norm}
    \qsvtnormfact = \norm{f_d}_\infty = \sup_{x\in [-1,1]}\abs{f_d(x)} .
\end{equation}
\polyfoqcs{}, on the other hand, directly block encodes the matrix polynomial
$p_d(\Ham)$:
\begin{equation}
    \begin{pmatrix}
        p_d(\Ham)/\polynormfact & * \\
        * & *
    \end{pmatrix} ,
\end{equation}
with normalization factor $\polynormfact$ defined in \cref{eq:def_W}.

\polyfoqcs{} is thus less restrictive, as it is not limited to functions of definite parity.
The two normalization factors, however, yield different success probabilities, which we detail in \cref{sec:qsvt_prob}, and the two circuit constructions differ in depth and qubit count, which we discuss in \cref{sec:qsvt_circuit}.

%%%%%%%%%%%%%%%%%%%%%%%%%%%%%%%%%%%%%%%%%%%%%
\subsection{Success probability comparison}
\label{sec:qsvt_prob}
%%%%%%%%%%%%%%%%%%%%%%%%%%%%%%%%%%%%%%%%%%%%%
Here we compare the success probabilities of the QSVT and matrix polynomial
algorithms.
Since \foxlcu{} and textbook LCU block encode $\Ham$ with the same normalization
factor $\normfact$ \cite{foqcs-lcu-arxiv}, the results derived below for
\qsvtfoqcs{} apply unchanged to \qsvtlcu{}.

QSVT realizes a definite-parity function $f_d$ of the normalized matrix
$\Ham/\normfact$, as in \cref{eq:qsvt_def}, whereas \polyfoqcs{} block encodes a
generic matrix polynomial $p_d(\Ham)$ with the global normalization factor
$\polynormfact$ of \cref{eq:def_W}.
For a fair comparison we restrict $p_d$ to the same class, i.e.\ to polynomials
of definite parity,
\begin{equation}
    p_d(\Ham) = \sum_{k=0}^{d} a_k \Ham^k ,
    \qquad a_k = 0 \ \text{ for } \ k \not\equiv d \ (\mathrm{mod}\ 2) ,
\end{equation}
and we absorb the normalization factor $\normfact$ into the coefficients,
\begin{equation}
    \bar{a}_k = a_k \normfact^k ,
\end{equation}
so that the rescaled polynomial
\begin{equation}
    f_d(x) = \sum_{k=0}^{d} \bar{a}_k x^k , \qquad x \in [-1,1] ,
\end{equation}
satisfies
\begin{equation}\label{eq:def_coefficients_qsvt}
    p_d(\Ham) = f_d(\Ham/\normfact)
    = \bar{a}_0 + \bar{a}_1 \frac{\Ham}{\normfact}
    + \bar{a}_2 \frac{\Ham^2}{\normfact^2} + \cdots
    + \bar{a}_d \frac{\Ham^d}{\normfact^d} .
\end{equation}
%Both methods therefore realize the same operator $f_d(\Ham/\normfact)$ and differ only in the overall normalization constant.
This is precisely the function $f_d$ implemented by QSVT in \cref{eq:qsvt_def}. As a result, the two methods implement the same operator but with different normalization factors.
Specifically, QSVT's normalization factor $\qsvtnormfact$ is the
infinity-norm of $f_d$, as defined in \cref{def:qsvt_norm}, while the
matrix polynomial \foxlcu{} block encoding divides $f_d(\Ham/\normfact)$ by
the $1$-norm of its coefficient vector:
\begin{equation}\label{eq:def_vecnorm}
    \polynormfact = \sum_{k=0}^d \abs{\bar{a}_k} = \vecnorm{f_d} ,
\end{equation}

As a result, the success probabilities for \qsvtfoqcs{} and \polyfoqcs{} are, respectively:
\begin{align}\label{eq:p_qsvt}
    p_{\text{QSVT}} &= \frac{\norm{f_d(\Ham/\normfact)\ket{\phi}}^2}{\qsvtnormfact^2}
    = \frac{\norm{f_d(\Ham/\normfact)\ket{\phi}}^2}{\norm{f_d}_\infty^2} , \\
    p_{\text{POLY}} &= \frac{\norm{p_d(\Ham)\ket{\phi}}^2}{\polynormfact^2}
    = \frac{\norm{f_d(\Ham/\normfact)\ket{\phi}}^2}{\vecnorm{f_d}^2} . \label{eq:p_poly}
\end{align}
For any $x \in [-1,1]$ we have $\abs{f_d(x)} \leq \sum_k \abs{\bar{a}_k}$, hence $\norm{f_d}_\infty \leq \vecnorm{f_d}$ and:
\begin{equation}
    p_{\text{QSVT}} \geq p_{\text{POLY}} .
\end{equation}

Note that the two norms can differ by a large factor.
This is the case, for example, when $f_d$ approximates the time evolution operator $e^{-ixt}$.
Implementing the sine and cosine components with QSVT and recombining them with a small LCU circuit \cite{Martyn2021} gives $\norm{f_d}_\infty = \bigO{1}$, and therefore $p_{\text{QSVT}} = \bigO{1}$.
The $1$-norm of the coefficient vector instead grows as the truncated series:
\begin{equation}
    \polynormfact = \sum_{k=0}^d \frac{\vert\normfact t\vert^k}{k!} \approx e^{\normfact t}.
\end{equation}
Since $f_d(\Ham/\normfact)$ is approximately unitary, the numerator of
\cref{eq:p_poly} stays close to unity, and the success probability
$p_{\text{POLY}} \approx e^{-2\normfact t}$ decays exponentially with the
simulated time.

Several techniques can mitigate this.
One can, for instance, choose a shorter timestep $\Delta t$ such that $\polynormfact \approx 2$ and apply Oblivious Amplitude Amplification (OAA) \cite{Berry_2015}.
This restores a constant success probability, but the procedure must then be repeated sequentially for $D = t/\Delta t$ steps, thereby increasing the circuit depth and canceling the depth advantage of the matrix polynomial approach.

As a result, whenever the target function is even or odd, or can be written as a
short combination of an even and an odd part, \qsvtfoqcs{} has a higher success
probability than \polyfoqcs{}.
The method implemented in this work is, however, more general and generates
shallower circuits, as explained in the following section.

%%%%%%%%%%%%%%%%%%%%%%%%%%%%%%%%%%%%%%%%
\subsection{Circuit constructions}
\label{sec:qsvt_circuit}
%%%%%%%%%%%%%%%%%%%%%%%%%%%%%%%%%%%%%%%%
The QSVT algorithm interleaves queries to a block encoding of $\Ham$ (\foxlcu{}
in our case) with signal processing operators:
\begin{equation}
    \begin{myqcircuitr}{0.6}
        \lstick{\ket{0}} & \qw & \qw & \multigate{1}{\sigp{\theta_0}} & \qw & \multigate{1}{\sigp{\theta_1}} & \qw & \qw && \cdots &&& \multigate{1}{\sigp{\theta_d}} & \qw \\
        \lstick{\ket{0}} & \qw{/\strut^{2n}} & \qw & \ghost{\sigp{\theta_1}} & \multigate{1}{\Ufox} & \ghost{\sigp{\theta_0}} & \multigate{1}{\Ufox^\dagger} & \qw && \cdots &&& \ghost{\sigp{\theta_d}} & \meter{} & \cw & 0\\
        \lstick{\ket{\phi}}{} & \qw{/\strut^{n}} & \qw & \qw & \ghost{\Ufox} & \qw & \ghost{\Ufox^\dagger} & \qw && \cdots &&& \qw & \qw &&&& \frac{f_d(\Ham/\normfact)\ket{\phi}}{\norm{f_d(\Ham/\normfact)\ket{\phi}}}
    \end{myqcircuitr}
\end{equation}
where the signal processing operators $\sigp{\theta}$ are defined as
\begin{equation}
    \sigp{\theta} = e^{i\theta(2\ket{0}\bra{0}_{2n} - \Ip)} ,
\end{equation}
with angles $\theta$ precomputed classically for the desired transformation $f_d$.
A standard implementation of this operator requires one additional clean ancilla
qubit initialized in the state $\ket{0}$, which is reused across all signal
processing steps \cite{Martyn2021}:
\begin{equation}
    \begin{myqcircuit}
        &&&\lstick{\ket{0}} & \targ{} & \gate{e^{-i\theta\Zp}} & \targ{} & \qw \\
        &&&\lstick{\ket{0}} & \ctrlo{-1} & \qw & \ctrlo{-1} & \qw \\
        &&&\lstick{\ket{0}} & \ctrlo{-1} & \qw & \ctrlo{-1} & \qw \\
        &&&& \vbox{\baselineskip4pt\lineskiplimit0pt\hbox{.}\hbox{.}\hbox{.}} \ar@{-}[u];[]+<0em,0.8em> && \vbox{\baselineskip4pt\lineskiplimit0pt\hbox{.}\hbox{.}\hbox{.}} \ar@{-}[u];[]+<0em,0.8em> & \\
        &&&\lstick{\ket{0}} & \ctrlo{0} \ar@{-}[];[u]+<0em,-0.8em> & \qw & \ctrlo{0} \ar@{-}[];[u]+<0em,-0.8em> & \qw \inputgroupv{2}{5}{1em}{2.5em}{2n}
    \end{myqcircuit}
\end{equation}
Decomposing this circuit into native one- and two-qubit gates requires
decomposing the multi-controlled Toffoli gates.
%By adding one more dirty ancilla, 
Linear \cnot{} and T-gate counts and depths can be achieved for either all-to-all or square-grid connectivity by adding one dirty ancilla~\cite{zindorf2025,zindorf2025_lnn}.
This ancilla acts as a catalyst; after the multi-controlled gate, it returns to its original state.  Thus the dirty ancilla can be borrowed from the system qubits, which are idle during the signal processing step.
%Since the system qubits are idle during the signal processing step, one of them can be borrowed as the required dirty ancilla.

The \foxlcu{} block encoding employs $2n$ ancillae, so each signal processing
operator $\sigp{\theta}$ consists of two $2n$-controlled Toffoli gates with \cnot{} count and
depth scaling as $\bigO{n}$.
For spin Hamiltonians considered in this paper, a single \foxlcu{} instance
also has \cnot{} count and depth linear in the system size $n$, although with a
much smaller prefactor.
The \qsvtfoqcs{} circuit hence has \cnot{} count and depth scaling as
$\bigO{n\cdot d}$.
%, at the cost of a single additional qubit for the signal processing operator.

For standard \qsvtlcu{}, the textbook LCU block encoding of the same Hamiltonians uses $\bigO{\log n}$ ancillae and has \cnot{} count and depth
scaling as $\bigO{n\log n}$ \cite{foqcs-lcu-arxiv}.
The signal processing operator is then negligible compared to the block encoding
itself, and the total \cnot{} depth of \qsvtlcu{} scales as $\bigO{n\log n \cdot d}$.
Note that the total qubit count scales as $\bigO{n}$ for both \qsvtfoqcs{} and
\qsvtlcu{}.

In contrast to the QSVT approaches, which require only $\bigO{n}$ qubits,
\polyfoqcs{} uses a larger number, $2dn+d+n$ in total, as shown in
\cref{fig:poly_pr_ab}.
In exchange, its overall depth scales only as $\bigO{n+d}$, instead of
$\bigO{n \cdot d}$, and with very small prefactors for the models considered in
this work, as shown in
\cref{tab:resource_estimation_xyz,tab:resource_estimation_xxz,tab:resource_estimation_ising}.
Among the methods compared here, it is therefore the best suited to the
late-NISQ era, in which circuit depth, rather than qubit count, is the limiting
resource.
These results are summarized in \cref{tab:qsvt_comparison}.

%%%%%%%%%%%%%%%%%%%%%%%%%%%%%%%%%%%%%%%%%%%%%%%%%%
\section{Hadamard test with or without post-selection}
\label{app: Hadamard test}
%%%%%%%%%%%%%%%%%%%%%%%%%%%%%%%%%%%%%%%%%%%%%%%%%%%
% \begin{figure}
%     \centering
%     \[
%     \begin{myqcircuit}
%         \lstick{\ket{0}} 
%         & \gate{H} & \ctrl{1} & \gate{H} & \qw
%         &  \\
%         \lstick{\ket{0}} 
%         & \qw{/\strut^{n}} & \multigate{2}{\Ufox} & \qw & \qw & \\
%         \lstick{\ket{0}} 
%         & \qw{/\strut^{n}} & \ghost{\Ufox} & \qw & \qw & \\
%         \lstick{\ket{\varphi}} 
%         & \qw{/\strut^{n}} & \ghost{\Ufox} & \qw & \qw &
%     \end{myqcircuit}
%     \]
%     \caption{Hadamard circuit.}
%     \label{fig:hadamard_test}
% \end{figure}

In this section, we establish that performing a Hadamard test on a block encoding circuit does not require post-selection on the ancillae.
More specifically, we consider the Hadamard test circuit shown in \cref{eq:hadamard_block_encoding} for a generic block encoded matrix $\Ham$ and apply it to the initial state $\ket{0}\ket{0}_{2n}\ket{\varphi}$. 
The encoded matrix's expectation value is obtained by computing the probabilities of measuring $0$ and $1$ on the first qubit, and we show how these quantities change with and without post-selection on the block-encoding ancillae.
Finally, we compare the two protocols in terms of accuracy and number of shots.

Note that our derivations extend directly to the Loschmidt echo of \cref{eq:def_loschmidt_circuit}. Compared with the standard Hadamard test in \cref{eq:hadamard_block_encoding}, the only differences are that the encoded matrix is $p_d(\Ham)$ instead of $\Ham$, and the normalization factor is $\polynormfact$ instead of $\normfact$.

\subsection{Hadamard test without post-selection}
We start from the quantum state immediately before measurement:
\begin{equation}
\label{eq:state after hadamard}
    \frac{1}{2}\ket{0} \left(\ket{0}_{2n}\ket{\varphi}+\Ufox\ket{0}_{2n}\ket{\varphi}\right) + \frac{1}{2}\ket{1} \left(\ket{0}_{2n}\ket{\varphi}-\Ufox\ket{0}_{2n}\ket{\varphi}\right).
\end{equation}
If we only measure the Hadamard ancilla, we obtain the probabilities:
\begin{align}
p(0)&=\frac{1}{4} \left\|\ket{0}_{2n}\ket{\varphi}+\Ufox\ket{0}_{2n}\ket{\varphi}\right\|^{2} \\[1ex]
p(1)& = \frac{1}{4} \left\|\ket{0}_{2n}\ket{\varphi}- \Ufox\ket{0}_{2n}\ket{\varphi}
\right\|^{2}.
\end{align}
Expanding $p(0)$ and using $\Ufox^\dagger \Ufox = \Ip$, we obtain:
\begin{align}
p(0)&=
\frac{1}{4}
\Big(\!
\bra{\varphi}\braz{2n} + \bra{\varphi}\braz{2n}\Ufox^\dagger
\Big)%
\Big(\!
\ket{0}_{2n}\ket{\varphi}
+\Ufox\ket{0}_{2n}\ket{\varphi}
\!\Big)
\\
&=
\frac{1}{4}\Big( 2+ \bra{\varphi}\braz{2n}\Ufox\ket{0}_{2n}\ket{\varphi}
+ \bra{\varphi}\braz{2n}\Ufox^\dagger
\ket{0}_{2n}\ket{\varphi}
\!\Big) \\
&=
\frac{1}{2}
+
\frac{1}{4}
\bra{\varphi}\braz{2n}\Ufox\ket{0}_{2n}\ket{\varphi}
+
\frac{1}{4}
\bra{\varphi}\braz{2n}\Ufox^\dagger
\ket{0}_{2n}\ket{\varphi}
\\
&=
\frac{1}{2}
+
\frac{1}{2}
\operatorname{Re}
\bra{\varphi}\braz{2n}\Ufox
\ket{0}_{2n}\ket{\varphi}
.
\end{align}
By the definition of $\Ufox$, we have that 
$\bra{\varphi}\braz{2n}\Ufox
\ket{0}_{2n}\ket{\varphi} = \bra{\varphi}\frac{\Ham}{\normfact}\ket{\varphi}$, which results in:
\begin{equation}
    \label{eq: def p(0)}
       p(0)= \frac{1}{2} \Bigl(1+\text{Re} \bra{\varphi}\frac{\Ham}{\normfact}\ket{\varphi} \Bigr) .
\end{equation}
The same argument applied to $p(1)$ gives:
\begin{equation}
    \label{eq: def p(1)}
    p(1) = \frac{1}{2} \Bigl(1- \text{Re}\bra{\varphi}\frac{\Ham}{\normfact}\ket{\varphi} \Bigr) .
\end{equation}
By combining these results, we get:
\begin{equation}\label{eq:hadamard_result_no_postselect}
    \real{\braket{\varphi|\frac{\Ham}{\normfact}|\varphi}} = p(0)-p(1) \approx \frac{N_0 - N_1}{N_\text{tot}} ,
\end{equation}
where $N_\text{tot}$ is the total number of shots and $N_0$ and $N_1$ are the numbers of $0$ and $1$ outcomes on the Hadamard ancilla, respectively.
The approximation is due to the statistical noise from the finite number of shots, and becomes exact in the limit $N_\text{tot}\to\infty$.
Similarly, measuring in the $Y$-basis yields $\imag \braket{\varphi|\frac{\Ham}{\normfact}|\varphi}$. 

To quantify the statistical cost of this estimator, we assign to each shot a score $s_X=+1$ if the Hadamard ancilla is measured in $0$ and $s_X=-1$ if it is measured in $1$, so that the estimator in \cref{eq:hadamard_result_no_postselect} is the average of $s_X$ over the $N_\text{tot}$ shots.
Denoting $a \equiv \real{\braket{\varphi|\frac{\Ham}{\normfact}|\varphi}}$, the mean of the score is $\mathbb{E}[s_X] = p(0)-p(1) = a$.
Since $s_X^2=1$ in every shot, the variance per shot is
\begin{equation}\label{eq:variance_no_postselect}
    \mathrm{Var}(s_X) = \mathbb{E}[s_X^2]-\mathbb{E}[s_X]^2 = 1-a^2 .
\end{equation}
The statistical error of an average over $N_\text{tot}$ independent shots, quantified by its standard deviation, is $\sqrt{\mathrm{Var}(s_X)/N_\text{tot}}$.
Since $\real{\braket{\varphi|\Ham|\varphi}} = \normfact a$, this error is multiplied by $\normfact$, and reaching a standard deviation $\varepsilon$ on $\real{\braket{\varphi|\Ham|\varphi}}$ requires
\begin{equation}\label{eq:shots_no_postselect}
    N_\text{tot}^{\text{no-PS}} \propto \frac{\normfact^2\,(1-a^2)}{\varepsilon^2} .
\end{equation}
The same holds for the imaginary part, with $a$ replaced by $\imag{\braket{\varphi|\frac{\Ham}{\normfact}|\varphi}}$.

\subsection{Hadamard test with post-selection}
If we include post-selection on the block encoding ancillae, the expression for $p(0)-p(1)$ changes significantly. We start from the same pre-measurement state as in \cref{eq:state after hadamard}. Measuring the block-encoding ancillae yields the favorable outcome $0^{2n}$ with probability:
 \begin{align}\label{eq:def_gamma}
     p(\text{success}) &= \left\Vert \braz{2n} \left[\frac{1}{2}\ket{0} \left(\ket{0}_{2n}\ket{\varphi}+\Ufox\ket{0}_{2n}\ket{\varphi}\right) + \frac{1}{2}\ket{1} \left(\ket{0}_{2n}\ket{\varphi}-\Ufox\ket{0}_{2n}\ket{\varphi}\right)\right] \right\Vert^2 \nonumber\\
     &= \frac{1}{4} \left\Vert  \ket{0}\left(\Ip+\frac{\Ham}{\normfact}\right)\ket{\varphi} + \ket{1}\left(\Ip-\frac{\Ham}{\normfact}\right)\ket{\varphi}\right\Vert^2 \nonumber\\
     %&= \left\Vert \frac{1}{\sqrt{2}}\ket{+}\ket{\varphi} + \frac{1}{\sqrt{2}}\ket{-}\braz{2n} \Ufox \ket{0}_{2n}\ket{\varphi} \right\Vert^2 \nonumber \\
     %&= \frac{1}{2} \left\Vert \ket{+}\ket{\varphi} + \ket{-} \frac{\Ham}{\normfact} \ket{\varphi} \right\Vert^2 \nonumber \\
     &= \frac{1}{2}\left( 1 + \bra{\varphi}\frac{\Ham^\dagger \Ham} {\normfact^2}\ket{\varphi}\right) \equiv \gamma ,
 \end{align}
By defining:
\begin{equation}
\hat{p} \equiv \bra{\varphi}\frac{\Ham^\dagger \Ham}{\normfact^2}\ket{\varphi},
\end{equation}
i.e., the success probability of the standard block-encoding circuit without the Hadamard test, we obtain:
\begin{equation} \label{eq:estimator_gamma}
    \gamma = \frac{1}{2}\left(1 + \hat{p}\right) \approx \frac{N_\text{success}}{N_\text{tot}},
\end{equation}
where $N_\text{success}$ counts the number of shots for which the block-encoding ancillae are measured in $0^{2n}$. Upon successful post-selection, the state in \cref{eq:state after hadamard} collapses into:
\begin{equation}
\frac{1}{2\sqrt{\gamma}} \ket{0}\bigl(\ket{\varphi}+\frac{\Ham}{\normfact}\ket{\varphi} \bigr) + \frac{1}{2\sqrt{\gamma}} \ket{1}\bigl(\ket{\varphi}-\frac{\Ham}{\normfact}\ket{\varphi} \bigr). 
\end{equation}

The probabilities of measuring $0$ and $1$ on the Hadamard ancilla, conditioned on successful post-selection of the block-encoding ancillae, are therefore:
\begin{align}
    \label{eq:def_p0_conditioned}
    p(0|\text{success}) &= \frac{1}{4\gamma}\left(1+ 2\,\text{Re}\bra{\varphi}\frac{\Ham}{\normfact}\ket{\varphi}+\bra{\varphi}\frac{\Ham^\dagger \Ham }{\normfact^2}\ket{\varphi}\right), \\
    \label{eq:def_p1_conditioned}
    p(1|\text{success}) &= \frac{1}{4\gamma}\left(1- 2\,\text{Re} \bra{\varphi}\frac{\Ham}{\normfact}\ket{\varphi}+\bra{\varphi}\frac{\Ham^\dagger \Ham }{\normfact^2}\ket{\varphi}\right),
\end{align}
which gives:
\begin{equation}\label{eq:hadamard_result_postselection}
 \text{Re}\bra{\varphi}\frac{\Ham}{\normfact}\ket{\varphi}\cdot \frac{1}{\gamma} = p(0|\text{success})- p(1|\text{success}) .
\end{equation}
Similarly, measuring in the $Y$-basis yields $\imag \braket{\varphi|\frac{\Ham}{\normfact}|\varphi}\cdot \frac{1}{\gamma}$. 
By using the joint probabilities:
\begin{align}
    \label{eq:def_p0_joint}
    p(0\&
\text{success})& =p(0|\text{success})\cdot\gamma \approx \frac{N_{0\&\text{success}}}{N_\text{tot}} \\
    \label{eq:def_p1_joint}
    p(1\&\text{success})&=p(1|\text{success})\cdot\gamma \approx \frac{N_{1\&\text{success}}}{N_\text{tot}} ,
\end{align}
we obtain the estimator:
\begin{equation}\label{eq:hadamard_result_postselection_final}
     \text{Re}\bra{\varphi}\frac{\Ham}{\normfact}\ket{\varphi} = p(0\&\text{success})-p(1\&\text{success}) \approx \frac{N_{0\&\text{success}}-N_{1\&\text{success}}}{N_\text{tot}} .
\end{equation}

To quantify the statistical cost of this estimator, we proceed as in the case without post-selection and assign to each shot a score.
In particular, we have $s_X=+1$ if the Hadamard ancilla is measured in $0$ and the post-selection succeeds, $s_X=-1$ if the Hadamard ancilla is measured in $1$ and the post-selection succeeds, and $s_X=0$ if the post-selection fails.
The estimator in \cref{eq:hadamard_result_postselection_final} is then the average of $s_X$ over all $N_\text{tot}$ shots, including the failed ones, and its mean is $\mathbb{E}[s_X] = p(0\&\text{success})-p(1\&\text{success}) = a$.
Since $s_X^2=1$ when the post-selection succeeds and $s_X^2=0$ otherwise, we have $\mathbb{E}[s_X^2] = \gamma \approx N_\text{success}/N_\text{tot} $, and the variance per shot is:
\begin{equation}\label{eq:variance_postselect}
    \mathrm{Var}(s_X) = \mathbb{E}[s_X^2]-\mathbb{E}[s_X]^2 = \gamma-a^2 .
\end{equation}
Reaching a standard deviation error $\varepsilon$ on $\real{\braket{\varphi|\Ham|\varphi}}$ then requires:
\begin{equation}\label{eq:shots_postselect}
    N_\text{tot}^{\text{PS}} \propto \frac{\normfact^2\,(\gamma-a^2)}{\varepsilon^2} .
\end{equation}

\subsection{Comparison of the Hadamard test with and without post-selection}

The two protocols run the same circuit of \cref{eq:hadamard_block_encoding} and estimate the same quantity, $\real{\braket{\varphi|\frac{\Ham}{\normfact}|\varphi}}$, or the imaginary part when measuring in the $Y$ basis.

Thus the estimators in \cref{eq:hadamard_result_postselection_final} and \cref{eq:hadamard_result_no_postselect} must converge to the same value in the limit of infinite shots, but this is not immediately obvious from the expressions.  To demonstrate their equivalence, we compute the joint probabilities of measuring $0$ and $1$ on the Hadamard ancilla when post-selection fails.
Since each shot either succeeds or fails, we have:
\begin{align}
    p(0\&\text{fail}) &= p(0) - p(0\&\text{success}) \nonumber\\[1ex]
    &= \frac{1+a}{2} - \frac{1+2a+\hat{p}}{4} \nonumber\\
    &= \frac{1-\hat{p}}{4} ,
\end{align}
where we used \cref{eq: def p(0),eq:def_p0_conditioned,eq:def_p0_joint}. Similarly,
\begin{align}
    p(1\&\text{fail}) &= p(1) - p(1\&\text{success}) \nonumber\\[1ex]
    &= \frac{1-a}{2} - \frac{1-2a+\hat{p}}{4} \nonumber\\
    &= \frac{1-\hat{p}}{4} ,
\end{align}
where we used \cref{eq: def p(1),eq:def_p1_conditioned,eq:def_p1_joint}.
The two failure probabilities are equal and independent of $a$: when the post-selection fails, the Hadamard ancilla yields $0$ and $1$ with equal probability.
As a consequence, they cancel out in the difference:
\begin{equation}
    a = p(0)-p(1) = \bigl[p(0\&\text{success})-p(1\&\text{success})\bigr] + \underbrace{\bigl[p(0\&\text{fail})-p(1\&\text{fail})\bigr]}_{=\,0} ,
\end{equation}
so that both protocols estimate the same quantity.

On the other hand, the two protocols differ in statistical error and the number of shots needed. 
Without post-selection, only the Hadamard ancilla is measured, and every shot contributes $\pm 1$ to the estimator.
With post-selection, the $2n$ block-encoding ancillae are measured as well, and the shots in which the post-selection fails contribute $0$.
In this case, the estimate must be built from the joint probabilities, as in \cref{eq:hadamard_result_postselection_final}.

This difference affects the statistical cost only through a constant prefactor.
Comparing \cref{eq:shots_postselect} with \cref{eq:shots_no_postselect}, the numbers of shots required by the two protocols to reach the same error $\varepsilon$ satisfy
\begin{equation}\label{eq:shots_ratio}
   \frac{N_\text{tot}^{\text{PS}}}{N_\text{tot}^{\text{no-PS}}} = \frac{\gamma-a^2}{1-a^2} \le 1 .
\end{equation}
Consequently, for a fixed number of shots $N_\text{tot}$, post-selection yields an equal or smaller statistical error on $\real{\braket{\varphi|\Ham|\varphi}}$.
In both cases, however, the shot count scales as $\mathcal{O}(\normfact^2/\varepsilon^2)$, and is thereby dominated by the normalization factor $\normfact$ rather than by the choice of protocol.

%%%%%%%%%%%%%%%%%%%%%%%%%%%%%%%%%%%%%%%%%%%%%
\section{\PR{} circuits and resource estimation}
\label{app:PR circuits}
%%%%%%%%%%%%%%%%%%%%%%%%%%%%%%%%%%%%%%%%%%%%%
In this section, we present an explicit construction of the \PR{} circuit for three one-dimensional spin Hamiltonians: the Heisenberg XYZ, Heisenberg XXZ, and Ising models. By leveraging the structure of the underlying Hamiltonians, we obtain implementations that are significantly more efficient than generic constructions. 

For each case, we provide the non-asymptotic resource analysis for both the \foxlcu{} block encoding and the associated matrix-polynomial circuits, considering implementations with and without the additional overhead induced by controlling the entire circuit. Furthermore, we analyze two distinct assumptions on the connectivity of the physical device, namely all-to-all connectivity and a square-grid architecture.
Finally, we provide an analysis of the non-Clifford cost for the matrix polynomial \foxlcu{} in general and for the aforementioned one-dimensional spin models.

%%%%%%%%%%%%%%%%%%%%%%%%%%%%%%%%%%%
\subsection{\PR{} circuit and resource estimation for the one-dimensional Heisenberg XYZ model}\label{app:xyz_model}
%%%%%%%%%%%%%%%%%%%%%%%%%%%%%%% %%%%
We consider the one-dimensional XYZ Heisenberg model:
\begin{equation}\label{eq:def_xyz_2}
\Ham = g \sum_{i=0}^{n-1} Z_i + \sum_{i=0}^{n-2} \Jx X_i X_{i+1} +\Jy Y_i Y_{i+1} +\Jz Z_i Z_{i+1}.
\end{equation}
\Cref{fig:simp_pr_1} shows the circuit implementing the \PR{} oracle for such a model and for $n=4$. For completeness, we repeat the circuit, and describe how to generalize it for other values of $n$.
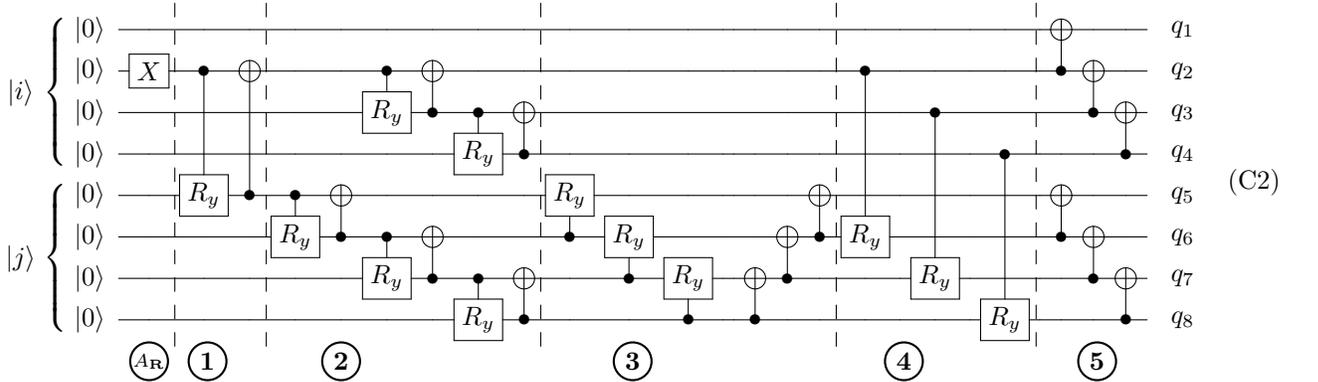
\begin{equation}\label{eq:circuit_pr_xyz}
\begin{myqcircuit*}{0}{0.4}
&&&&&\lstick{\ket{0}} & \qw \barrier[-1.1em]{7} & \qw & \qw \barrier[-1.1em]{7} & \qw & \qw & \qw & \qw & \qw & \qw \barrier[-1.1em]{7} & \qw & \qw & \qw & \qw & \qw & \qw & \qw & \qw \barrier[-1.1em]{7} & \qw & \qw & \qw& \qw & \qw & \qw \barrier{7} & \targ & \qw & \qw & \qw & \rstick{q_1} \\
&&&&&\lstick{\ket{0}} & \gate{X} & \ctrl{3} & \targ & \qw &\qw &\ctrl{1} & \targ & \qw & \qw & \qw & \qw & \qw & \qw & \qw & \qw & \qw & \qw & \ctrl{4} & \qw & \qw & \qw & \qw & \qw & \ctrl{-1} & \targ & \qw & \qw & \rstick{q_2} \\
&&&&&\lstick{\ket{0}} & \qw & \qw & \qw & \qw& \qw& \gate{R_y} & \ctrl{-1} & \ctrl{1} & \targ & \qw & \qw & \qw & \qw & \qw & \qw & \qw & \qw & \qw & \qw & \ctrl{4} & \qw & \qw& \qw & \qw & \ctrl{-1} & \targ & \qw & \rstick{q_3} \\
&&&&&\lstick{\ket{0}} & \qw & \qw & \qw & \qw & \qw &\qw &\qw  & \gate{R_y} & \ctrl{-1} & \qw & \qw & \qw & \qw & \qw & \qw & \qw & \qw & \qw & \qw & \qw & \qw & \ctrl{4} & \qw & \qw & \qw & \ctrl{-1} & \qw & \rstick{q_4} \\
&&&&&\lstick{\ket{0}} & \qw & \gate{R_y} & \ctrl{-3} & \ctrl{1} & \targ & \qw & \qw & \qw & \qw & \gate{R_y} & \qw & \qw & \qw & \qw & \qw & \qw & \targ & \qw & \qw & \qw & \qw & \qw & \qw & \targ & \qw & \qw & \qw & \rstick{q_5} \\
&&&&&\lstick{\ket{0}} & \qw & \qw & \qw & \gate{R_y} & \ctrl{-1} & \ctrl{1} & \targ & \qw & \qw & \ctrl{-1} & \gate{R_y} & \qw & \qw & \qw & \qw & \targ & \ctrl{-1} & \gate{R_y} & \qw & \qw & \qw & \qw & \qw & \ctrl{-1} & \targ & \qw & \qw & \rstick{q_6} \\
&&&&&\lstick{\ket{0}} & \qw & \qw & \qw & \qw & \qw & \gate{R_y} & \ctrl{-1} & \ctrl{1} & \targ & \qw & \ctrl{-1} & \gate{R_y} & \qw & \qw & \targ & \ctrl{-1} & \qw & \qw & \qw & \gate{R_y} & \qw & \qw & \qw & \qw & \ctrl{-1} & \targ & \qw & \rstick{q_7} \\
&&&&&\lstick{\ket{0}} & \qw & \qw & \qw & \qw & \qw & \qw & \qw & \gate{R_y} & \ctrl{-1} & \qw & \qw & \ctrl{-1} & \qw & \qw & \ctrl{-1} & \qw & \qw & \qw & \qw & \qw & \qw & \gate{R_y} & \qw & \qw & \qw & \ctrl{-1} & \qw & \rstick{q_8} \inputgroupv{1}{4}{0.8em}{2.4em}{\ket{i}} \inputgroupv{5}{8}{0.8em}{2.4em}{\ket{j}} \\
% labels
&&&&&& \circled[0.35]{{\scriptsize\PRc{}}} & \;\circled{1} &&& \circled{2} &&&&&& \;\circled{3} &&&&&&&& \;\circled{4} &&&&&& \;\circled{5}
\end{myqcircuit*}
\end{equation}
\bigskip

For clarity, we denote the $\Xp$ and $\Zp$ registers with the states $\ket{i}$ and $\ket{j}$ and we partition the \PR{} circuit into $6$ sections, in which we specify the angles $\theta$ required for the $R_y$ gates.

\begin{enumerate}
 \item[{\circled[0.35]{{\scriptsize\PRc{}}}}.] The first gate acts as an \emph{activation} for the whole circuit: the rest of the circuit has the state $\ket{0}_{2n}$ as eigenstate, therefore, and following the notation from \cref{ass:PR-PL}, the gate \PRc{} corresponds to an $\Xp$ gate.
 \item[\circled{1}.] These two gates, together, correspond to the gate $\subr{\theta}$ defined in \cite{foqcs-lcu-arxiv}. There, it was proven that it can be implemented with \cnot{} count and depth of $2$. We define the standard angle as:
 \begin{equation}
\tilde{\theta}
=
2 \arccos \left(
\sqrt{\dfrac{(|\Jx| + |\Jy|)(n-1)}{\normfact}}
\right),
\end{equation}
Then, the rotation angle depends on the sign of the coefficients $g$ and $\Jx$: 
\begin{equation}
\theta =
\begin{cases}
\tilde{\theta}, & g \ge 0,\; \Jx \ge 0, \\[2mm]
 2\pi-\tilde{\theta} , & g \ge 0,\; \Jx < 0, \\[2mm]
\tilde{\theta} + 2\pi, & g < 0,\; \Jx < 0, \\[2mm]
-\tilde{\theta} , & g < 0,\; \Jx \ge 0.
\end{cases}
\end{equation}
 We remark that, even when adding angles of $2\pi$, the cosine and sine change because in the $R_y$ gate, the cosine and sine are evaluated at $\frac\theta2$, as defined in \cref{eq:rydef}.
 \item[\circled{2}.] Implementation of two balanced Dicke states \cite{foqcs-lcu-arxiv,Bartschi_2019}: $n-2$ $\Gamma$ gates on the $\Xp$ register and $n-1$ $\Gamma$ gates on the $\Zp$ register. The first rotation angle on the $\Zp$ register is:
 \begin{equation}\label{eq: first angle}
 \theta_0 = 2\arccos\left(\sqrt{\frac{\abs{g}}{\abs{g}n+\abs{\Jz}(n-1)}}\right),
 \end{equation}
 while the other angles for both registers are defined as:
 \begin{equation}\label{eq:angles Dicke}
 \theta_i = 2\arccos\left(\sqrt{\frac{1}{n-i}}\right),
 \end{equation}
 for $i = 1,\dots, n-2$. The total \cnot{} count is $2(2n-3)$ since we have $2n-3$ $\Gamma$ gates. 
 \item[\circled{3}.] Acting on the $\Zp$ register, we have a ladder of $n-1$ controlled $R_y$ gates with constant angles: 
 \begin{align}
  \theta &=  
\begin{cases}
\tilde{\theta} & g\cdot\Jz \ge 0 \\[1ex]
-\tilde{\theta} & g\cdot \Jz < 0
\end{cases}, &
 \tilde\theta &= 2\arccos\left(\sqrt{\frac{|g|}{|g| + |\Jz|}}\right).
 \end{align}
 Then we apply a ladder of $n-1$ \cnot{} gates. The total \cnot{} count of this part is $3(\n-1)$ since we need 2 \cnots{} to decompose each controlled $R_y$ gate. 
 \item[\circled{4}.] Ladder of $n-1$ $R_y$ gates connecting the two registers, with constant rotation angles:
 \begin{align}\label{eq:angles_xyz_4}
 \theta &=  
\begin{cases}
\tilde{\theta} & \Jx\cdot\Jy \le 0 \\[1ex]
-\tilde{\theta} & \Jx\cdot\Jy>0
\end{cases}, &
 \tilde\theta &= 2\arccos\left(\sqrt{\frac{\abs{\Jx}}{\abs{\Jx}+\abs{\Jy}}}\right).
 \end{align}
 This part needs $2(\n-1)$ \cnots{}.
 \item[\circled{5}.] A ladder of $n-1$ \cnots{} is finally added on both registers.
\end{enumerate}

For \PL{}, the structure of the circuit is the same as \cref{eq:circuit_pr_xyz}.
Since the phase information of the Hamiltonian coefficients is entirely stored in \PR{}, we choose the standard rotation angles $\theta \equiv\tilde{\theta}$ throughout the \PL{} circuit.

The resource analysis for the state preparation oracles, the entire \foxlcu{} block encoding and the matrix polynomial algorithm presented in \cref{sec:poly_general} is reported in \cref{tab:resource_estimation_xyz}.
We consider both the all-to-all connectivity and the square grid connectivity.

%%%%%%%%%%%%%%%%%%%%%%%%%%%%%%%%%%%
\subsection{\PR{} circuit and resource estimation for the one-dimensional Heisenberg XXZ model}\label{app:XXZ_model}
%%%%%%%%%%%%%%%%%%%%%%%%%%%%%%%%%%%
Next, we show how the circuit in \cref{fig:simp_pr_1} (and equivalently \cref{eq:circuit_pr_xyz}), further simplifies when we have $g= 0$ and $\Jx = \Jy = J$, i.e, the XXZ model:
\begin{equation}\label{eq:def_heisenberg xxz}
\Ham = \sum_{i=0}^{n-2} J X_i X_{i+1} +J Y_i Y_{i+1} +\Jz Z_i Z_{i+1}.
\end{equation}
More specifically, the circuit for \PR{} is:
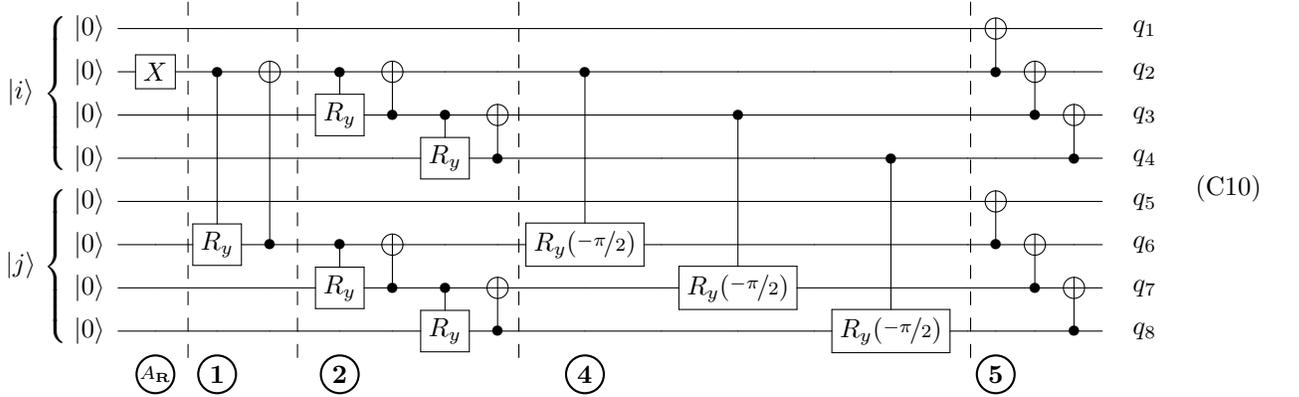
\begin{equation}\label{eq:circuit_pr_xxz}
\begin{myqcircuitr}{0}
&&&\lstick{\ket{0}} & \qw \barrier[-1.1em]{7} & \qw & \qw & \qw \barrier[-1.6em]{7} & \qw & \qw & \qw & \qw \barrier[-2.5em]{7} & \qw & \qw & \qw & \qw & \qw & \qw \barrier{7} & \targ & \qw & \qw & \qw & \rstick{q_1} \\
&&&\lstick{\ket{0}} & \gate{X} & \ctrl{4} & \targ & \qw & \ctrl{1} & \targ & \qw & \qw & \ctrl{4} & \qw & \qw & \qw & \qw & \qw & \ctrl{-1} & \targ & \qw & \qw & \rstick{q_2}\\
&&&\lstick{\ket{0}} & \qw & \qw & \qw & \qw & \gate{R_y} & \ctrl{-1} & \ctrl{1} & \targ & \qw & \qw & \ctrl{4} & \qw & \qw& \qw & \qw & \ctrl{-1} & \targ & \qw & \rstick{q_3} \\
&&&\lstick{\ket{0}} & \qw & \qw & \qw & \qw & \qw & \qw & \gate{R_y} & \ctrl{-1} & \qw & \qw & \qw & \qw & \ctrl{4} & \qw & \qw & \qw & \ctrl{-1} & \qw & \rstick{q_4} \\
&&&\lstick{\ket{0}} & \qw & \qw & \qw & \qw & \qw & \qw & \qw & \qw & \qw & \qw & \qw & \qw & \qw & \qw & \targ & \qw & \qw & \qw & \rstick{q_5} \\
&&&\lstick{\ket{0}} & \qw & \gate{R_y} & \ctrl{-4} & \qw & \ctrl{1} & \targ & \qw & \qw & \gate{R_y(\nicefrac{-\pi}{2})} & \qw & \qw & \qw & \qw & \qw & \ctrl{-1} & \targ & \qw & \qw & \rstick{q_6} \\
&&&\lstick{\ket{0}} & \qw & \qw & \qw & \qw & \gate{R_y} & \ctrl{-1} & \ctrl{1} & \targ & \qw & \qw & \gate{R_y(\nicefrac{-\pi}{2})} & \qw & \qw & \qw & \qw & \ctrl{-1} & \targ & \qw & \rstick{q_7} \\
&&&\lstick{\ket{0}} & \qw & \qw & \qw & \qw & \qw & \qw & \gate{R_y} & \ctrl{-1} & \qw & \qw & \qw & \qw & \gate{R_y(\nicefrac{-\pi}{2})} & \qw & \qw & \qw & \ctrl{-1} & \qw & \rstick{q_8} \inputgroupv{1}{4}{0.8em}{2.4em}{\ket{i}} \inputgroupv{5}{8}{0.8em}{2.4em}{\ket{j}}\\
% \labels
&&&& \circled[0.35]{{\scriptsize\PRc{}}} & \circled{1} && & \circled{2} &&&& \circled{4} &&&&&& \circled{5}
%\inputgroupv{5}{8}{10em}{3em}{\ket{j}}
\end{myqcircuitr}
\end{equation}
\bigskip

Following the same structure as in \cref{app:xyz_model}, we only list and describe the sections in \cref{eq:circuit_pr_xxz} that differ from \cref{eq:circuit_pr_xyz}:
\begin{enumerate}
 \item[\circled{1}.] Same as \cref{eq:circuit_pr_xyz}, except the target of $R_y$ and the control of the following \cnot{} is on the second qubit of the $\Zp$ register instead of the first.
 \item[\circled{2}.] Implementation of two balanced Dicke states \cite{foqcs-lcu-arxiv,Bartschi_2019}: $n-2$  $\Gamma$ gates on both the $\Xp$ and $\Zp$ register, with rotation angles as defined in \cref{eq:angles Dicke}.
 \item[\circled{3}.] There is no equivalent section here compared to \cref{eq:circuit_pr_xyz}.
 \item[\circled{4}.] Ladder of $n-1$ $R_y$ gates connecting the two registers. Since $\Jx=\Jy$, \cref{eq:angles_xyz_4} becomes:%, with constant rotation angles:
 \begin{equation} \label{eq: angle Ry fixed}
 \theta = -\tilde{\theta}  = -\frac{\pi}{2} . 
 \end{equation}
\end{enumerate}

For \PL{}, the structure of the circuit is the same as \cref{eq:circuit_pr_xxz}.
Since the phase information of the Hamiltonian coefficients is entirely stored in \PR{}, we choose the standard rotation angles $\theta \equiv\tilde{\theta}$ throughout the \PL{} circuit. For instance, \cref{eq: angle Ry fixed} becomes $\theta=\nicefrac{\pi}{2}$.
A resource analysis of the presented circuits for the one-dimensional Heisenberg XXZ model is summarized in \cref{tab:resource_estimation_xxz}.

\begin{table*}
 \begin{tabular}{|c|cc|cc|cc|}
 \hline
 \multirow{2}{*}{\bf Circuit} & \multicolumn{2}{c|}{\bf \cnot{} count} & \multicolumn{2}{c|}{\bf \cnot{} depth} & \multicolumn{2}{c|}{\bf Number of qubits} \\
 \cline{2-7}
 & all-to-all 
 & square grid 
 & all-to-all 
 & square grid 
 & all-to-all 
 & square grid  \\
 \hline
 \PR{}/\PLdag{} & $8 \n -8$ & $8\n-8$& $2\n + 6$ &$2\n +6$ & $2n\phantom{{}+1}$ &$2n\phantom{{}+1}$\\
 Controlled-\PR{}/\PLdag{} & $8 \n -7$ & $8\n-7$& $2\n + 7$ &$2\n +7$ & $2n+1$ &$2n+1$\\
 \hline
 \foxlcu{} & $18\n-16$& $20\n-16$ & $4\n+14$&$4\n+16$ & $3n\phantom{{}+1}$&$3n\phantom{{}+1}$\\
 Controlled-\foxlcu{} & $18\n-14$ & $20\n-14$ & $4\n+16$ & $4\n+18$ & $3n+1$ & $ 3n+1 $\\
 \hline
 $p_d(\Ham)$ & \,$18d\n-12d-4$\, & \,$20d\n-4d-6$\, & \,$4\n+6d+10$\, & \,$4\n+10d+12$\, & $2dn+d+n\phantom{{}+1}$& $2dn+2d+n\phantom{{}+1}$\\
 Controlled-$p_d(\Ham)$ & $18d\n-12d$ &$20d\n-4d-2$ &$4\n+6d+14$ & $4\n+10d+16$& \,$2dn+d+n+1$\, & \,$2dn+2d+n+1$\, \\
 \hline
 \end{tabular}
 \caption{Resource analysis for the \PR{} and \PLdag{} oracles, the \foxlcu{} block encoding, a generic matrix polynomial of $\Ham$ and their controlled circuits for the one-dimensional XXZ Heisenberg Hamiltonian defined in \cref{eq:def_heisenberg xxz}.}
 \label{tab:resource_estimation_xxz}
\end{table*}

%%%%%%%%%%%%%%%%%%%%%%%%%%%%%%%%%%%%%
\subsection{\PR{} circuit and resource estimation for the one-dimensional Ising model}
\label{app:ising_model}
%%%%%%%%%%%%%%%%%%%%%%%%%%%%%%%%%%%
Finally, we focus on how to simplify the circuit implementing the \PR{} oracles presented in the previous sections, for the Ising model, where the Hamiltonian is defined as:
\begin{equation}\label{eq:def_Ising}
\Ham = g\sum_{i=0}^{n-1}Z_i+J\sum_{i=0}^{n-2} X_i X_{i+1}.
\end{equation}
The circuit for \PR{} is the following:
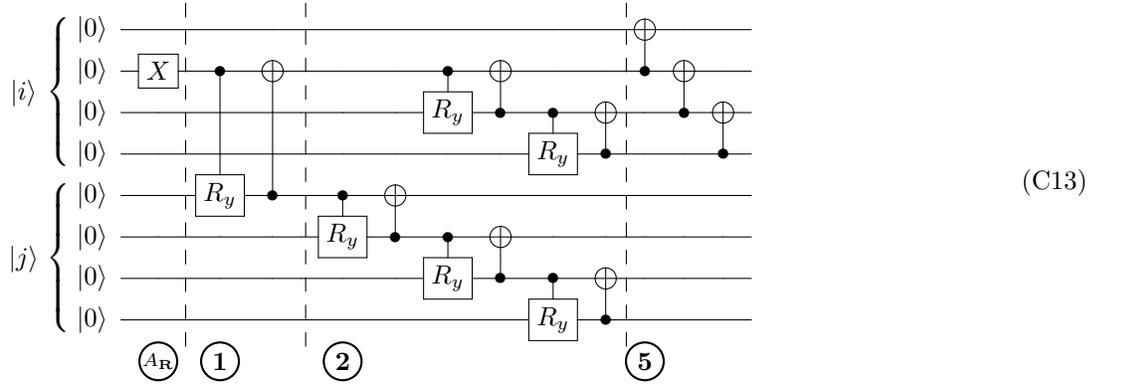
\begin{equation}\label{eq:circuit_ising_model}
\begin{myqcircuitr}{0}
&&&\lstick{\ket{0}} & \qw \barrier[-1.3em]{7} & \qw & \qw & \qw \barrier[-1.4em]{7}& \qw & \qw	& \qw	& \qw & \qw & \qw \barrier[-0.7em]{7} & \targ &	\qw	& \qw & \qw & \\
&&&\lstick{\ket{0}} & \gate{X} & \ctrl{3}	&	\targ & \qw & \qw& \qw & \ctrl{1} & \targ  & \qw & \qw & \ctrl{-1} & \targ & \qw & \qw \\
&&&\lstick{\ket{0}} & \qw & \qw & \qw & \qw & \qw & \qw& \gate{R_y} & \ctrl{-1} & \ctrl{1} & \targ & \qw & \ctrl{-1} &	\targ & \qw \\
&&&\lstick{\ket{0}} & \qw & \qw & \qw & \qw & \qw & \qw & \qw & \qw& \gate{R_y} & \ctrl{-1} & \qw & \qw & \ctrl{-1} & \qw \\
&&&\lstick{\ket{0}} & \qw & \gate{R_y} & \ctrl{-3} & \qw & \ctrl{1} & \targ & \qw & \qw & \qw & \qw & \qw & \qw & \qw & \qw \\
&&&\lstick{\ket{0}} & \qw & \qw	& \qw & \qw & \gate{R_y} & \ctrl{-1} & \ctrl{1} & \targ & \qw & \qw	& \qw & \qw	& \qw & \qw \\
&&&\lstick{\ket{0}} & \qw & \qw	& \qw &	\qw	& \qw &	\qw	& \gate{R_y} & \ctrl{-1} & \ctrl{1} & \targ & \qw	& \qw & \qw & \qw \\
&&&\lstick{\ket{0}} & \qw & \qw	& \qw & \qw & \qw & \qw & \qw & \qw & \gate{R_y} & \ctrl{-1} & \qw & \qw & \qw & \qw \inputgroupv{1}{4}{0.8em}{2.4em}{\ket{i}} \inputgroupv{5}{8}{0.8em}{2.4em}{\ket{j}}\\
% labels
&&&& \circled[0.35]{{\scriptsize\PRc{}}} & \circled{1} && & \circled{2} &&&&&& \circled{5} 
\end{myqcircuitr}
\end{equation} 
\bigskip

Looking at the \PR{} circuit defined in \cref{eq:circuit_pr_xyz} for the XYZ Heisenberg model, we notice that the sections \raisebox{0.30ex}{\circled[0.35]{{\scriptsize\PRc{}}}}, \circled{1}, and \circled{2} are identical, while the sections \circled{3} and \circled{4} are missing.
The only section that differs is \circled{5}, where instead of having one \cnot{} ladder on both the first and second ancilla registers, we only have one ladder in the first register. 

For \PL{}, the structure of the circuit is the same as \cref{eq:circuit_ising_model}.
Since the phase information of the Hamiltonian coefficients is entirely stored in \PR{}, we choose the standard rotation angles $\theta \equiv\tilde{\theta}$ throughout the \PL{} circuit.
A resource analysis of the circuits presented in this paper for the specific case of the Ising model is presented in \cref{tab:resource_estimation_ising}.

\begin{table*}
 \begin{tabular}{|c|cc|cc|cc|}
 \hline
 \multirow{2}{*}{\bf Circuit} & \multicolumn{2}{c|}{\bf \cnot{} count} & \multicolumn{2}{c|}{\bf \cnot{} depth} & \multicolumn{2}{c|}{\bf Number of qubits} \\
 \cline{2-7}
 & all-to-all 
 & square grid 
 & all-to-all 
 & square grid 
 & all-to-all 
 & square grid  \\
 \hline
 \PR{}/\PLdag{} & $5 \n -5$& $5\n-1$& $2\n\phantom{{}+1}$ &$2\n +4$ & $2n\phantom{{}+1}$ &$2n\phantom{{}+1}$\\
 Controlled-\PR{}/\PLdag{} &$5\n-4$ & $5\n\phantom{{}-1}$ &$2\n+1$ &$2\n+5$ & $2n+1$ &$2n+1$\\
 \hline
 \foxlcu{} &$12\n-10$ & $14\n-2$ &$4\n+2$ & $4\n+10$& $3n\phantom{{}+1}$ & $3n\phantom{{}+1}$\\
 Controlled-\foxlcu{} &$12\n-8\phantom{1}$ & $14\n\phantom{{}-2}$ & $4\n+4$& $4\n+12$& $3n+1$ &$3n+1$\\
 \hline
 $p_d(\Ham)$ & \,$12d\n-6d-4$\, & \,$14d\n+11d-6$\, & \,$4\n+6d-2$\, & $4n+10d+8\phantom{1}$ & $2dn+d+n\phantom{{}+1}$ & $2dn+2d+n\phantom{{}+1}$ \\
 Controlled-$p_d(\Ham)$ &$12d\n-6d\phantom{{}-4}$ & $14d\n+11d-2$ &$4\n+6d+2$ & \,$4n+10d+12$\, & \,$2dn+d+n+1$\, & \,$2dn+2d+n+1$\, \\
 \hline
 \end{tabular}
 \caption{Resource analysis for the \PR{} and \PLdag{} oracles, the \foxlcu{} block encoding, a generic matrix polynomial of $\Ham$ and their controlled circuits for the one-dimensional Ising Hamiltonian defined in \cref{eq:def_Ising}.}
 \label{tab:resource_estimation_ising}
\end{table*}

%%%%%%%%%%%%%%%%%%%%%%%%%%%%%%%%%%%%%%%%%%%%
\subsection{Non-Clifford cost}
\label{sec: non_clifford_cost}
%%%%%%%%%%%%%%%%%%%%%%%%%%%%%%%%%%%%%%%%%%%%%
In this section we analyze the non-Clifford cost of the matrix polynomial \cmlcu{} block encoding for all the models considered in the previous sections. 
In this framework, it is sufficient to count the number of single-qubit rotation gates to have an estimate of the non-Clifford cost, since our circuits require no Toffolis nor multi-controlled operations. Moreover, the \select{} subroutine consists solely of $2\n$ singly controlled Pauli operators, which are Clifford, so that the non-Clifford cost arises exclusively from the single-qubit rotations appearing in the state-preparation circuits.
The corresponding resource estimates are summarized in \cref{tab:extra non-clifford}.

\paragraph{Brief derivations of rotation gate counts:} The \cmlcu{} block encoding is composed of three building blocks: \PR{}, \select{}, and \PLdag{}. Since \PLdag{} has the same gate complexity as \PR{}, while \select{} is entirely Clifford, the total number of single-qubit rotations required by \cmlcu{} is simply:
\begin{equation}
    \text{Rotations}_{\text{\cmlcu{}} }=  2(\text{Rotations}_{\PR{}})\, .
\end{equation}

For the block encoding of a degree-$d$ polynomial, the circuit contains $d$ instances of \PR{} and $d$ instances of \PLdag{}, together with the additional subroutines \polyR{} and \polyL{}, which also contribute to the non-Clifford cost.
The \polyL{} circuit consists of $d-1$ controlled rotation gates and one single-qubit rotation. Since each controlled rotation can be decomposed into two single-qubit rotations and two \cnots{}, the total number of rotations required by \polyL{} is $2d-1$. Similarly, \polyR{} contains the same $d-1$ controlled rotations, together with an additional layer of $d$ phase rotations used to encode the phases of the polynomial coefficients. Therefore, the total number of rotations required by \polyR{} is $3d-1$. Combining all contributions, the total number of single-qubit rotations required to block encode a matrix polynomial of degree $d$ is
\begin{equation}
    \text{Rotations}_{p_d(\Ham)} =  2d(\text{rotations}_{\PR{}})+5d-2\, .
\end{equation}
We now consider each one-dimensional spin Hamiltonian separately and derive the number of single-qubit rotations required to implement the corresponding state-preparation circuit \PR{}.

For the XYZ Heisenberg model, the \PR{} circuit is shown in \cref{eq:circuit_pr_xyz} for the case $n=4$. The circuit is divided into six sections. The first section contains one controlled rotation. The second section contains $n-2$ controlled rotations acting on the $\ket{i}$ register and $n-1$ controlled rotations acting on the $\ket{j}$ register. The third and fourth sections each contribute an additional $n-1$ controlled rotations. Since each controlled rotation is decomposed into two single-qubit rotations, the total number of rotations is
\begin{equation}
\text{Rotations}_{\PR{}} =
2+2(n-2)+2(n-1)+2(n-1)+2(n-1)
= 8n-8.
\end{equation}

For the XXZ Heisenberg model, the \PR{} circuit is shown in \cref{eq:circuit_pr_xxz} for $n=4$. The circuit is divided into five sections. The first section contains one controlled rotation, while the second section contributes $n-2$ controlled rotations on both the $\ket{i}$ and $\ket{j}$ registers. Finally, the fourth section contains $n-1$ controlled rotations. Therefore, the total number of single-qubit rotations is
\begin{equation}
\text{Rotations}_{\PR{}} =
2+4(n-2)+2(n-1)
= 6n-8.
\end{equation}

For the Ising model, the \PR{} circuit is shown in \cref{eq:circuit_ising_model} for $n=4$. The circuit is divided into four sections. The first section contains one controlled rotation, while the second section contributes $n-2$ controlled rotations on the $\ket{i}$ register and $n-1$ controlled rotations on the $\ket{j}$ register. Hence, the total number of single-qubit rotations is
\begin{equation}
\text{Rotations}_{\PR{}} =
2+2(n-2)+2(n-1)
= 4n-4.
\end{equation}
\begin{table}[h]
    \centering
    \begin{tabular}{|c|c|c|c|}
    \hline
     & \textbf{XYZ}  & \textbf{XXZ} & \textbf{Ising} \\ \hline
\PR{} / \PLdag{}& $8\n-8$ & $6\n-8$ & $4\n-4$ \\ \hline
    \foxlcu{} & $16n-16$ & $12\n-16$ & $8\n-8$ \\ \hline
    $p_d(\Ham)$ & $16dn-11d-2$ & $12dn-11d-2$ & $8dn-3d-2$ \\ \hline
\end{tabular}
\caption{Rotation gate counts for the \PR{} and \PLdag{} oracles, the \foxlcu{} block encoding, and a generic matrix polynomial of $\Ham$ for the one-dimensional XYZ Heisenberg Hamiltonian, the XXZ Heisenberg Hamiltonian, and the Ising Hamiltonian.}
    \label{tab:extra non-clifford}
\end{table}

%%%%%%%%%%%%%%%%%%%%%%%%%%%%%%%%%%%%%%%
\section{Explicit 2D circuits with square grid connectivity}\label{app:2D circuit}
%%%%%%%%%%%%%%%%%%%%%%%%%%%%%%%%%%%%%%%%
Standard quantum circuits are typically described using a one-dimensional (1D) qubit ordering and assume arbitrary two-qubit interactions. In contrast, real quantum hardware often restricts two-qubit gates to nearest-neighbor interactions on a two-dimensional (2D) square grid.
In this section, we describe how to map 1D circuits onto a 2D qubit layout that respects these qubit connectivity constraints.

We present explicit 2D circuit implementations for the \select{} oracle of the \foxlcu{} block encoding and for the \PR{} oracle in the specific case of the one-dimensional XXZ Heisenberg model.
Assuming square-grid connectivity, we begin by mapping the qubits of the original 1D circuit onto the 2D grid.
%We present the qubit mapping and the explicit two-dimensional circuit for the \select{} oracle of the \foxlcu{} block encoding and the $\PR$ oracle for the specific case of the one-dimensional XXZ Heisenberg model.
%We assume that the real quantum hardware has square grid connectivity, and we start by mapping qubits from the one-dimensional circuit to the two-dimensional grid.

The \foxlcu{} block encoding consists of three qubit registers: $\n$ qubits for the \Xp{} register, denoted by $\ket{i}$, $\n$ qubits for the \Zp{} register, denoted by $\ket{j}$, and $\n$ qubits for the input quantum state $\ket{\varphi}$.
We therefore arrange these qubits in a rectangular $3 \times \n$ layout, as shown in \cref{fig:2dmapping} for the case $n=6$.

%%% 2D CIRCUIT COMMANDS %%%
\newcommand{\myIIIxVIempty}[1]{%
\begin{tikzpicture}[scale=0.5]
\draw[gray] (0,0) grid (6,-3);
#1
\end{tikzpicture}
}

%%% FIGURE %%%
\begin{figure}[hbtp]
\centering
\subfloat[Standard linear grid with 1D qubit ordering]{%
\begin{minipage}[][14em]{0.5\textwidth}
\centering
$\begin{myqcircuit*}{0.75}{1}
    \lstick{\text{\tiny 0}~~} & \qw & \qw & \\
    \lstick{\text{\tiny 1}~~} & \qw & \qw & \\
    \lstick{\text{\tiny 2}~~} & \qw & \qw & \\
    \lstick{\text{\tiny 3}~~} & \qw & \qw & \\
    \lstick{\text{\tiny 4}~~} & \qw & \qw & \\
    \lstick{\text{\tiny 5}~~} & \qw & \qw & \\
    \lstick{\text{\tiny 6}~~} & \qw & \qw & \\
    \lstick{\text{\tiny 7}~~} & \qw & \qw & \\
    \lstick{\text{\tiny 8}~~} & \qw & \qw & \\
    \lstick{\text{\tiny 9}~~} & \qw & \qw & \\
    \lstick{\text{\tiny10}~~} & \qw & \qw & \\
    \lstick{\text{\tiny11}~~} & \qw & \qw & \\
    \lstick{\text{\tiny12}~~} & \qw & \qw & \\
    \lstick{\text{\tiny13}~~} & \qw & \qw & \\
    \lstick{\text{\tiny14}~~} & \qw & \qw & \\
    \lstick{\text{\tiny15}~~} & \qw & \qw & \\
    \lstick{\text{\tiny16}~~} & \qw & \qw & \\
    \lstick{\text{\tiny17}~~} & \qw & \qw & {}
    \gategroup{1}{4}{6}{4}{0.5em}{\}}
    \gategroup{7}{4}{12}{4}{0.5em}{\}}
    \gategroup{13}{4}{18}{4}{0.5em}{\}}
    \inputgrouph{3}{4}{0.33em}{\longrightarrow\ket{i}}{-5.5em}
    \inputgrouph{9}{4}{0.33em}{\longrightarrow\ket{j}}{-5.5em}
    \inputgrouph{15}{4}{0.33em}{\longrightarrow\ket{\varphi}}{-5.5em}
\end{myqcircuit*}$
\end{minipage}
\label{fig:1d_qubit_foqcs}
}%
\subfloat[Square grid with 2D qubit ordering]{%
\quad\begin{minipage}[][14em]{0.5\textwidth}
\centering
\myIIIxVIempty{%
\mynumbers{0}{0/12,1/13,2/14,3/15,4/16,5/17}
\mynumbers{1}{0/0,1/1,2/2,3/3,4/4,5/5}
\mynumbers{2}{0/6,1/7,2/8,3/9,4/10,5/11}
\draw (6.2,-0.5) node[right] {$\longrightarrow\ket{\varphi}$};
\draw (6.2,-1.5) node[right] {$\longrightarrow\ket{i}$};
\draw (6.2,-2.5) node[right] {$\longrightarrow\ket{j}$};
\draw (3,0) node[above] {\footnotesize Before \foxlcu{}};
} \\ \vspace*{1em}
\myIIIxVIempty{%
\mynumbers{0}{0/0,1/1,2/2,3/3,4/4,5/5}
\mynumbers{1}{0/6,1/7,2/8,3/9,4/10,5/11}
\mynumbers{2}{0/12,1/13,2/14,3/15,4/16,5/17}
\draw (6.2,-0.5) node[right] {$\longrightarrow\ket{i}$};
\draw (6.2,-1.5) node[right] {$\longrightarrow\ket{j}$};
\draw (6.2,-2.5) node[right] {$\longrightarrow\ket{\varphi}$};
\draw (3,0) node[above] {\footnotesize After \foxlcu{}};
} \\
\end{minipage}
\label{fig:2d_qubit_foqcs}
}
\caption{Qubit ordering for \subref{fig:1d_qubit_foqcs} 1D layout and \subref{fig:2d_qubit_foqcs} 2D layout of the quantum circuit implementing the \foxlcu{} block encoding, assuming square grid connectivity. The grid consists of $3\times n$ qubits, here visualized for $n=6$. $\ket{\varphi}$, $\ket{i}$ and $\ket{j}$ correspond to the system, $\Xp$, and $\Zp$ registers, respectively.}
 \label{fig:2dmapping}
\end{figure}

%%%%%%%%%%%%%%%%%%%%%%%%%%%%%
\subsection{2D circuits of the \select{} oracle}
\label{sec:2d_select}
%%%%%%%%%%%%%%%%%%%%%%%%%%%%%
Starting from the qubit mapping in \cref{fig:2dmapping}, we now discuss how to implement the \select{} oracle of the \foxlcu{} block encoding.

First, we notice that while the $\Xp$ register and the system qubits are directly connected, we need to perform SWAP gates to bring the $\Zp$ register and the system qubits close to each other. The 2D circuit is shown in \cref{fig:2d select}.
%Using only $2\n$ SWAP gates, $\n$ \cnots{}, and $\n$ \cz{} gates as shown in \cref{fig:2d select}.
We remark that each SWAP gate can be decomposed into three \cnots{} \cite{Nielsen_2011}.
Since $n$ \cnot{}s are applied immediately after the first $n$ SWAP gates on the same qubits, these operations can be merged, yielding a total of $2\n$ \cnots{} for the first two diagrams of \cref{fig:2d select}.
Similarly, the \cz{} gates and the final layer of SWAPs can also be merged:
\begin{equation}
    \begin{myqcircuitr}{1.35}
    &\ctrl{1} & \qw & \qswap{1} & \qw \\
    &\ctrl{0} & \qw & \qswap{0} \qwx & \qw
    \end{myqcircuitr} \ = \
    \raisebox{0.6ex}{$\begin{myqcircuit}
    &\gate{H}&\targ &\gate{H} &\qswap{1} & \qw \\
    &\qw&\ctrl{-1} &\qw &\qswap{0} \qwx & \qw
    \end{myqcircuit}$}
    \ = \
    \begin{myqcircuit}
    &\gate{H}&\targ  &\qswap{1} & \qw &\qw\\
    &\qw&\ctrl{-1} &\qswap{0} \qwx & \gate{H} &\qw
    \end{myqcircuit}
    \ = \
    \begin{myqcircuit}
    &\gate{H}&\ctrl{1} &\targ & \qw &\qw\\
    &\qw&\targ &\ctrl{-1}  & \gate{H} &\qw
    \end{myqcircuit}
\end{equation}
Note that this simplification is made possible by the choice of inverting the system register $\ket{\varphi}$ with the ancilla registers $\ket{i}$ and $\ket{j}$ after the \select{} operation.
Overall, the construction requires $4\n$ \cnots{} and has a \cnot{} depth equal to $4$. 

%%% FIGURE %%%
\begin{figure}[hbtp]
\centering
\myIIIxVIempty{%
\draw (-0.2,-0.5) node[left] {$\ket{\varphi}\longleftarrow$};
\draw (-0.2,-1.5) node[left] {$\ket{i}\longleftarrow$};
\draw (-0.2,-2.5) node[left] {$\ket{j}\longleftarrow$};
\mySWAP0010
\mySWAP0111
\mySWAP0212
\mySWAP0313
\mySWAP0414
\mySWAP0515
} \quad
\myIIIxVIempty{%
\myCNOT0010
\myCNOT0111
\myCNOT0212
\myCNOT0313
\myCNOT0414
\myCNOT0515
} \quad
\myIIIxVIempty{%
\myCZ1020
\myCZ1121
\myCZ1222
\myCZ1323
\myCZ1424
\myCZ1525
} \quad
\myIIIxVIempty{%
\draw (6.2,-0.5) node[right] {$\longrightarrow\ket{i}$};
\draw (6.2,-1.5) node[right] {$\longrightarrow\ket{j}$};
\draw (6.2,-2.5) node[right] {$\longrightarrow\ket{\varphi}$};
\mySWAP1020
\mySWAP1121
\mySWAP1222
\mySWAP1323
\mySWAP1424
\mySWAP1525
} \\
\caption{2D quantum circuit implementing the \select{} oracle for $n=6$ and assuming grid connectivity.}
 \label{fig:2d select}
\end{figure}

%%%%%%%%%%%%%%%%%%%%%%%%%%%%%%%%%%%
\subsection{2D circuits of the \PR{} oracle}
%%%%%%%%%%%%%%%%%%%%%%%%%%%%%%%%%%%%
In this section, we utilize the 2D qubit mapping from \cref{fig:2dmapping} to implement the \PR{} oracle of the \foxlcu{} block encoding on a quantum device with $2$D grid connectivity.
Since the state preparation oracles differ between models, here we consider the specific case of the one-dimensional XXZ Heisenberg model defined in \cref{app:XXZ_model}. The results can be immediately generalized to the other spin Hamiltonians presented in this paper.

Starting from the \PR{} circuit from \cref{eq:circuit_pr_xxz}, the two-dimensional circuit is shown in \cref{fig:2dpr} for the case $n=6$.
\PL{} follows the same construction and just differs in the rotation angles, as explained in \cref{app:XXZ_model}.

%%%%%%%%%%%%%%%%%%%%%%%%%%%%%%%%%%%%%%%%%%%%%%
\subsection{2D circuit description for the matrix polynomial block encoding}
%%%%%%%%%%%%%%%%%%%%%%%%%%%%%%%%%%%%%%%%%%%%%%%
Finally, we show the 2D qubit mapping utilized for the circuits in \cref{fig:poly}.
Starting from the layout in \cref{fig:2dmapping}, the registers $\ket{i_k}$ and $\ket{j_k}$ for $k = 0,\ldots,d-1$ are stacked vertically, as shown in \cref{fig: 2D circuit poly}. In addition, a column of $2d$ ancillary qubits is placed on the left-hand side of the layout.

%%% 2D CIRCUIT COMMANDS %%%
\newcommand{\myPRempty}[4][]{%
\begin{tikzpicture}[scale=0.75]
\draw[#4] (0,0) grid (#3,-#2);
#1
\end{tikzpicture}}
\newcommand{\myPR}[3]{%
\myPRempty[%
\mynumbers{0}{0/0,1/1,2/2,3/3,4/4,5/5}
\mynumbers{1}{0/6,1/7,2/8,3/9,4/10,5/11}
#3
]{#1}{#2}{gray}}

%%% FIGURE %%%
\begin{figure}[hbtp]
\centering
\myPR{2}{6}{%
\myX01
} \quad
\myPR{2}{6}{%
\myCRY0111
} \quad
\myPR{2}{6}{%
\myCNOT1101
} \\[1em]
\myPR{2}{6}{%
\myCRY0102
\myCRY1112
} \quad
\myPR{2}{6}{%
\myCNOT0201
\myCNOT1211
} \quad
\myPR{2}{6}{%
\myCRY0111
\myCRY1213
\myCRY0203
} \\[1em]
\myPR{2}{6}{%
\myCNOT0100
\myCNOT1110
\myCNOT0302
\myCNOT1312
} \quad
\myPR{2}{6}{%
\myCRY0212
\myCRY0304
\myCRY1314
} \quad
\myPR{2}{6}{%
\myCNOT0201
\myCNOT1211
\myCNOT0403
\myCNOT1413
} \\[1em]
\myPR{2}{6}{%
\myCRY0313
\myCRY0405
\myCRY1415
} \quad
\myPR{2}{6}{%
\myCNOT0302
\myCNOT1312
\myCNOT0504
\myCNOT1514
} \quad
\myPR{2}{6}{%
\myCRY0414
\myCRY0515
} \\[1em]
\myPR{2}{6}{%
\myCNOT0403
\myCNOT1413
} \quad 
\myPR{2}{6}{%
\myCNOT0504
\myCNOT1514
} \quad
\myPRempty{2}{6}{white} \\
\caption{2D quantum circuit implementing the \PR{} oracle of the one-dimensional XXZ Heisenberg model. We assume $n=6$ and square grid connectivity.}
 \label{fig:2dpr}
\end{figure}

%%% 2D CIRCUIT COMMANDS %%%
\newcommand{\mypolyIVxVIempty}[2][]{%
\begin{tikzpicture}[scale=0.5]
#1
\draw[gray] (0,0) grid (7,-9);
\draw[very thick] (1,0) -- (1,-9);
#2
\end{tikzpicture}}
\newcommand{\mypolyIVxVIbefore}[2][]{%
\mypolyIVxVIempty[{%
#1
\mynumbers{0}{0/0a,1/52,2/53,3/54,4/55,5/56,6/57}
\mynumbers{1}{0/0,1/4,2/5,3/6,4/7,5/8,6/9}
\mynumbers{2}{0/1a,1/10,2/11,3/12,4/13,5/14,6/15}
\mynumbers{3}{0/1,1/16,2/17,3/18,4/19,5/20,6/21}
\mynumbers{4}{0/2a,1/22,2/23,3/24,4/25,5/26,6/27}
\mynumbers{5}{0/2,1/28,2/29,3/30,4/31,5/32,6/33}
\mynumbers{6}{0/3a,1/34,2/35,3/36,4/37,5/38,6/39}
\mynumbers{7}{0/3,1/40,2/41,3/42,4/43,5/44,6/45}
\mynumbers{8}{0/,1/46,2/47,3/48,4/49,5/50,6/51}
}]{#2}}
\newcommand{\mypolyIVxVIafter}[2][]{%
\mypolyIVxVIempty[{%
#1
\mynumbers{0}{0/0,1/4,2/5,3/6,4/7,5/8,6/9}
\mynumbers{1}{0/1a,1/10,2/11,3/12,4/13,5/14,6/15}
\mynumbers{2}{0/1,1/16,2/17,3/18,4/19,5/20,6/21}
\mynumbers{3}{0/2a,1/22,2/23,3/24,4/25,5/26,6/27}
\mynumbers{4}{0/2,1/28,2/29,3/30,4/31,5/32,6/33}
\mynumbers{5}{0/3a,1/34,2/35,3/36,4/37,5/38,6/39}
\mynumbers{6}{0/3,1/40,2/41,3/42,4/43,5/44,6/45}
\mynumbers{7}{0/0a,1/46,2/47,3/48,4/49,5/50,6/51}
\mynumbers{8}{0/,1/52,2/53,3/54,4/55,5/56,6/57}
}]{#2}}

%%% FIGURE %%%
\begin{figure}[hbtp]
\centering
\subfloat[Linear grid with 1D qubit ordering]{%
\begin{minipage}[][19em]{0.3\textwidth}
\centering
$\begin{myqcircuitr}{1.5}
    \lstick{\text{\tiny0,1,2,3}~~} & \qw & {/\strut^{4}}\qw & \qw & \qw & \rstick{\longrightarrow\polyR{}} \\
    \lstick{\text{\tiny\phantom{1}4,\phantom{1}5,\ldots,\phantom{1}9}~~} & \qw & {/\strut^{6}}\qw & \qw & \qw & \rstick{\longrightarrow\ket{i_0}} \\
    \lstick{\text{\tiny10,11,\ldots,15}~~} & \qw & {/\strut^{6}}\qw & \qw & \qw & \rstick{\longrightarrow\ket{j_0}} \\
    \lstick{\text{\tiny16,17,\ldots,21}~~} & \qw & {/\strut^{6}}\qw & \qw & \qw & \rstick{\longrightarrow\ket{i_1}} \\
    \lstick{\text{\tiny22,23,\ldots,27}~~} & \qw & {/\strut^{6}}\qw & \qw & \qw & \rstick{\longrightarrow\ket{j_1}} \\
    \lstick{\text{\tiny28,29,\ldots,33}~~} & \qw & {/\strut^{6}}\qw & \qw & \qw & \rstick{\longrightarrow\ket{i_2}} \\
    \lstick{\text{\tiny34,35,\ldots,39}~~} & \qw & {/\strut^{6}}\qw & \qw & \qw & \rstick{\longrightarrow\ket{j_2}} \\
    \lstick{\text{\tiny40,41,\ldots,45}~~} & \qw & {/\strut^{6}}\qw & \qw & \qw & \rstick{\longrightarrow\ket{i_3}} \\
    \lstick{\text{\tiny46,47,\ldots,51}~~} & \qw & {/\strut^{6}}\qw & \qw & \qw & \rstick{\longrightarrow\ket{j_3}} \\
    \lstick{\text{\tiny52,53,\ldots,57}~~} & \qw & {/\strut^{6}}\qw & \qw & \qw & \rstick{\longrightarrow\ket{\varphi}} \\
\end{myqcircuitr}$
\end{minipage}
\label{fig:1d_qubit_pd}
}%
\subfloat[Square grid with 2D qubit ordering]{%
\begin{minipage}[][19em]{0.7\textwidth}
\centering
\mypolyIVxVIbefore{%
\draw (7.2,-0.5) node[right] {$\longrightarrow\ket{\varphi}$};
\draw (7.2,-1.5) node[right] {$\longrightarrow\ket{i_0}$};
\draw (7.2,-2.5) node[right] {$\longrightarrow\ket{j_0}$};
\draw (7.2,-3.5) node[right] {$\longrightarrow\ket{i_1}$};
\draw (7.2,-4.5) node[right] {$\longrightarrow\ket{j_1}$};
\draw (7.2,-5.5) node[right] {$\longrightarrow\ket{i_2}$};
\draw (7.2,-6.5) node[right] {$\longrightarrow\ket{j_2}$};
\draw (7.2,-7.5) node[right] {$\longrightarrow\ket{i_3}$};
\draw (7.2,-8.5) node[right] {$\longrightarrow\ket{j_3}$};
\draw (0.5,-9.0) node[below] {\rotatebox{-90}{$\longrightarrow$}};
\draw (0.5,-10.3) node[below] {\polyR{}};
\draw (3.5,0) node[above] {\footnotesize Before $p_d(\Ham)$-\foxlcu{}};
} \quad
\mypolyIVxVIafter{%
\draw (7.2,-0.5) node[right] {$\longrightarrow\ket{i_0}$};
\draw (7.2,-1.5) node[right] {$\longrightarrow\ket{j_0}$};
\draw (7.2,-2.5) node[right] {$\longrightarrow\ket{i_1}$};
\draw (7.2,-3.5) node[right] {$\longrightarrow\ket{j_1}$};
\draw (7.2,-4.5) node[right] {$\longrightarrow\ket{i_2}$};
\draw (7.2,-5.5) node[right] {$\longrightarrow\ket{j_2}$};
\draw (7.2,-6.5) node[right] {$\longrightarrow\ket{i_3}$};
\draw (7.2,-7.5) node[right] {$\longrightarrow\ket{j_3}$};
\draw (7.2,-8.5) node[right] {$\longrightarrow\ket{\varphi}$};
\draw (0.5,-9.0) node[below] {\rotatebox{-90}{$\longrightarrow$}};
\draw (0.5,-10.3) node[below] {\polyR{}};
\draw (3.5,0) node[above] {\footnotesize After $p_d(\Ham)$-\foxlcu{}};
} \\
\end{minipage}
\label{fig:2d_qubit_pd}} \\
\caption{Qubit ordering for \subref{fig:1d_qubit_pd} 1D layout and \subref{fig:2d_qubit_pd} 2D layout of the quantum circuit implementing the matrix polynomial encoding shown in \cref{fig:poly}. The grid consists of $2d+(2d+1) n$ qubits, here visualized for $n=6$ and $d=4$. The first $2d$ qubits are positioned in the first column to implement the \polyR{} gate as described in \cref{circuit:poly_2d_general}.}
 \label{fig: 2D circuit poly}
\end{figure}

On these additional $2d$ ancillary qubits, we implement modified versions of the \polyR{} and \polyL{} oracles, differing from their standard realizations presented in \cref{lem:POLYR,lem:POLYL}. In particular, following the qubit mapping of \cref{fig: 2D circuit poly}, we redefine the circuit implementing \polyR{} as follows:
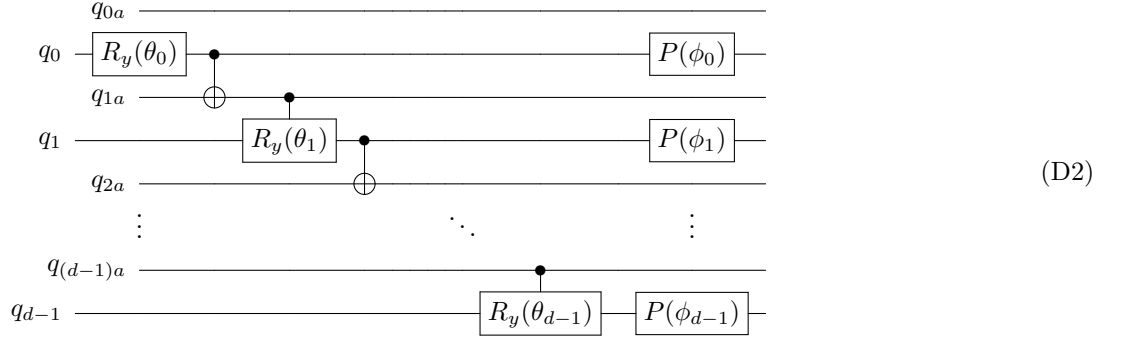
\begin{equation}\label{circuit:poly_2d_general}
\begin{myqcircuitr}{0}
&\lstick{q_{0a}} & \qw & \qw & \qw & \qw & \qw & \qw & \qw & \qw & \qw & \qw & \qw & \qw \\
\lstick{q_0} & \gate{R_y(\theta_0)} & \ctrl{1} & \qw & \qw & \qw & \qw & \qw & \qw & \qw & \qw & \qw & \gate{P(\phi_0)} & \qw \\
&\lstick{q_{1a}} & \targ{} & \ctrl{1} & \qw & \qw & \qw & \qw & \qw & \qw & \qw & \qw & \qw & \qw \\
\lstick{q_1} & \qw & \qw & \gate{R_y(\theta_1)} & \ctrl{1} & \qw & \qw & \qw & \qw & \qw & \qw & \qw & \gate{P(\phi_1)} & \qw \\
&\lstick{q_{2a}} & \qw & \qw & \targ{} & \qw & \qw & \qw & \qw & \qw & \qw & \qw & \qw & \qw  \\
& \raisebox{1.5ex}{$\vdots$} & & & & & & & & \raisebox{1.5ex}{$\ddots$} & & & \raisebox{1.5ex}{$\vdots$} \\
&\lstick{q_{(d-1)a}} & \qw & \qw & \qw & \qw & \qw & \qw & \qw & \qw & \ctrl{1} & \qw & \qw & \qw \\
\lstick{q_{d-1}} & \qw & \qw & \qw & \qw & \qw & \qw & \qw & \qw & \qw & \gate{R_y(\theta_{d-1})} & \qw&\gate{P(\phi_{d-1})} & \qw 
\end{myqcircuitr}
\end{equation}
The rationale behind doubling the number of qubits used to implement \polyR{} can be understood by inspecting \cref{fig:poly_pr_ab} in the specific setting where $\PR{}$ and $\PL{}$ admit the decompositions given in \cref{eq:trivial_control_pr,eq:trivial_control_pldag}. After applying \polyR{}, one encounters $d$ \cnot{} gates with target on the second qubit of each $\ket{i_k}$ ancilla register. 
Consequently, it is advantageous for the $d$ controlling qubits, on which \polyR{} acts, to be placed as close as possible to the corresponding target qubits in the $\ket{i_k}$ registers.
This proximity is achieved by interleaving these $d$ qubits with an additional set of $d-1$ qubits, using the modified implementation shown in \cref{circuit:poly_2d_general}.  We remark that the wire $q_{0a}$ is intentionally left unused in \polyR{}; however, it is needed for \polyL{}, as described below.

Even with this layout, for $k = 0,\ldots,d-1$, the control qubit $q_k$ and its corresponding target qubit in the $\ket{i_k}$ register, namely $q_{d+2nk+1}$, are still not nearest neighbors. A \cnot{} gate between these two qubits can be implemented using three nearest-neighbor \cnot{} gates:
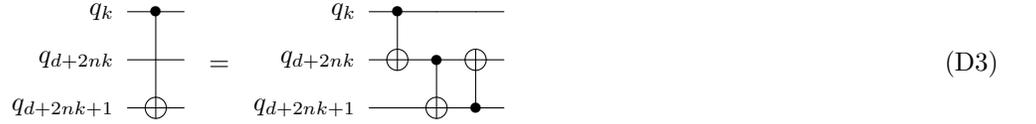
\begin{equation}\label{eq:d2_cnot}
    \begin{myqcircuitr}{0.7}
        \lstick{q_k} & \qw & \qw  & \ctrl{2} & \qw \\
        \lstick{q_{d+2nk}} & \push{\rule{0.5em}{0em}} & \lstick{\ket{0}}  & \qw & \qw \\
        \lstick{q_{d+2nk+1}} & \push{\rule{0.5em}{0em}} & \lstick{\ket{0}}  & \targ{} & \qw
    \end{myqcircuitr}
    \qquad = \qquad\qquad\qquad
    \begin{myqcircuitr}{0.7}
        \lstick{q_k} & \qw & \qw  & \ctrl{1} & \qw & \qw & \qw \\
        \lstick{q_{d+2nk}} & \push{\rule{0.5em}{0em}} & \lstick{\ket{0}}  & \targ{} & \ctrl{1} & \targ{} & \qw \\
        \lstick{q_{d+2nk+1}} & \push{\rule{0.5em}{0em}} & \lstick{\ket{0}}  & \qw & \targ{} & \ctrl{-1} & \qw
    \end{myqcircuitr}
\end{equation}
where we used the fact that the qubits $q_{d+2nk}$ and $q_{d+2nk+1}$ are initially in the state $\ket{0}$.
In the case of $\polyL^\dagger$, the \cnot{} 
%is shifted up by one row relative to where it was previously applied and
is decomposed as follows:
\begin{equation}\label{eq:d2_cnot_pl}
    \begin{myqcircuitr}{0.7}
        \lstick{q_{k}} & \ctrl{2} & \qw & \qw \\
        \lstick{q_{d+2nk}} & \qw & \meter & \rstick{\hspace{0.2em}0} \ar@{=}[l];[]+<0.2em,0em> \\
        \lstick{q_{d+2nk+1}} & \targ{} & \meter & \rstick{\hspace{0.2em}0} \ar@{=}[l];[]+<0.2em,0em>
    \end{myqcircuitr}
    \qquad\quad = \qquad\qquad\qquad
    \begin{myqcircuitr}{0.7}
        \lstick{q_{k}} & \qw & \qw & \ctrl{1} & \qw & \qw \\
        \lstick{q_{d+2nk}} & \targ{} & \ctrl{1} & \targ{} & \meter & \rstick{\hspace{0.2em}0} \ar@{=}[l];[]+<0.2em,0em> \\
        \lstick{q_{d+2nk+1}} & \ctrl{-1} & \targ{} & \qw & \meter & \rstick{\hspace{0.2em}0} \ar@{=}[l];[]+<0.2em,0em>
    \end{myqcircuitr}
\end{equation}
This identity applies here because this is the last operation on qubits $q_{d+2nk}$ and $q_{d+2nk+1}$ before they are measured and post-selected on $\ket{0}$.

With these explicit circuits for \polyR{} and for the \cnot{} gates activating each block-encoding of $\Ham{}$, the full 2D circuit for encoding a polynomial proceeds as follows:
\begin{enumerate}
    \item Apply \polyR{} on qubits $q_{0a}$ through $q_d$ (leftmost column in \cref{fig: 2D circuit poly}) following \cref{circuit:poly_2d_general}.
    \item Apply the ``distance-2 \cnot{}'' gate, \cref{eq:d2_cnot}, on each pair $q_k$, $q_{d+2nk+1}$, for controlled activation of each $\PR{}$. %of each power of each $\PR{}$.
    \item Apply $\PRmod{}$ on each pair of ancilla registers: $(i_0,j_0),\cdots,(i_{d-1},j_{d-1})$.
    \item Successively apply $d$ \select{} oracles, \cref{fig:2d select}, from top to bottom.  First we apply it on registers $(\varphi, i_0, j_0)$, and in the process the $\varphi$ register is swapped with the ancillae, so that the order from top to bottom reads $i_0$, $j_0$, $\varphi$, $i_1$, $\cdots$.  Then repeat on $(\varphi,i_1,j_1)$, after which the data register $\varphi$ ends up below $j_1$.  Continuing in this fashion, after $d$ applications we now have the data register in the lowest row, and all ancilla registers shifted up by 1.
    \item In parallel with the previous step, we also apply SWAPs to move all the ancillae in the leftmost column, used for \polyR{} and \polyL{} up by 1.
    \item Apply $\PLdagmod{}$ on each pair of ancilla registers (now shifted up by 1 row relative to where we applied $\PRmod{}$).
    \item Apply the $d$ ``inverse distance-2 \cnot{}'' gates, \cref{eq:d2_cnot_pl}, shifted up by one row compared to point $2$.
    \item Apply {\upshape$\polyL^\dag$} in analogous form to \cref{circuit:poly_2d_general}, but with all gates shifted up by one qubit in the 2D layout. $q_{0a}$ is still unused, but is in a different physical location.
    \item Measure all ancilla registers $i_0,\cdots, j_{d-1}$ and the $2d$ qubits of the leftmost column and post-select on outcome $\ket{0}$ on each, then the output state with the polynomial applied is stored in the $\varphi$ register, now at the bottom row of the grid.
\end{enumerate}
%%% FIGURE %%%
\begin{figure}[hbtp]
    \centering
    \captionsetup[subfigure]{labelformat=simple}
    \renewcommand{\thesubfigure}{\arabic{subfigure}.}
    Grey = sequential \qquad \& \qquad
    Multiple colors = parallel
    \subfloat[\polyR{}]{%
        \mypolyIVxVIbefore[{%
            \fill[mygrey!30!white] (0,-8) rectangle +(1,8);
        }]{}
    }
    \subfloat[``distance-2 \cnots{}'']{%
        ~\qquad\mypolyIVxVIbefore[{%
            \fill[myred!30!white]    (0,-2) rectangle +(3,1);
            \fill[myblue!30!white]   (0,-4) rectangle +(3,1);
            \fill[myorange!30!white] (0,-6) rectangle +(3,1);
            \fill[mypurple!30!white] (0,-8) rectangle +(3,1);
        }]{%
            \myCNOT1012
            \myCNOT3032
            \myCNOT5052
            \myCNOT7072
        }\qquad~
    }
    \subfloat[\PRmod{}]{%
        \mypolyIVxVIbefore[{%
            \fill[myred!30!white]    (1,-3) rectangle +(6,2);
            \fill[myblue!30!white]   (1,-5) rectangle +(6,2);
            \fill[myorange!30!white] (1,-7) rectangle +(6,2);
            \fill[mypurple!30!white] (1,-9) rectangle +(6,2);
        }]{}
    } \\
    \subfloat[$d$ \select{} oracles + \swaps{}]
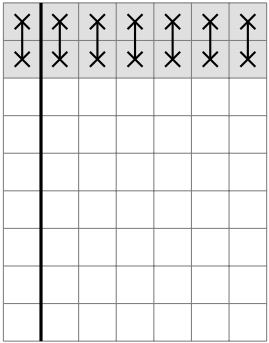
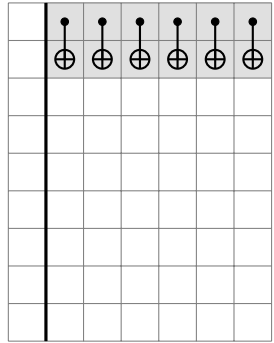
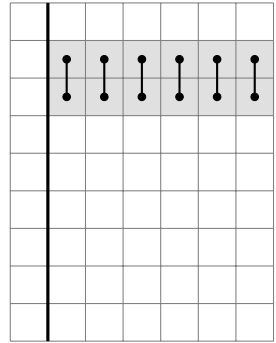
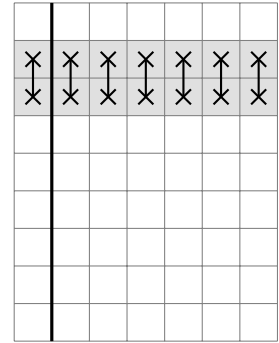
{%
        \begin{minipage}{\textwidth}
         \mypolyIVxVIempty[{%
            \fill[mygrey!30!white] (0,-2) rectangle +(7,2);
        }]{%
            \mySWAP0010
            \mySWAP0111
            \mySWAP0212
            \mySWAP0313
            \mySWAP0414
            \mySWAP0515
            \mySWAP0616
        } \quad
        \mypolyIVxVIempty[{%
            \fill[mygrey!30!white] (1,-2) rectangle +(6,2);
        }]{%
            \myCNOT0111
            \myCNOT0212
            \myCNOT0313
            \myCNOT0414
            \myCNOT0515
            \myCNOT0616
        } \quad
        \mypolyIVxVIempty[{%
            \fill[mygrey!30!white] (1,-3) rectangle +(6,2);
        }]{%
            \myCZ1121
            \myCZ1222
            \myCZ1323
            \myCZ1424
            \myCZ1525
            \myCZ1626
        } \quad
        \mypolyIVxVIempty[{%
            \fill[mygrey!30!white] (0,-3) rectangle +(7,2);
        }]{%
            \mySWAP1020
            \mySWAP1121
            \mySWAP1222
            \mySWAP1323
            \mySWAP1424
            \mySWAP1525
            \mySWAP1626
        } \\[1em]
        $\cdots$ \\[1.5em]
        \mypolyIVxVIempty[{%
            \fill[mygrey!30!white] (0,-8) rectangle +(7,2);
        }]{%
            \mySWAP6070
            \mySWAP6171
            \mySWAP6272
            \mySWAP6373
            \mySWAP6474
            \mySWAP6575
            \mySWAP6676
        } \quad
        \mypolyIVxVIempty[{%
            \fill[mygrey!30!white] (1,-8) rectangle +(6,2);
        }]{%
            \myCNOT6171
            \myCNOT6272
            \myCNOT6373
            \myCNOT6474
            \myCNOT6575
            \myCNOT6676
        } \quad
        \mypolyIVxVIempty[{%
            \fill[mygrey!30!white] (1,-9) rectangle +(6,2);
        }]{%
            \myCZ7181
            \myCZ7282
            \myCZ7383
            \myCZ7484
            \myCZ7585
            \myCZ7686
        } \quad
        \mypolyIVxVIempty[{%
            \fill[mygrey!30!white] (1,-9) rectangle +(6,2);
        }]{%
            \mySWAP7181
            \mySWAP7282
            \mySWAP7383
            \mySWAP7484
            \mySWAP7585
            \mySWAP7686
        } \\
        \end{minipage}
    } \\
    \stepcounter{subfigure}
    \subfloat[\PLdagmod{}]{%
        ~\quad\mypolyIVxVIafter[{%
            \fill[myred!30!white]    (1,-2) rectangle +(6,2);
            \fill[myblue!30!white]   (1,-4) rectangle +(6,2);
            \fill[myorange!30!white] (1,-6) rectangle +(6,2);
            \fill[mypurple!30!white] (1,-8) rectangle +(6,2);
        }]{}\quad~
    }
    \subfloat[``inverse distance-2 \cnots{}'']{%
        ~\quad\mypolyIVxVIafter[{%
            \fill[myred!30!white]    (0,-1) rectangle +(3,1);
            \fill[myblue!30!white]   (0,-3) rectangle +(3,1);
            \fill[myorange!30!white] (0,-5) rectangle +(3,1);
            \fill[mypurple!30!white] (0,-7) rectangle +(3,1);
        }]{%
            \myCNOT0002
            \myCNOT2022
            \myCNOT4042
            \myCNOT6062
        }\quad~
    }
    \subfloat[$\polyL^\dag$]{%
        ~\quad\mypolyIVxVIafter[{%
            \fill[mygrey!30!white] (0,-8) rectangle +(1,8);
        }]{}\quad~
    }
    \subfloat[Measure and post-select]{%
        ~\quad\mypolyIVxVIafter[{%
            \fill[mygrey!30!white] (0,-8) rectangle +(7,8);
        }]{}\quad~
    }
    \caption{$2$D circuit mapping for the matrix polynomial transformation algorithm based on \foxlcu{} for $n=6$ and $d=4$. The one-dimensional version is shown in \cref{fig:poly_pr_ab}, where $\PRc$ and $\PLc$ reduce to a single $\Xp$ gate on the second qubit of each $\ket{i_k}$ register.}
    \label{fig:placeholder}
\end{figure}

%%%%%%%%%%%%%%%%%%%%%%%%%%%%%%%%%
\section{Matrix polynomial full circuit for the Ising Hamiltonian for $n = 3$ and $d = 3$}
\label{sec:full_circuit_ising}
%%%%%%%%%%%%%%%%%%%%%%%%%%%%%%%%%

\begin{center}
\rotatebox{90}{\scalebox{0.9}{
$
\begin{myqcircuitc}{0.33}
&&&&&&\lstick{\ket{0}}	&	\qw &  \gate{R_y}	&	\ctrl{1}	&	\qw	&	\gate{P}	&	\ctrl{3}	&	\qw	&	\qw	&	\qw	&	\qw	&	\qw	&	\qw	&	\qw	&	\qw	&	\qw	&	\qw	&	\qw	&	\qw	&	\qw	&	\qw	&	\qw	&	\qw	&	\qw	&	\qw	&	\qw	&	\qw	&	\qw	&	\qw	&	\qw	&	\qw	&	\qw	&	\qw	&	\qw	&	\qw	&	\qw	&	\qw	&	\qw	&	\qw	&	\qw	&	\qw	&	\qw	&	\qw	&	\qw	&	\qw	&	\ctrl{3}	&	\qw	&	\ctrl{1}	&	\gate{R_y}	&	\qw	&	\meter{} & \cw & 0 \\
&&&&&&\lstick{\ket{0}}	&	\qw &  \qw	&	\gate{R_y}	&	\ctrl{1}	&	\gate{P}	&	\qw	&	\ctrl{8}	&	\qw	&	\qw	&	\qw	&	\qw	&	\qw	&	\qw	&	\qw	&	\qw	&	\qw	&	\qw	&	\qw	&	\qw	&	\qw	&	\qw	&	\qw	&	\qw	&	\qw	&	\qw	&	\qw	&	\qw	&	\qw	&	\qw	&	\qw	&	\qw	&	\qw	&	\qw	&	\qw	&	\qw	&	\qw	&	\qw	&	\qw	&	\qw	&	\qw	&	\qw	&	\qw	&	\qw	&	\ctrl{8}	&	\qw	&	\ctrl{1}	&	\gate{R_y}	&	\qw	&	\qw	&	\meter{} & \cw & 0 \\
&&&&&&\lstick{\ket{0}}	&	\qw &  \qw	&	\qw	&	\gate{R_y}	&	\gate{P}	&	\qw	&	\qw	&	\ctrl{13}	&	\qw	&	\qw	&	\qw	&	\qw	&	\qw	&	\qw	&	\qw	&	\qw	&	\qw	&	\qw	&	\qw	&	\qw	&	\qw	&	\qw	&	\qw	&	\qw	&	\qw	&	\qw	&	\qw	&	\qw	&	\qw	&	\qw	&	\qw	&	\qw	&	\qw	&	\qw	&	\qw	&	\qw	&	\qw	&	\qw	&	\qw	&	\qw	&	\qw	&	\qw	&	\ctrl{13}	&	\qw	&	\qw	&	\gate{R_y}	&	\qw	&	\qw	&	\qw	&	\meter{} & \cw & 0	\\
&&&&&&\lstick{\ket{0}}	&	\qw &  \qw	&	\qw	&	\qw	&	\qw	&	\targ	&	\qw	&	\qw	&	\ctrl{4}	&	\targ	&	\ctrl{1}	&	\targ	&	\qw	&	\qw	&	\qw	&	\qw	&	\ctrl{18}	&	\qw	&	\qw	&	\qw	&	\qw	&	\qw	&	\qw	&	\qw	&	\qw	&	\qw	&	\qw	&	\qw	&	\qw	&	\qw	&	\qw	&	\qw	&	\qw	&	\qw	&	\qw	&	\qw	&	\qw	&	\qw	&	\targ	&	\ctrl{1}	&	\targ	&	\ctrl{4}	&	\qw	&	\qw	&	\targ	&	\qw	&	\qw	&	\qw	&	\qw	& \meter{} & \cw & 0\\
&&&&&&\lstick{\ket{0}}	&	\qw &  \qw	&	\qw	&	\qw	&	\qw	&	\qw	&	\qw	&	\qw	&	\qw	&	\qw	&	\gate{R_y}	&	\ctrl{-1}	&	\ctrl{1}	&	\targ	&	\qw	&	\qw	&	\qw	&	\ctrl{18}	&	\qw	&	\qw	&	\qw	&	\qw	&	\qw	&	\qw	&	\qw	&	\qw	&	\qw	&	\qw	&	\qw	&	\qw	&	\qw	&	\qw	&	\qw	&	\qw	&	\qw	&	\qw	&	\targ	&	\ctrl{1}	&	\ctrl{-1}	&	\gate{R_y}	&	\qw	&	\qw	&	\qw	&	\qw	&	\qw	&	\qw	&	\qw	&	\qw	&	\qw	&	\meter{} & \cw & 0	\\
&&&&&&\lstick{\ket{0}}	&	\qw &  \qw	&	\qw	&	\qw	&	\qw	&	\qw	&	\qw	&	\qw	&	\qw	&	\qw	&	\qw	&	\qw	&	\gate{R_y}	&	\ctrl{-1}	&	\qw	&	\qw	&	\qw	&	\qw	&	\ctrl{18}	&	\push{\rule{0.8em}{0.4pt}}\qw	&	\qw	&	\qw	&	\qw	&	\qw	&	\qw	&	\qw	&	\qw	&	\qw	&	\qw	&	\qw	&	\qw	&	\qw	&	\qw	&	\qw	&	\qw	&	\qw	&	\ctrl{-1}	&	\gate{R_y}	&	\qw	&	\qw	&	\qw	&	\qw	&	\qw	&	\qw	&	\qw	&	\qw	&	\qw	&	\qw	&	\qw	&	\meter{} & \cw & 0	\\
&&&&&&\lstick{\ket{0}}	&	\qw &  \qw	&	\qw	&	\qw	&	\qw	&	\qw	&	\qw	&	\qw	&	\qw	&	\qw	&	\qw	&	\qw	&		\targ	&\qw	&	\qw	&	\qw	&	\qw	&	\qw	&	\qw	&	\ctrl{15}	&	\push{\rule{0.8em}{0.4pt}}\qw	&	\qw	&	\qw	&	\qw	&	\qw	&	\qw	&	\qw	&	\qw	&	\qw	&	\qw	&	\qw	&	\qw	&	\qw	&	\qw	&\qw	&	\qw	&	\qw	&	\targ	&	\qw	&	\qw	&	\qw	&		\qw	&	\qw	&	\qw	&	\qw	&	\qw	&	\qw	&	\qw	&	\qw	&	\meter{} & \cw & 0	\\
&&&&&&\lstick{\ket{0}}	&	\qw &  \qw	&	\qw	&	\qw	&	\qw	&	\qw	&	\qw	&	\qw	&	\gate{R_y}	&	\ctrl{-4}	&	\ctrl{1}	&	\targ	&	\ctrl{-1}	&	\targ	&\qw	&	\qw	&		\qw	&	\qw	&	\qw	&	\qw	&	\ctrl{15}	&	\push{\rule{0.8em}{0.4pt}}\qw	&	\qw	&	\qw	&	\qw	&	\qw	&	\qw	&	\qw	&	\qw	&	\qw	&	\qw	&	\qw	&	\qw	&	\qw	&	\qw	&	\qw	&	\targ	&	\ctrl{-1}	&	\targ	&	\ctrl{1}	&	\ctrl{-4}	&	\gate{R_y}	&	\qw	&	\qw	&	\qw	&	\qw	&	\qw	&	\qw	&	\qw	&	\meter{} & \cw & 0\\
&&&&&&\lstick{\ket{0}}	&	\qw &  \qw	&	\qw	&	\qw	&	\qw	&	\qw	&	\qw	&	\qw	&	\qw	&	\qw	&	\gate{R_y}	&	\ctrl{-1}	&	\qw	&	\ctrl{-1}	&	\qw	&	\qw	&	\qw	&	\qw	&	\qw	&	\qw	&	\qw	&	\ctrl{15}	&	\qw	&	\qw	&	\qw	&	\qw	&	\qw	&	\qw	&	\qw	&	\qw	&	\qw	&	\qw	&	\qw	&\qw	&	\qw	&	\qw	&	\ctrl{-1}	&	\qw	&	\ctrl{-1}	&	\gate{R_y}	&		\qw	&	\qw	&	\qw	&	\qw	&	\qw	&	\qw	&	\qw	&	\qw	&	\qw	&	\meter{} & \cw & 0	\\
&&&&&&\lstick{\ket{0}}	&	\qw &  \qw	&	\qw	&	\qw	&	\qw	&	\qw	&	\targ	&	\qw	&	\ctrl{4}	&	\targ	&	\ctrl{1}	&	\targ	&	\qw	&	\qw	&	\qw	&	\qw	&	\qw	&	\qw	&	\qw	&	\qw	&	\qw	&	\qw	&	\ctrl{12}	&	\qw	&	\qw	&	\qw	&	\qw	&	\qw	&	\qw	&	\qw	&	\qw	&	\qw	&	\qw	&	\qw	&	\qw	&	\qw	&	\qw	&	\qw	&	\targ	&	\ctrl{1}	&	\targ	&	\ctrl{4}	&	\qw	&	\targ	&	\qw	&	\qw	&	\qw	&	\qw	&	\qw	&	\meter{} & \cw & 0	\\
&&&&&&\lstick{\ket{0}}	&	\qw &  \qw	&	\qw	&	\qw	&	\qw	&	\qw	&	\qw	&	\qw	&	\qw	&	\qw	&	\gate{R_y}	&	\ctrl{-1}	&	\ctrl{1}	&	\targ	&	\qw	&	\qw	&	\qw	&	\qw	&	\qw	&	\qw	&	\qw	&	\qw	&	\qw	&	\ctrl{12}	&	\qw	&	\qw	&	\qw	&	\qw	&	\qw	&	\qw	&	\qw	&	\qw	&	\qw	&	\qw	&	\qw	&	\qw	&	\targ	&	\ctrl{1}	&	\ctrl{-1}	&	\gate{R_y}	&	\qw	&	\qw	&	\qw	&	\qw	&	\qw	&	\qw	&	\qw	&	\qw	&	\qw	&	\meter{} & \cw & 0	\\
&&&&&&\lstick{\ket{0}}	&	\qw &  \qw	&	\qw	&	\qw	&	\qw	&	\qw	&	\qw	&	\qw	&	\qw	&	\qw	&	\qw	&	\qw	&	\gate{R_y}	&	\ctrl{-1}	&	\qw	&	\qw	&	\qw	&	\qw	&	\qw	&	\qw	&	\qw	&	\qw	&	\qw	&	\qw	&	\ctrl{12}	&	\push{\rule{0.8em}{0.4pt}}\qw	&	\qw	&	\qw	&	\qw	&	\qw	&	\qw	&	\qw	&	\qw	&	\qw	&	\qw	&	\qw	&	\ctrl{-1}	&	\gate{R_y}	&	\qw	&	\qw	&	\qw	&	\qw	&	\qw	&	\qw	&	\qw	&	\qw	&	\qw	&	\qw	&	\qw	&	\meter{} & \cw & 0	\\
&&&&&&\lstick{\ket{0}}	&	\qw &  \qw	&	\qw	&	\qw	&	\qw	&	\qw	&	\qw	&	\qw	&	\qw	&	\qw	&	\qw	&	\qw	&		\targ	&\qw	&	\qw	&	\qw	&	\qw	&	\qw	&	\qw	&	\qw	&	\qw	&	\qw	&	\qw	&	\qw	&	\qw	&	\ctrl{9}	&	\push{\rule{0.8em}{0.4pt}}\qw	&	\qw	&	\qw	&	\qw	&	\qw	&	\qw	&	\qw	&	\qw	&	\qw	&\qw	&	\qw	&	\targ	&		\qw	&	\qw	&	\qw	&	\qw	&	\qw	&	\qw	&	\qw	&	\qw	&	\qw	&	\qw	&	\qw	&	\meter{} & \cw & 0	\\
&&&&&&\lstick{\ket{0}}	&	\qw &  \qw	&	\qw	&	\qw	&	\qw	&	\qw	&	\qw	&	\qw	&	\gate{R_y}	&	\ctrl{-4}	&	\ctrl{1}	&	\targ	&		\ctrl{-1}	&	\targ	&\qw	&	\qw	&	\qw	&	\qw	&	\qw	&	\qw	&	\qw	&	\qw	&	\qw	&	\qw	&	\qw	&	\qw	&	\ctrl{9}	&	\push{\rule{0.8em}{0.4pt}}\qw	&	\qw	&	\qw	&	\qw	&	\qw	&	\qw	&	\qw&	\qw	&	\qw		&	\targ	&	\ctrl{-1}	&	\targ	&	\ctrl{1}	&	\ctrl{-4}	&	\gate{R_y}	&	\qw	&	\qw	&	\qw	&	\qw	&	\qw	&	\qw	&	\qw	&	\meter{} & \cw & 0	\\
&&&&&&\lstick{\ket{0}}	&	\qw &  \qw	&	\qw	&	\qw	&	\qw	&	\qw	&	\qw	&	\qw	&	\qw	&	\qw	&	\gate{R_y}	&	\ctrl{-1}	&	\qw	&	\ctrl{-1}	&	\qw	&	\qw	&	\qw	&	\qw	&	\qw	&	\qw	&	\qw	&	\qw	&	\qw	&	\qw	&	\qw	&	\qw	&	\qw	&	\ctrl{9}	&	\qw	&\qw	&	\qw	&	\qw	&	\qw	&	\qw	&	\qw	&	\qw	&	\ctrl{-1}	&	\qw	&	\ctrl{-1}	&	\gate{R_y}	&		\qw	&	\qw	&	\qw	&	\qw	&	\qw	&	\qw	&	\qw	&	\qw	&	\qw	&	\meter{} & \cw & 0	\\
&&&&&&\lstick{\ket{0}}	&	\qw &  \qw	&	\qw	&	\qw	&	\qw	&	\qw	&	\qw	&	\targ	&	\ctrl{4}	&	\targ	&	\ctrl{1}	&	\targ	&	\qw	&	\qw	&	\qw	&	\qw	&	\qw	&	\qw	&	\qw	&	\qw	&	\qw	&	\qw	&	\qw	&	\qw	&	\qw	&	\qw	&	\qw	&	\qw	&	\ctrl{6}	&	\qw	&	\qw	&	\qw	&	\qw	&	\qw	&	\qw	&	\qw	&	\qw	&	\qw	&	\targ	&	\ctrl{1}	&	\targ	&	\ctrl{4}	&	\targ	&	\qw	&	\qw	&	\qw	&	\qw	&	\qw	&	\qw	&	\meter{} & \cw & 0	\\
&&&&&&\lstick{\ket{0}}	&	\qw &  \qw	&	\qw	&	\qw	&	\qw	&	\qw	&	\qw	&	\qw	&	\qw	&	\qw	&	\gate{R_y}	&	\ctrl{-1}	&	\ctrl{1}	&	\targ	&	\qw	&	\qw	&	\qw	&	\qw	&	\qw	&	\qw	&	\qw	&	\qw	&	\qw	&	\qw	&	\qw	&	\qw	&	\qw	&	\qw	&	\qw	&	\ctrl{6}	&	\qw	&	\qw	&	\qw	&	\qw	&	\qw	&	\qw	&	\targ	&	\ctrl{1}	&	\ctrl{-1}	&	\gate{R_y}	&	\qw	&	\qw	&	\qw	&	\qw	&	\qw	&	\qw	&	\qw	&	\qw	&	\qw	&	\meter{} & \cw & 0	\\
&&&&&&\lstick{\ket{0}}	&	\qw &  \qw	&	\qw	&	\qw	&	\qw	&	\qw	&	\qw	&	\qw	&	\qw	&	\qw	&	\qw	&	\qw	&	\gate{R_y}	&	\ctrl{-1}	&	\qw	&	\qw	&	\qw	&	\qw	&	\qw	&	\qw	&	\qw	&	\qw	&	\qw	&	\qw	&	\qw	&	\qw	&	\qw	&	\qw	&	\qw	&	\qw	&	\ctrl{6}	&	\push{\rule{0.8em}{0.4pt}}\qw	&	\qw	&	\qw	&	\qw	&	\qw	&	\ctrl{-1}	&	\gate{R_y}	&	\qw	&	\qw	&	\qw	&	\qw	&	\qw	&	\qw	&	\qw	&	\qw	&	\qw	&	\qw	&	\qw	&	\meter{} & \cw & 0	\\
&&&&&&\lstick{\ket{0}}	&	\qw &  \qw	&	\qw	&	\qw	&	\qw	&	\qw	&	\qw	&	\qw	&	\qw	&	\qw	&	\qw	&	\qw	&		\targ	&\qw	&	\qw	&	\qw	&	\qw	&	\qw	&	\qw	&	\qw	&	\qw	&	\qw	&	\qw	&	\qw	&	\qw	&	\qw	&	\qw	&	\qw	&	\qw	&	\qw	&	\qw	&	\ctrl{3}	&	\push{\rule{0.8em}{0.4pt}}\qw	&	\qw	&	\qw&	\qw	&	\qw		&	\targ	&	\qw	&	\qw	&	\qw	&	\qw	&	\qw	&	\qw	&	\qw	&	\qw	&	\qw	&	\qw	&	\qw	&	\meter{} & \cw & 0	\\
&&&&&&\lstick{\ket{0}}	&	\qw &  \qw	&	\qw	&	\qw	&	\qw	&	\qw	&	\qw	&	\qw	&	\gate{R_y}	&	\ctrl{-4}	&	\ctrl{1}	&	\targ	&		\ctrl{-1}	&	\targ	&\qw	&	\qw	&	\qw	&	\qw	&	\qw	&	\qw	&	\qw	&	\qw	&	\qw	&	\qw	&	\qw	&	\qw	&	\qw	&	\qw	&	\qw	&	\qw	&	\qw	&	\qw	&	\ctrl{3}	&	\push{\rule{0.8em}{0.4pt}}\qw	&\qw	&	\qw	&	\targ	&	\ctrl{-1}	&	\targ	&	\ctrl{1}	&		\ctrl{-4}	&	\gate{R_y}	&	\qw	&	\qw	&	\qw	&	\qw	&	\qw	&	\qw	&	\qw	&	\meter{} & \cw & 0	\\
&&&&&&\lstick{\ket{0}}	&	\qw &  \qw	&	\qw	&	\qw	&	\qw	&	\qw	&	\qw	&	\qw	&	\qw	&	\qw	&	\gate{R_y}	&	\ctrl{-1}	&		\qw	&	\ctrl{-1}	&	\qw	&	\qw	& \qw	&	\qw	&	\qw	&	\qw	&	\qw	&	\qw	&	\qw	&	\qw	&	\qw	&	\qw	&	\qw	&	\qw	&	\qw	&	\qw	&	\qw	&	\qw	&	\qw	&	\ctrl{3}	&		\qw	&	\qw&\ctrl{-1}	&	\qw	&	\ctrl{-1}	&	\gate{R_y}		&	\qw	&	\qw	&	\qw	&	\qw	&	\qw	&	\qw	&	\qw	&	\qw	&	\qw	&	\meter{} & \cw & 0	\\
&&&&&&	&	\qw &  \qw	&	\qw	&	\qw	&	\qw	&	\qw	&	\qw	&	\qw	&	\qw	&	\qw	&	\qw	&	\qw	&	\qw	&	\qw	&	\qw	&	\qw	&	\targ	&	\qw	&	\qw	&	\ctrl{0}	&	\qw	&	\qw	&	\targ	&	\qw	&	\qw	&	\ctrl{0}	&	\qw	&	\qw	&	\targ	&	\qw	&	\qw	&	\ctrl{0}	&	\qw	&	\qw	&	\qw	&	\qw	&	\qw	&	\qw	&	\qw	&	\qw	&	\qw	&	\qw	&	\qw	&	\qw	&	\qw	&	\qw	&	\qw	&	\qw	&	\qw	&	\qw	\\
&&&&&&	&	\qw &  \qw	&	\qw	&	\qw	&	\qw	&	\qw	&	\qw	&	\qw	&	\qw	&	\qw	&	\qw	&	\qw	&	\qw	&	\qw	&	\qw	&	\qw	&	\qw	&	\targ	&	\qw	&	\qw	&	\ctrl{0}	&	\qw	&	\qw	&	\targ	&	\qw	&	\qw	&	\ctrl{0}	&	\qw	&	\qw	&	\targ	&	\qw	&	\qw	&	\ctrl{0}	&	\qw	&	\qw	&	\qw	&	\qw	&	\qw	&	\qw	&	\qw	&	\qw	&	\qw	&	\qw	&	\qw	&	\qw	&	\qw	&	\qw	&	\qw	&	\qw	&	\qw	\\
&&&&&&	&	\qw &  \qw	&	\qw	&	\qw	&	\qw	&	\qw	&	\qw	&	\qw	&	\qw	&	\qw	&	\qw	&	\qw	&	\qw	&	\qw	&	\qw	&	\qw	&	\qw	&	\qw	&	\targ	&	\qw	&	\qw	&	\ctrl{0}	&	\qw	&	\qw	&	\targ	&	\qw	&	\qw	&	\ctrl{0}	&	\qw	&	\qw	&	\targ	&	\qw	&	\qw	&	\ctrl{0}	&	\qw	&	\qw	&	\qw	&	\qw	&	\qw	&	\qw	&	\qw	&	\qw	&	\qw	&	\qw	&	\qw	&	\qw	&	\qw	&	\qw	&	\qw	&	\qw	\inputgroupv{4}{6}{0.8em}{2.4em}{\ket{i_0}} \inputgroupv{7}{9}{0.8em}{2.4em}{\ket{j_0}} \inputgroupv{10}{12}{0.8em}{2.4em}{\ket{i_1}} \inputgroupv{13}{15}{0.8em}{2.4em}{\ket{j_1}} \inputgroupv{16}{18}{0.8em}{2.4em}{\ket{i_2}} \inputgroupv{19}{21}{0.8em}{2.4em}{\ket{j_2}} \inputgroupv{22}{24}{0.8em}{2.4em}{\ket{\varphi}}\\
\end{myqcircuitc}
$
}}
\end{center}

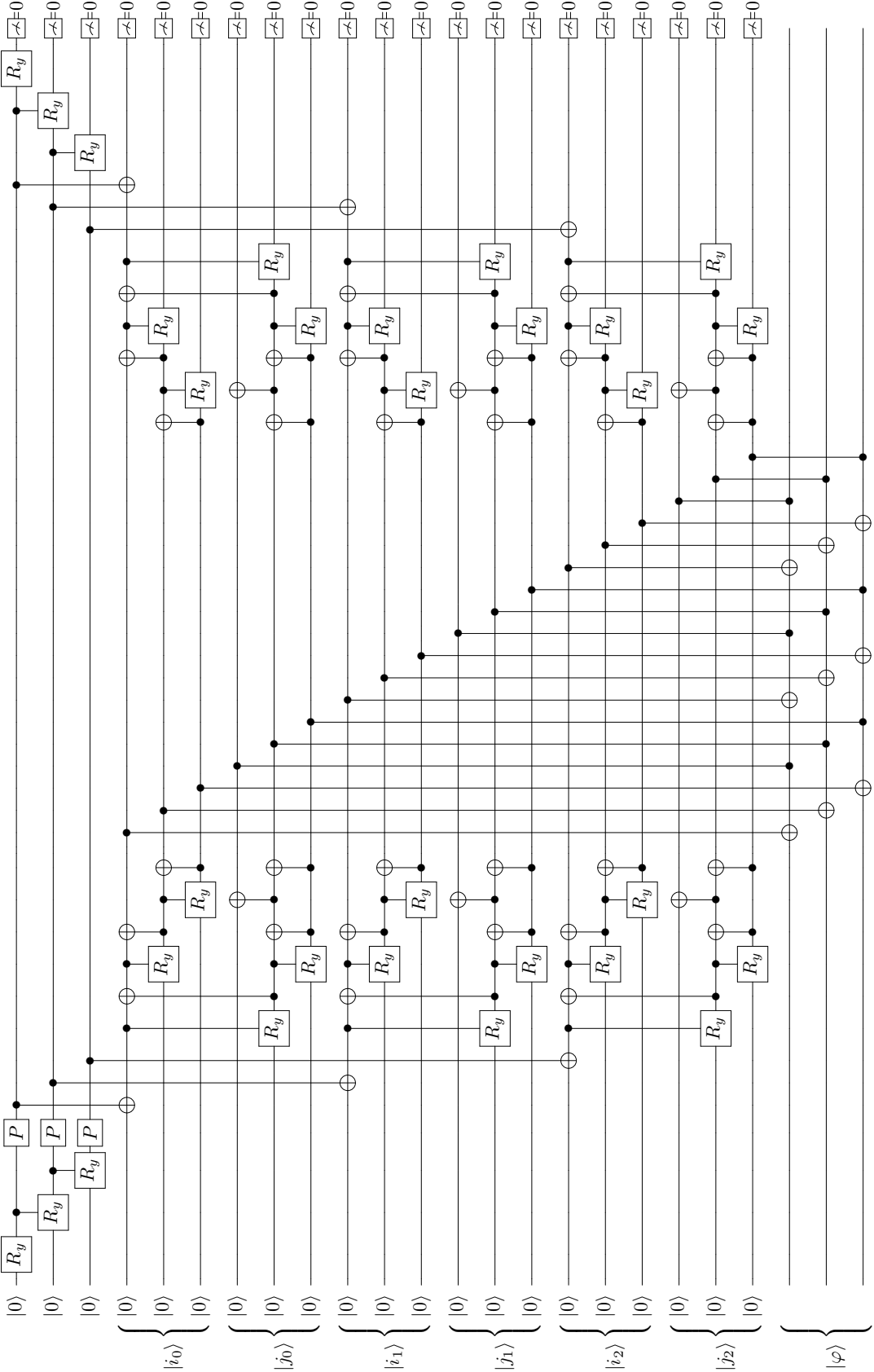
\end{document}